\PassOptionsToPackage{colorlinks=true,linkcolor=cerise!80!black,citecolor=indigo(dye)!70!capri}{hyperref}

\documentclass[conference]{IEEEtran}

\IEEEoverridecommandlockouts

\usepackage{xspace}
\usepackage{amssymb}
\usepackage{amstext}
\usepackage{amsmath}
\usepackage{bbm}
\usepackage{amsthm}
\usepackage{float}

\usepackage{thmtools}
\declaretheorem{remark}
\declaretheorem{example}
\def\fakebold#1{\leavevmode\setbox0=\hbox{#1}\kern-.025em\copy0 \kern-\wd0
  \kern .05em\copy0 \kern-\wd0
  \kern-.025em\raise.0433em\box0
}
\usepackage{hyperref}
\usepackage{hvfloat}

\usepackage{csquotes}
\usepackage{latexcolors}
\usepackage{nowidow}
\usepackage{tikz}
\usepackage{wrapfig}
\usetikzlibrary{decorations.pathreplacing,hobby,backgrounds,trees,arrows,automata,shapes,positioning,fit,patterns,calc,matrix}
\usetikzlibrary{svg.path}

\usepackage{pgfplots}
\usetikzlibrary{intersections, pgfplots.fillbetween}

\newcommand{\nest}{navigational\xspace}
\newcommand{\Nest}{Navigational\xspace}
\newcommand{\nesth}{\nest hierarchy\xspace}
\newcommand{\nesths}{\nest hierarchies\xspace}

\newcommand{\Nesths}{\Nest hierarchies\xspace}

\newcommand{\veps}{\ensuremath{\varepsilon}\xspace}
\newcommand{\inv}{^{-1}}
\newcommand{\nat}{\ensuremath{\mathbb{N}}\xspace}

\newcommand{\cbeta}{\ensuremath{\Cs_\beta}\xspace}
\newcommand{\cgamma}{\ensuremath{\Cs_\gamma}\xspace}
\newcommand{\frA}{\ensuremath{\mathbb{A}}\xspace}
\newcommand{\frB}{\ensuremath{\mathbb{B}}\xspace}
\newcommand{\frAd}{\ensuremath{\frA_D}\xspace}
\newcommand{\frAnd}{\ensuremath{\overline{\frA_D^{}}}\xspace}

\newcommand{\Bs}{\ensuremath{\mathcal{B}}\xspace}
\newcommand{\Cs}{\ensuremath{\mathcal{C}}\xspace}
\newcommand{\Ds}{\ensuremath{\mathcal{D}}\xspace}

\newcommand{\Gs}{\ensuremath{\mathcal{G}}\xspace}

\newcommand{\Is}{\ensuremath{\mathcal{I}}\xspace}

\newcommand{\Ps}{\ensuremath{\mathcal{P}}\xspace}

\newcommand{\Hb}{\ensuremath{\mathbf{H}}\xspace}

\newcommand{\Kb}{\ensuremath{\mathbf{K}}\xspace}
\newcommand{\Lb}{\ensuremath{\mathbf{L}}\xspace}

\newcommand{\Ub}{\ensuremath{\mathbf{U}}\xspace}
\newcommand{\Vb}{\ensuremath{\mathbf{V}}\xspace}

\newcommand{\Wsuit}{Well-suited\xspace}

\newcommand{\vari}{prevariety\xspace}

\newcommand{\varis}{prevarieties\xspace}

\newcommand{\pvari}{positive prevariety\xspace}

\newcommand{\pvaris}{positive prevarieties\xspace}

\newcommand{\infsig}[1]{\ensuremath{\mathbb{I}_{#1}}\xspace}
\newcommand{\infsigc}{\infsig{\Cs}}

\newcommand{\dclosp}[1]{\ensuremath{\mathord{\downarrow_{#1}}}\xspace}
\newcommand{\dclosr}{\dclosp{R}}

\newcommand{\imprint}{imprint\xspace}
\newcommand{\imprints}{imprints\xspace}

\newcommand{\tame}{multiplicative\xspace}

\newcommand{\ratms}{rating maps\xspace}
\newcommand{\ratm}{rating map\xspace}

\newcommand{\full}{full\xspace}

\newcommand{\mratm}{multiplicative rating map\xspace}

\newcommand{\Mratms}{Multiplicative rating maps\xspace}

\newcommand{\fratm}{\full rating map\xspace}
\newcommand{\fratms}{\full rating maps\xspace}

\newcommand{\prin}[2]{\ensuremath{\Is[#1](#2)}\xspace}

\newcommand{\opti}[2]{\ensuremath{\Is_{#1}\left[#2\right]}\xspace}
\newcommand{\copti}[1]{\opti{\Cs}{#1}}

\newcommand{\popti}[3]{\ensuremath{\Ps_{#1}[#2,#3]}\xspace}

\newcommand{\typ}[2]{\ensuremath{[#1]_{#2}}\xspace}

\tikzset{every state/.style={draw=blue!50!green,very thick,fill=blue!50!green!20}}
\tikzset{statesub/.style={state,minimum size=1.3cm,inner sep=1pt}}
\tikzset{pattstate/.style={state,draw=red!50!yellow,line width=2pt,fill=red!50!yellow!20}}
\tikzset{pdotstate/.style={state,minimum size=0.75cm,inner sep=0.5pt,draw=red!50!yellow,line
		width=2pt,dashed,fill=red!50!yellow!20}}
\tikzset{lstate/.style={state,minimum size=0.65cm,inner sep=0.5pt}}
\tikzstyle{trans}=[shorten >= 1pt,thick,->]

\makeatletter
\tikzstyle{initial by arrow}=   [after node path=
{
	{
		[to path=
		{
			[->,double=none,shorten >= 1pt,thick,every initial by arrow]
			([shift=(\tikz@initial@angle:\tikz@initial@distance)]\tikztostart.\tikz@initial@angle)
			node [shape=coordinate,anchor=\tikz@initial@anchor,draw] {\tikz@initial@text}
			-- (\tikztostart)}]
		edge ()
	}
}]
\makeatother

\makeatletter
\tikzstyle{accepting by arrow}=   [after node path=
{
	{
		[to path=
		{
			[->,double=none,shorten >= 1pt,thick,every accepting by arrow]
			(\tikztostart) --
			([shift=(\tikz@accepting@angle:\tikz@accepting@distance)]\tikztostart.\tikz@accepting@angle)
			node [shape=coordinate,anchor=\tikz@accepting@anchor,draw] {\tikz@accepting@text}
		}
		]
		edge ()
	}
}]
\makeatother

\newcommand{\fo}{\ensuremath{\textup{FO}}\xspace}
\newcommand{\fow}{\mbox{\ensuremath{\fo(<)}}\xspace}

\newcommand{\fod}{\ensuremath{\fo^2}\xspace}

\newcommand{\fodw}{\ensuremath{\fod(<)}\xspace}
\newcommand{\fodws}{\ensuremath{\fod(<,+1)}\xspace}

\newcommand{\dew}[1]{\ensuremath{\Delta_{#1}(<)}\xspace}

\newcommand{\dewd}{\dew{2}}

\newcommand{\dewt}{\dew{3}}

\newcommand{\bsc}[1]{\ensuremath{\Bs\Sigma_{#1}}\xspace}

\newcommand{\at}{\ensuremath{\textup{AT}}\xspace}

\newcommand{\grp}{\ensuremath{\textup{GR}}\xspace}

\newcommand{\abg}{\ensuremath{\textup{AMT}}\xspace}

\newcommand{\md}{\ensuremath{\textup{MOD}}\xspace}

\newcommand{\su}{\ensuremath{\textup{SU}}\xspace}

\newcommand{\sfr}{\ensuremath{\textup{SF}}\xspace}

\newcommand{\lt}{\ensuremath{\textup{LT}}\xspace}

\newcommand{\stzer}{\textup{ST}\xspace}
\newcommand{\dotzer}{\textup{DD}\xspace}

\newcommand{\cocl}[1]{\ensuremath{\mathit{co\textup{-}}#1}\xspace}
\newcommand{\bool}[1]{\ensuremath{\textup{Bool}(#1)}\xspace}
\newcommand{\pol}[1]{\ensuremath{\textup{Pol}(#1)}\xspace}
\newcommand{\bpol}[1]{\ensuremath{\textup{BPol}(#1)}\xspace}
\newcommand{\ipol}[1]{\ensuremath{\textup{IPol}(#1)}\xspace}

\newcommand{\copol}[1]{\ensuremath{\cocl{\textup{Pol}(#1)}}\xspace}
\newcommand{\sfp}[1]{\ensuremath{\mathit{\textup{SF}}(#1)}\xspace}

\newcommand{\bpolo}{\ensuremath{\textup{BPol}}\xspace}

\newcommand{\polp}[2]{\ensuremath{\textup{Pol}_{#1}(#2)}\xspace}

\newcommand{\bpolp}[2]{\ensuremath{\textup{BPol}_{#1}(#2)}\xspace}
\newcommand{\bpoln}[1]{\bpolp{n}{#1}}

\newcommand{\ipolp}[2]{\ensuremath{\textup{IPol}_{#1}(#2)}\xspace}

\newcommand{\tls}{\ensuremath{\textup{TL}}\xspace}
\newcommand{\tlxs}{\ensuremath{\textup{TLX}}\xspace}
\newcommand{\tla}[1]{\ensuremath{\tls[#1]}\xspace}

\newcommand{\tlc}[1]{\ensuremath{\tls(#1)}\xspace}

\newcommand{\tlhp}[2]{\ensuremath{\textup{TL}_{#1}(#2)}\xspace}
\newcommand{\tlhn}[1]{\tlhp{n}{#1}}

\newcommand{\ltl}{\ensuremath{\textup{LTL}}\xspace}

\newcommand{\finally}[1]{\ensuremath{\textup{F}~#1}\xspace}
\newcommand{\nex}[1]{\ensuremath{\textup{X}~#1}\xspace}

\newcommand{\finallym}[1]{\ensuremath{\textup{P}~#1}\xspace}

\newcommand{\finallyop}{\ensuremath{\textup{F}}\xspace}
\newcommand{\finallymop}{\ensuremath{\textup{P}}\xspace}
\newcommand{\nexop}{\ensuremath{\textup{X}}\xspace}
\newcommand{\nexmop}{\ensuremath{\textup{Y}}\xspace}

\newcommand{\nexi}[2]{\ensuremath{\textup{X}^{#2}~#1}\xspace}
\newcommand{\nexmi}[2]{\ensuremath{\textup{Y}^{#2}~#1}\xspace}

\newcommand{\finallyp}[2]{\ensuremath{\textup{F}_{#1}~#2}\xspace}
\newcommand{\finallymp}[2]{\ensuremath{\textup{P}_{#1}~#2}\xspace}
\newcommand{\finallyl}[1]{\finallyp{L}{#1}}
\newcommand{\finallyml}[1]{\finallymp{L}{#1}}

\newcommand{\poschar}{\textup{{Pos}}}

\newcommand{\pos}[1]{\ensuremath{\poschar(#1)\xspace}}

\newcommand{\infix}[3]{\ensuremath{#1(#2,#3)}\xspace}

\newcommand{\wpos}[2]{\ensuremath{#1[#2]}\xspace}

\newcommand{\etaat}{\ensuremath{\eta_\at}\xspace}

\newcommand{\atyp}[1]{\typ{#1}{\alpha}}
\newcommand{\etyp}[1]{\typ{#1}{\eta}}

\newcommand{\btyp}[1]{\typ{#1}{\beta}}

\theoremstyle{plain}
\newtheorem{theorem}{Theorem}

\newtheorem{fact}[theorem]{Fact}

\begin{document}

\title{\Nesths of regular languages}

% \author{\IEEEauthorblockN{Anonymous author(s)
% }
% }
\author{
  \IEEEauthorblockN{Thomas Place}
  \IEEEauthorblockA{LaBRI, Univ. Bordeaux, France}
  \vspace*{-6ex}
\and
  \IEEEauthorblockN{Marc Zeitoun}
  \IEEEauthorblockA{LaBRI, Univ. Bordeaux, France}
  \vspace*{-6ex}
}
\maketitle

\begin{abstract}
  We study a celebrated class of regular languages: star-free languages. A long-standing goal is to classify them by the complexity of their descriptions. The most influential research effort involves \emph{concatenation hierarchies}, which measure alternations between ``complement'' and ``union plus concatenation''.  

  We explore alternative hierarchies that also stratify star-free languages (and extensions thereof). They are built with an operator $\Cs\hspace*{-.4ex}\mapsto\hspace*{-.4ex}\tlc\Cs$. It takes  a class of languages~\Cs as input, and produces a larger one \tlc{\Cs},  consisting of all languages definable in a variant of unary temporal logic, where the future/past modalities depend on~\Cs. Level~$n$ in the \emph{\nesth of basis~\Cs} is then constructed by applying this operator~$n$ times~to~\Cs. 

  As bases \Gs, we focus on \emph{group languages} and natural extensions thereof, denoted~$\Gs^+$. We prove that the \nesths of bases~\Gs and $\Gs^+$ are strictly intertwined and we conduct a thorough investigation of their relationships with their concatenation hierarchy counterparts. We also look at two standard problems on classes of languages: membership (decide if a language is in the class) and separation (decide, for two languages $L_1,L_2$, if there is a language $K$ in the class with $L_1 \subseteq K$ and $L_2 \cap K = \emptyset$). We prove that if separation is decidable for \Gs, then so is membership for level \emph{two} in the \nesths of bases \Gs and $\Gs^+$.

  We take a closer look at the trivial class $\stzer=\{\emptyset,A^*\}$. For the bases $\stzer$ and $\stzer^+$, the levels \emph{one}  are~the standard variants of unary temporal logic. The levels \emph{two} correspond to variants of two-variable logic, investigated recently by Krebs, Lodaya, Pandya and Straubing. We solve one of their conjectures. We also prove that for these two bases, level \emph{two} has decidable \emph{separation}. Combined with earlier results on the operator $\Cs\!\mapsto\!\tlc{\Cs}$, this implies that level \emph{three} has decidable membership.

\end{abstract}


\section{Introduction}
\label{sec:intro}
\textbf{Context.}
This paper addresses a fundamental question in automata theory: understanding specific formalisms used to describe regular languages. A prime example is the class of \emph{star-free} languages, which are defined by regular expressions that exclude Kleene star but allow complement. This class is remarkable for its robustness. McNaughton and Papert~\cite{mnpfosf} showed that star-free languages are precisely those expressible in first-order logic \fow, and Kamp~\cite{kltl} proved that \fow matches the expressiveness of linear time temporal logic \ltl. Additionally, Schützenberger~\cite{schutzsf} showed that \emph{membership} in this class is decidable: given a regular language, one can decide whether it is star-free. This led Brzozowski and Cohen~\cite{BrzoDot} to classify star-free languages according to the difficulty of expressing them, resulting in the \emph{dot-depth hierarchy}.

The dot-depth is a \emph{concatenation hierarchy}. Such a hierarchy consists of an increasing sequence of classes of languages, called its \emph{levels}. It is fully determined by its initial level, called its \emph{basis}. Once the basis is fixed, the construction process is  uniform, and made through the  operator \bpolo. For a class \Ds, the class~$\pol\Ds$ consists of all languages $L_0a_1L_1\cdots a_nL_n$, where the $L_i$'s are in $\Ds$ and the $a_i$'s are letters. Then, the class $\bpol\Ds$ is the least Boolean algebra containing~\pol\Ds.

Level \emph{zero} in the concatenation hierarchy of basis \Cs is $\bpolp{0}{\Cs}=\Cs$. Level~$n+1$ is $\bpolp{n+1}{\Cs}=\bpol{\bpoln{\Cs}}$, \emph{i.e.}, it is obtained by applying the ``\bpolo'' operator to~level~$n$. Therefore, the union of all levels, called the \emph{star-free closure} \sfp\Cs of~\Cs, is stratified by such a concatenation hierarchy. The dot-depth starts from the basis $\dotzer=\{\emptyset, \{\varepsilon\}, A^+, A^*\}$, where~$A$ is the alphabet. Brzozowski and Knast~\cite{BroKnaStrict} showed that it is strict: $\bpoln{\dotzer}\subsetneq\bpolp{n+1}{\dotzer}$. Other standard hierarchies include the one of Straubing and Thérien~\cite{StrauConcat,TheConcat} of basis $\stzer=\{\emptyset, A^*\}$, the group hierarchy~\cite{MargolisP85} whose basis \grp consists of languages recognized by a permutation automaton~\cite{permauto}, and the modulo hierarchy~\cite{ChaubardPS06} whose basis $\md\subseteq\grp$ consists of languages for which membership of a word is determined by its length modulo an integer.

Individual levels of concatenation hierarchies are themselves robust: level~$n$ corresponds to the class \bsc{n} of all languages expressible in first-order logic with $n$ blocks of quantifiers $\exists^{*}/\forall^{*}$, under a tailored signature constructed from the basis~\cite{PZ:generic18}. This was first discovered by Thomas~\cite{ThomEqu} for the dot-depth. In this context, group languages are relevant bases, as they lead to natural signatures. For example, the class  \md yields predicates testing the value of a position modulo a fixed integer. Interesting bases are also obtained by extending such classes \Gs into larger ones, denoted~$\Gs^+$, obtained by adding or deleting the empty word to the languages of~\Gs. On the logical side, this adds the successor relation~\cite{pmixed,pzupol2}.

This motivates to study membership for \emph{individual levels} of concatenation hierarchies. However, this is challenging. For instance, in the dot-depth and Straubing-Thérien hierarchies, decades of research have shown decidable membership only for levels \emph{one}~\cite{simonthm,knast83}, \emph{two}~\cite{pzjacm19} and \emph{three}~\cite{PZ-Level3}. Most of the recent advances were obtained by introducing new decision problems, such as \emph{\Cs-separation} (given two regular languages, decide if there exists a third one in \Cs containing the first and disjoint from the second). Typically, solving such problems is harder than solving \Cs-membership. However, they lead to \emph{transfer results}. For instance, the decidability of \Cs-separation implies the decidability of \pol\Cs-membership~\cite{PZ:generic18}. While a similar result exists for \bpol\Cs~\cite{PZ-Level3}, it requires going beyond separation, and is, above all, significantly more~complex.

\vspace{-0.01cm}

\smallskip\noindent{\bfseries\Nesths.} In this paper, we adopt a different point of view: to get insight on concatenation hierarchies, we study an alternate, yet closely related classification of the star-free closure \sfp\Cs of a class~\Cs. This new kind of hierarchy is called the \emph{\nesth}. It draws inspiration from another renowned type of class of languages: those based on two-variable first-order logic and/or temporal logic~\cite{iwfo2alt,ksfo2alt,kwfo2alt3,evwutl,twfo2,small_fragments08,DartoisP13,betweenlics,betweenconf,between}. Its construction is similar to concatenation hierarchies: one replaces the operator $\Cs\mapsto\bpol\Cs$ by an operator $\Cs\mapsto\tlc\Cs$ based on unary temporal logic. The class \tlc{\Cs} consists of all languages definable in a variant of \emph{unary temporal logic} with future/past~modalities built from~\Cs. It is defined as follows: each language in \Cs gives rise to a future modality $\finallyop_{L}$ and a past one $\finallymop_{L}$ generalizing the classic ones \finallyop and~\finallymop. For instance, $\finallyp{L}{\varphi}$ holds at position~$i$ in a word $w$ if there is a position $j>i$ such that $\varphi$ holds at~$j$ in~$w$, and the infix of $w$ between $i$ and~$j$ belongs to~$L$. The operator $\Cs\mapsto\tlc\Cs$ seems to have more robust properties than \bpolo. In particular, we have a transfer result: $\tlc{\Cs}$-membership reduces to \Cs-separation~\cite{pztl}.

One denotes level~$n$ in the \nesth of \mbox{basis}~\Cs by \tlhp{n}{\Cs}. The union of all levels remains the star-free closure of~\Cs, so that the \nest and concatenation hierarchies stratify the same class: \sfp\Cs.
Level~one in the \nesth of basis $\stzer=\{\emptyset,A^*\}$ is the standard class of languages definable in \emph{unary temporal logic}~\tls with the \finallyop and \finallymop modalities. Results of~\cite{schul,pwdelta2,twfo2,evwutl} underscore its multiple characterizations: it consists of languages definable in two-variable first-order logic~\fodw, corresponds also to the ``intermediate level'' \dew2, is the class of unambiguous languages and has a nice decidable algebraic characterization. Level~one in the \nesth of basis $\dotzer=\stzer^+$ is also standard: it is the class of languages definable in \emph{unary temporal logic}~\tlxs with the \finallyop, \finallymop, \nexop (``next'') and \nexmop (``yesterday'') modalities. Finally, level~two in the \nesth of basis \stzer corresponds to a class recently investigated in~\cite{between,betweenconf,betweenlics}: \emph{two-variable first-order logic with between predicates}, where the link with temporal logic is also observed and considered.

\vspace{-0.01cm}

\smallskip
\noindent\textbf{{Contributions.}}
We carry out a general study of \nesths. We first briefly compare them with concatenation hierarchies: we show that $\bpoln\Cs\subseteq\tlhn\Cs$ for all $n\in\nat$ for any class~\Cs. Moreover, as announced above, we prove that both kinds of hierarchies classify the star-free closure of \Cs.

The remainder of the paper focuses on specific bases of the form \Gs or $\Gs^+$, where \Gs represents a class of group languages (using the basis $\Gs^+$ instead of \Gs boils down to adding the standard ``next'' and ``yesterday'' modalities). Group languages are those recognized by a complete, deterministic and co-deterministic automaton~\cite{permauto}. Bases made of group languages are natural in our context. Indeed, if the basis includes  no group language other than $\emptyset$ and $A^*$, then each level of the concatenation or \nesth of basis~\Cs is a subclass of \sfr. Thus, extending beyond star-free languages requires bases containing nontrivial group~languages.

We first prove that for \emph{any} class of group languages~\Gs, the \nesths of bases~\Gs and $\Gs^+$ are both strict and that they are strictly intertwined (if the alphabet has at least two letters). As a byproduct, we reprove the same property for concatenation hierarchies. We next sharpen the comparison between both types of hierarchies: we show that if $\Cs=\Gs$ or $\Cs=\Gs^+$, the class~$\tlhp{2}\Cs$ is incomparable with any level above three of the concatenation hierarchy of basis~$\Cs$. The situation is depicted in~\figurename~\ref{fig:hiera}.

\begin{remark}
  Separation is a key tool here: to prove that a language $L$ is not in level $n$ of a hierarchy, it suffices to exhibit a disjoint language $K$ and to show that $L$ and $K$ cannot be separated by this level. This approach significantly simplifies the inductive construction of languages outside each level. \end{remark}

\vspace*{-2ex}

\begin{figure}[!h]
  \centering
  \includegraphics[width=8.5cm]{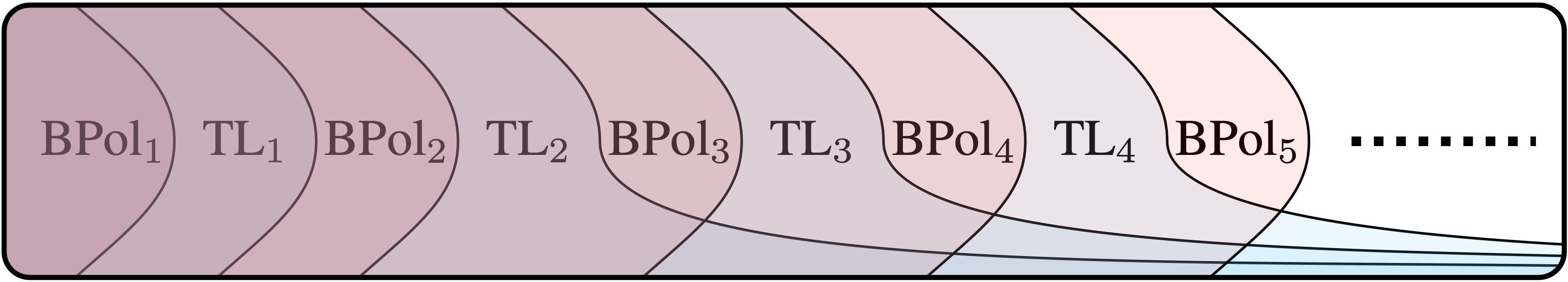}
  \caption{Comparison of~the two hierarchies (bases \Gs or~$\Gs^{+}$)}
  \label{fig:hiera}
  \vspace{-2ex}
\end{figure}

\vspace{-1ex}
These results generalize those of~\cite{betweenlics,between}, which already implied that \tlhp{2}{\stzer} is incomparable with any level above three in the concatenation hierarchy of basis \stzer. Interestingly, this particular basis requires arbitrarily large alphabets: it is shown in~\cite{betweenlics,between} that $\tlhp{2}{\stzer} \subseteq \bpolp{3}{\stzer}$ when the alphabet is binary. We strengthen this result by showing that $\tlhp{2}{\stzer} \subseteq \bpolp{n+1}{\stzer}$ when the alphabet is of size $n$. Moreover, we prove a surprising result: when the alphabet is \emph{binary}, \tlhp{2}{\stzer} is exactly the class of languages in \polp{3}{\stzer} whose complement also belong to \polp{3}{\stzer} (also written \dewt). This result is reminiscent of a famous theorem of Thérien and Wilke, which states that $\tlhp{1}{\stzer} = \dewd$.

While we need arbitrarily large alphabets when the basis is \stzer, we show that a binary alphabet suffices for bases of the form $\Gs^+$, or when \Gs includes \md. In these cases, the classes $\tlhp{k}\Cs$ and $\bpoln\Cs$ are incomparable if $2\leq k<n$.

Finally, we prove that membership for both $\tlhp{2}{\Gs}$ and $\tlhp{2}{\Gs^+}$ reduces to~\Gs-separation, by combining two known transfer results of~\cite{pzconcagroup,pztl}. This theorem applies for $\Gs=\stzer$, $\md$ and $\grp$ as they have decidable separation~\cite{Ash91,pzgr}.

In the second part of the paper, we focus on the class \tlhp{2}{\stzer}, proving that it has decidable \emph{separation}. This result has several interesting implications. First, \tlhp{2}{\stzer} is notable because it corresponds to \emph{two-variable first-order logic with between predicates}~\cite{between,betweenconf,betweenlics}. Second, by combining this result with the theorem of~\cite{pztl}, we show that level \emph{three} (\tlhp{3}{\stzer}) has decidable membership. Third, using transfer theorems~\cite{pzsucc,StrauVD}, we lift these results to the basis $\dotzer = \stzer^+$, so that \tlhp{2}{\dotzer} has decidable separation and  \tlhp{3}{\dotzer} has decidable membership. Fourth, the separation algorithm itself is significant as it solves a more general problem: \tlhp{2}{\stzer}-\emph{covering}, which generalizes separation to arbitrarily many input languages~\cite{pzcovering2}. This generalization is \emph{necessary}, as no specialized separation algorithm exists. Indeed, the procedure relies on a fixpoint computation and carrying it out requires more information that what is needed to decide separation alone. Notably, we use a covering algorithm for $\tlc{\stzer^+}$~\cite{pzsucc} as a subprocedure.

\vspace{.2mm}
\noindent \textbf{Organization.} Section~\ref{sec:prelims} introduces the notations and definitions, followed by a recall of concatenation hierarchies in Section~\ref{sec:concat}. \Nesths are defined in Section~\ref{sec:utl}. Section~\ref{sec:groups} focus on group bases, proving the strictness of their \nesths and comparing them to concatenation hierarchies. The remaining sections address separation and covering: Section~\ref{sec:ratms} presents the framework, and Section~\ref{sec:carav} shows the decidability of covering and separation for $\tlhp2\stzer$.

\section{Preliminaries}
\label{sec:prelims}
We fix a \emph{finite alphabet} $A$. As usual, $A^*$ is the set of all finite words over~$A$, including the empty one \veps. A \emph{language} is a subset of $A^*$. We let $A^+ = A^* \setminus \{\veps\}$. For $u,v \in A^*$, we let $uv$ be the word obtained by concatenating $u$ and $v$. For $K,L \subseteq A^*$, we let \mbox{$KL = \{uv \mid u \in K \text{ and } v \in L\}$}. If $w \in A^*$, we write $|w| \in \nat$ for its length. We view a word \mbox{$w =a_1 \cdots a_{|w|} \in A^*$} as an \emph{ordered set $\pos{w} = \{0,\dots,|w|+1\}$ of $|w|+2$ positions}. Position $i \in \{1,\dots,|w|\}$ carries label $a_i \in A$. We write $\wpos{w}{i} = a_i$. Positions $0$ and $|w|+1$ carry \emph{no label}, and we write $\wpos{w}{0} = min$ and $\wpos{w}{|w|+1} = max$. For $w = a_1\cdots a_{|w|} \in A^*$ and $i,j \in \pos{w}$ such that $i < j$, we write $\infix{w}{i}{j} = a_{i+1} \cdots a_{j-1} \in A^*$. In particular, $\infix{w}{0}{|w|+1} = w$ and $\infix{w}{i}{i+1} = \veps$ for $0\leq i\leq|w|$.

\vspace{-0.02cm}
\smallskip
\noindent
{\bf {Classes}.} A class of languages is a set of languages. A class~\Cs is a \emph{lattice} if $\emptyset\in\Cs$, $A^* \in \Cs$ and~\Cs is closed under union and intersection: for all $K,L \in \Cs$, we have $K \cup L \in \Cs$ and $K \cap L \in \Cs$. A \emph{Boolean algebra} is a lattice \Cs which is also closed under complement: for all $L \in \Cs$, we have $A^* \setminus L \in \Cs$. Finally,  \Cs is \emph{closed under quotients} when for all $L \in \Cs$ and $u \in A^*$, the languages $u\inv L\!=\!\{v\in A^*\mid uv\in L\}$ and $Lu\inv\!=\! \{v\in A^*\mid vu\in L\}$ both belong to \Cs. A \emph{\vari} (resp.~\emph{\pvari}) is a Boolean algebra (resp.~lattice) closed under quotients and containing only \emph{regular languages}. Regular languages are those that can be equivalently defined by finite automata or finite monoids. Here, we use the latter.

\vspace{-0.02cm}

\smallskip
\noindent
{\bf {Monoids.}} A \emph{monoid} is a set $M$  endowed with an associative multiplication $(s,t)\mapsto st$ having an identity element $1_M$, \emph{i.e.}, such that ${1_M} s=s {1_M}=s$ for every~$s \in M$.  A \emph{group} is a monoid $G$ such that for all $g\in G$, there exists $g^{-1}$ such that $gg^{-1}=g^{-1}g = 1_G$. An \emph{idempotent} of a monoid $M$ is an element $e \in M$ such that $ee = e$. We write $E(M) \subseteq M$ for the set of all idempotents in $M$. It is folklore that for every \emph{finite} monoid~$M$, there is a number $\omega(M)$ (written $\omega$ when~$M$ is understood) such that $s^\omega$ is idempotent for all $s \in M$.

Note that $A^*$ is a monoid whose multiplication is concatenation \mbox{(its identity is~\veps)}. A morphism is a map $\alpha: A^* \to M$, where $M$ is a monoid, such that $\alpha(uv)=\alpha(u)\alpha(v)$ and $\alpha(\veps)=1_M$ for all $u,v\in A^*$.  A language $L \subseteq A^*$ is \emph{recognized} by a \mbox{morphism $\alpha:A^*\to M$} if there exists a set $F \subseteq M$ such that $L= \alpha\inv(F)$. Every language~$L$ is recognized by a canonical morphism, which we recall. One associates to $L$ an equivalence relation $\equiv_L$ over $A^*$. Given  $u,v \in A^*$, we let $u \equiv_L v$ if and only if $xuy \in L \Leftrightarrow xvy \in L$ for all $x,y \in A^*$. One can check that $\equiv_L$ is a congruence on $A^*$, \emph{i.e.}, an equivalence relation compatible with concatenation. Thus, the set of equivalence classes $M_L = {A^*}/{\equiv_L}$ is a monoid, called the \emph{syntactic monoid}~of~$L$. The map $\alpha_L : A^* \to M_L$ that maps every word to its equivalence class is a morphism recognizing~$L$, called the \emph{syntactic morphism} of $L$. A language $L$ is regular if and only if $M_L$ is finite: this is the Myhill-Nerode theorem. In this case, one can compute the syntactic morphism $\alpha_L : A^* \to M_L$ from any representation of $L$.

\begin{remark} \label{rem:finmono}
  The only \emph{infinite monoids} that we consider are of the form $A^*$. From now on, we implicitly assume that every other monoid $M,N,\dots$ appearing in the paper is \emph{finite}.
\end{remark}

\noindent
{\bf {The class \at.}} The class \at of \emph{alphabet testable languages} is central in this paper.  It consists of all finite Boolean combinations of languages $A^*aA^*$ where $a \in A$. One can check that \at is a \vari. We associate a morphism to \at. Consider the monoid~$2^A$ (the set of all subalphabets of~$A$), with union as multiplication. We let $\etaat: A^* \to 2^A$ be the morphism defined by $\etaat(a) = \{a\}$. That is, $\etaat(w) = \{a \in A \mid w \in A^*aA^*\}$ for all $w \in A^*$. One can check that \etaat recognizes exactly the languages in \at. We call \etaat the \emph{canonical \at-morphism}.

We shall also consider the class \lt of \emph{locally testable languages}~\cite{bslt73,zalclt72}. It consists of all finite Boolean combinations of languages $A^*wA^*$, $wA^*$ and $A^*w$ where $w \in A^*$ is a word. 

\smallskip
\noindent
{\bf {Decision problems}}. We are interested in three problems that depend on a parameter class \Cs. We use them as tools for analyzing \Cs: intuitively, proving their decidability requires a solid understanding of~\Cs. The simplest one is \emph{\Cs-membership}. Its input is a regular language~$L$. It asks whether $L\in \Cs$.

The next problem, \emph{\Cs-separation}, takes \emph{two} regular languages $L_1,L_2$ as input and asks whether $L_1$ is \emph{\Cs-separable} from $L_2$: does there exist a third language $K \in \Cs$ such that $L_1 \subseteq K$ and $L_2 \cap K = \emptyset$? This generalizes membership: $L \in \Cs$ if and only if it is \Cs-separable from its complement (which is also regular). We shall need to following lemma (proved in Appendix~\ref{app:appdefs}): language concatenation preserves non-separability.

\begin{restatable}{lemma}{nsepconcat} \label{lem:nsepconcat}
	Let $L_1,L_2,H_1,H_2 \subseteq A^*$ and let \Cs be a \pvari. If $L_1$ is \textbf{not} \Cs-separable from $L_2$ and $H_1$ is \textbf{not} \Cs-separable from $H_2$, then $L_1H_1$ is \textbf{not} \Cs-separable from $L_2H_2$.
\end{restatable}

Finally, \Cs-covering is defined in~\cite{pzcovering2}. For a language~$L$, a \emph{\Cs-cover of $L$} is a \emph{finite} set of languages \Kb such that all $K \in \Kb$ belong to \Cs and~$L \subseteq \bigcup_{K \in \Kb} K$. The \emph{\Cs-covering problem} takes as input a pair $(L_1,\Lb_2)$ where \emph{$L_1$ is a regular language} and $\Lb_2$ is a \emph{finite set of regular languages}. It asks whether $(L_1,\Lb_2)$ is \emph{\Cs-coverable}: does there exist a \Cs-cover \Kb of~$L_1$ such that for every $K\in\Kb$, there exists $L \in \Lb_2$ satisfying $K \cap L = \emptyset$? This generalizes separation (if \Cs is closed under union): $L_1$ is \Cs-separable from $L_2$ if and only if $(L_1,\{L_2\})$ is \Cs-coverable.

\smallskip
\noindent
{\bf {The \Cs-pair relation.}} In this paper, we characterize certain classes using a binary relation, which we now define. Let~\Cs be a class and let $\alpha: A^* \to M$ be a morphism. We say that a pair $(s,t) \in M^2$ is a \emph{\Cs-pair} when $\alpha\inv(s)$ is \emph{not} \Cs-separable from $\alpha\inv(t)$. Clearly, \Cs-pairs can be computed from the morphism~$\alpha$ if \Cs-separation is~decidable. Typically, \Cs-pairs are used to characterize a class constructed on top of~\Cs.

\section{Concatenation hierarchies}
\label{sec:concat}
We briefly recall the definition of standard concatenation hierarchies, which is based on certain \emph{operators}.

\smallskip
\noindent
{\bf Operators.} The \emph{polynomial closure} \pol\Cs of a class~\Cs consists of all \emph{finite unions} of \emph{marked products} $L_0a_1L_1 \cdots a_n L_n$ with $a_1,\dots,a_n \in A$ and $L_0,\dots,L_n \in \Cs$. It is known~\cite{arfi87,jep-intersectPOL,PZ:generic18} that if \Cs is \pvari, then so is \pol{\Cs}. However, \pol{\Cs} may \emph{not} be closed under complement. It is therefore combined with two other operators. First, for every class \Ds, we write $\cocl{\Ds} = \{A^* \setminus L \mid L \in \Ds\}$ for the class consisting of all complements of languages in \Ds. Moreover, the \emph{Boolean closure} of a class \Ds is the least Boolean algebra \bool{\Ds} containing~\Ds. We write~\bpol{\Cs} for \bool{\pol{\Cs}}.

\smallskip
\noindent
{\bf Hierarchies.} A \emph{concatenation hierarchy} is built as follows. One starts from an arbitrary class \Cs called its \emph{basis}. We let $\bpolp{0}{\Cs}= \Cs$ and for $n>0$, $\bpolp{n}{\Cs} = \bpol{\bpolp{n-1}{\Cs}}$. Clearly, $\bpolp{n}{\Cs} \subseteq \bpolp{n+1}{\Cs}$ for all $n \in \nat$. Level~$n$ in the concatenation hierarchy of basis \Cs is then $\bpolp{n}{\Cs}$.

It is also standard to consider ``intermediary'' levels. First, we let $\polp{0}{\Cs} = \Cs$ and $\polp{n}{\Cs} = \pol{\bpolp{n-1}{\Cs}}$ if $n \geq 1$. Next, we let $\ipolp{n}{\Cs} = \polp{n}{\Cs} \cap \cocl{\polp{n}{\Cs}}$ for $n \in \nat$ (it is therefore the greatest Boolean algebra contained in \polp{n}{\Cs}).

\begin{figure}[H]
	\centering
	\includegraphics[width=8.5cm]{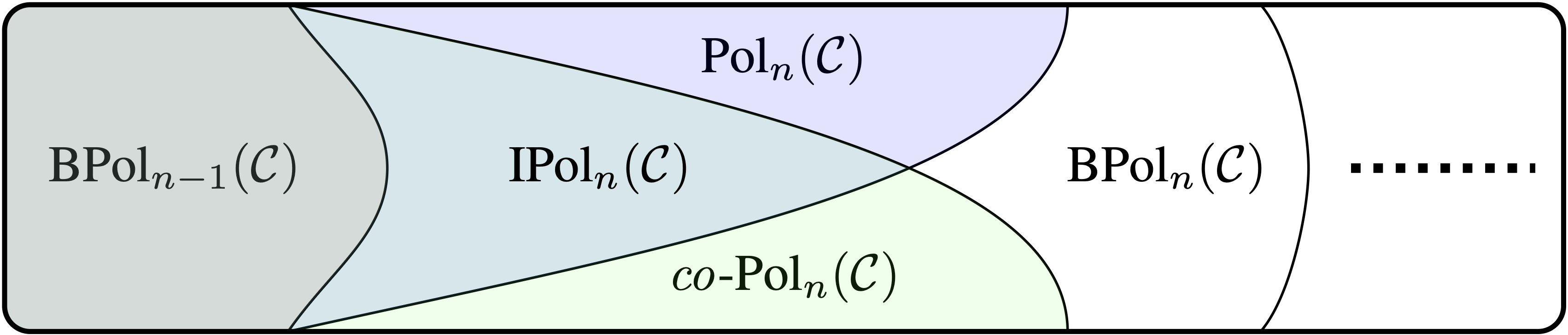}
% \begin{tikzpicture}[xscale=0.78,yscale=0.3, every node/.style={scale=0.8},fill between/on layer=main]

% 		\begin{scope}
% 			\clip[rounded corners] (0,0) rectangle (11,6);
			
% 			\draw[name path=Z] (0,0) to (0,6) {};
			
% 			\draw[name path=A] (2.0,0) to [out=75,in=-90] (2.7,3) to [out=90,in=-75] (2.0,6) ;
% \node at (1.2,3) {\bpolp{n-1}{\Cs}};
% 			\tikzfillbetween[of=Z and A]{orange, opacity=0.11}

% 			\node at (5.5,5.0) {\polp{n}{\Cs}};
% 			\draw[name path=B] (2.0,0) to [out=30,in=-90] (7.5,6) ;

% 			\tikzfillbetween[of=Z and B]{blue, opacity=0.11}
			
% 			\draw[name path=Y] (7.5,0) to (0,0) to (0,6) to (2.0,6);			
% 			\node at (5.5,1.0) {\cocl{\polp{n}{\Cs}}};
% 			\draw[name path=C] (2.0,6) to [out=-30,in=90] (7.5,0) ;
% 			\tikzfillbetween[of=C and Y]{green, opacity=0.11}
			
% 			\node at (4.2,3) {\ipolp{n}{\Cs}};

% 			\node at (8.2,3) {\bpolp{n}{\Cs}};
% 			\draw[name path=E] (8.5,0) to [out=-70,in=-90] (9.2,3) to [out=90,in=70] (8.5,6) ;

% 			\draw[line width=1.5pt,dotted] (9.5,3.0) to (10.8,3.0) ;
% 		\end{scope}
		
% 		\draw[thick,rounded corners] (0,0) rectangle (11,6);
% 	\end{tikzpicture}
	\caption{The four kinds of levels in a concatenation hierarchy.}
	\label{fig:chiera}
	\vspace{-2ex}
\end{figure}

The following standard statement was originally proved in~\cite{arfi91} (see also~\cite{jep-intersectPOL,PZ:generic18} for recent proofs).

\begin{restatable}{proposition}{bpvar} \label{prop:bpvar}
	Let \Cs be a \pvari and $n \in \nat$. Then \polp{n}{\Cs} is a \pvari. Moreover, \bpolp{n}{\Cs}  and \ipolp{n}{\Cs} are \varis.
\end{restatable}

For each class \Cs, we define $\sfp{\Cs} = \bigcup_{n \in \nat} \bpoln{\Cs}$. This union of all levels in the concatenation hierarchy of basis~\Cs is the \emph{star-free closure} of \Cs. This is the least class containing \Cs and closed under Boolean operations and marked product.

The first concatenation hierarchy to be introduced was the dot-depth hierarchy of basis $\dotzer =  \{\emptyset,\{\veps\}, A^+,A^*\}$. It was defined by Brzozowski and Cohen~\cite{BrzoDot} to classify star-free languages~\cite{schutzsf}. Most bases in the literature are \emph{classes of group languages}. These bases will be revisited in Section~\ref{sec:groups}.
\smallskip

We complete the presentation with some useful properties. The purpose of the next two lemmas is to prove that a particular language does \textbf{not} belong to a given level. The first is simple and proved in~\cite[Lemma~33]{PZ:generic18} or \cite[Lemma~4.8]{pzupol2}.

\begin{restatable}{lemma}{polcopol} \label{lem:polcopol}
	If \Cs is a \vari, $\polp{2}{\Cs} = \pol{\copol{\Cs}}$.
\end{restatable}

The second lemma is proved in Appendix~\ref{app:concat} as a corollary of a characterization of \pol{\Cs} given in~\cite[Theorem~54]{PZ:generic18}.

\begin{restatable}{lemma}{polinduc} \label{lem:polinduc}
	Let \Cs be a \pvari. Let $H_1,H_2 \subseteq A^*$ be languages such that $H_1$ is \textbf{not} \Cs-separable from $H_2$. Then, $H_1^*$ is \textbf{not} \pol{\Cs}-separable from $H_1^*H_2H_1^*$.
\end{restatable}

Finally, we will require the following algebraic characterization of \ipolp{2}{\Cs}. It is proved in~\cite[Theorem~6.7]{pzupol2}.

\begin{restatable}{theorem}{capolbp} \label{thm:capolbp}
	Let \Cs be a \pvari.
Let $L$ be a regular language and $\alpha: A^* \to M$ be its syntactic morphism. Then, $L \in  \ipolp{2}{\Cs}$ if and only if $\alpha$ satisfies the following property:
	\begin{equation}\label{eq:capolbp}
		\begin{array}{c}
			\hspace*{-1.5ex}(esete)^{\omega+1} =		(esete)^{\omega} ete	(esete)^{\omega} \\
			\text{\hspace*{-1.5ex}for all \Cs-pairs $(e,s)\! \in\! E(M)\! \times\! M$ for $\alpha$ and all $t \in M$}
		\end{array}
	\end{equation}
\end{restatable}

\section{\Nesths}
\label{sec:utl}
In this section, we introduce the \emph{operator} $\Cs \mapsto \tlc{\Cs}$ and use it to define \nesths. Next, we compare these hierarchies with concatenation hierarchies.

\subsection{Definition}

\noindent
\textbf{The operator.} For each class \Cs, we define a set  \tla{\Cs} of temporal formulas. They are built inductively from the atomic formulas $\top$, $\bot$, $min$, $max$ and ``$a$'' for all letters $a \in A$. All Boolean connectives are allowed: if $\psi_1,\psi_2 \in \tla{\Cs}$, then $(\psi_{1} \vee \psi_{2}) \in \tla{\Cs}$, $(\psi_{1} \wedge \psi_{2})\in \tla{\Cs}$ and $(\neg \psi_1)\in \tla{\Cs}$. There are \emph{two temporal modalities} $\textup{F}_L$ and $\textup{P}_L$ for each $L \in \Cs$: if $\psi \in \tla{\Cs}$, then $(\finallyl{\psi})\in \tla{\Cs}$ and $(\finallyml{\psi})\in \tla{\Cs}$. We omit parentheses when there is no ambiguity.

Evaluating a formula $\varphi \in \tla{\Cs}$ requires a word $w \in A^*$ and a position $i \in \pos{w}$. We define by induction what it means for  \emph{$(w,i)$ to satisfy $\varphi$}, which one denotes by $w,i \models \varphi$:
\begin{itemize}
  \item {\bf Atomic formulas:} $w,i \models \top$ always holds, $w,i \models \bot$ never holds and for every symbol $\ell \in A \cup \{min,max\}$, $w,i \models \ell$ holds when $\ell = \wpos{w}{i}$.
   \item {\bf Boolean connectives:} $w,i \models \psi_1 \vee \psi_2$ holds if $w,i \models \psi_1$ \emph{or} $w,i \models \psi_2$, $w,i \models \psi_1 \wedge \psi_2$ holds if $w,i \models \psi_1$ \emph{and} $w,i \models \psi_2$, and $w,i \models \neg \psi$ holds if $w,i \models \psi$ \emph{does not}.
\item {\bf Finally:} for $L \in \Cs$,  $w,i \models \finallyl{\psi}$ holds if there exists $j \in \pos{w}$ such that $i < j$,   $\infix{w}{i}{j} \in L$ and $w,j \models \psi$.
  \item {\bf Previously:} for $L \in \Cs$,  $w,i \models \finallyml{\psi}$ holds if there exists $j \in \pos{w}$ such that $j < i$,   $\infix{w}{j}{i} \in L$ and $w,j \models \psi$.
\end{itemize}
If no distinguished position is specified, we evaluate formulas at the \emph{leftmost} unlabeled position. For a formula $\varphi\in\tla{\Cs}$, the language defined by $\varphi$ is $L(\varphi)=\{w \in A^*\mid w,0\models\varphi\}$.

Finally, given a class \Cs, we write \tlc{\Cs} for the class consisting of all languages $L(\varphi)$ where $\varphi\in \tla{\Cs}$. Clearly, $\Cs \subseteq \tlc{\Cs}$ since $L \in \Cs$ is defined by $\finallyl{max} \in \tla{\Cs}$.

\begin{remark} \label{rem:fo2}
	It is well-known~\cite{evwutl} that the standard variants \tls and \tlxs of unary temporal logic (see Section~\ref{sec:groups}) correspond to the variants \fodw  and \fodws of two-variable first-order logic. This correspondence can be generalized: with a class~\Cs, we associate a set \infsigc of first-order predicates. For each $L \in \Cs$, it contains a \emph{binary predicate} $I_L(x,y)$, interpreted as follows: if $w \in A^*$ and $i,j \in \pos{w}$, then $w \models I_L(i,j)$ if and only if $i < j$ and $\infix{w}{i}{j} \in L$. It is shown in~\cite{pzupol2} that for every \emph{Boolean algebra} \Cs, we have $\tlc{\Cs} = \fod(\infsigc)$.
\end{remark}

We turn to a generic algebraic characterization of the classes \tlc{\Cs} from~\cite[Theorem~12]{pztl}. It will serve as a central tool.

\begin{restatable}{theorem}{tlcar} \label{thm:tlcar}
  Let \Cs be a \vari, let $L\subseteq A^{*}$ be a regular language and let $\alpha: A^* \to M$ be its~syntactic morphism. Then, $L \in \tlc{\Cs}$ if and only if $\alpha$ satisfies:
  \begin{equation} \label{eq:tlcar}
    \begin{array}{c}
      \hspace*{-.5ex}(esete)^\omega = (esete)^\omega ete (esete)^\omega \quad \text{for all $e \in E(M)$}\\
      \hspace*{-.5ex}\text{and $s,t \in M$ such that $(e,s),(e,t)$ are \Cs-pairs for $\alpha$.}
    \end{array}
  \end{equation}
\end{restatable}

In view of Theorem~\ref{thm:tlcar}, \tlc{\Cs}-\emph{membership} reduces to \Cs-separation for every \vari~\Cs. Indeed, deciding \Cs-separation suffices to compute the \Cs-pairs, by definition.

\begin{restatable}{corollary}{tlcor} \label{cor:tlcar}
	Let \Cs be a \vari. If \Cs-separation is decidable, then so is \tlc{\Cs}-membership.
\end{restatable}

We will use the next corollary of Theorem~\ref{thm:tlcar} (proved in Appendix~\ref{app:utl}) to show that languages do \textbf{not} belong to \tlc{\Cs}.

\begin{restatable}{lemma}{tlcinduc} \label{lem:tlcinduc}
	Let $L,U,V\!\subseteq\! A^*$, $P\!=\! L^*(UL^*VL^*)^*$ and \Cs be a \vari. If $L$ is \textbf{not} \Cs-separable from~either $U$ or $V$, then~$P$ is  \textbf{not} \tlc{\Cs}-separable from either $P U P$ or~$P V P$.
\end{restatable}

\noindent
\textbf{\Nesths.}  Each such hierarchy is determined by an arbitrary initial class \Cs called its \emph{basis}. For $n \in \nat$, level $n$ is written \tlhn{\Cs}. Level zero is the basis $\tlhp{0}{\Cs}= \Cs$. Then, for all $n \geq 1$, we let $\tlhn{\Cs} = \tlc{\tlhp{n-1}{\Cs}}$. Clearly, we have $\tlhp{n}{\Cs} \subseteq \tlhp{n+1}{\Cs}$ for all $n \in \nat$. The next proposition can be verified using Theorem~\ref{thm:tlcar} (see Appendix~\ref{app:utl} for details).

\begin{restatable}{proposition}{tlvar} \label{prop:tlvar}
	Let \Cs be a \vari. Then, \tlhn{\Cs} is a \vari as well for all $n \in \nat$.
\end{restatable}

\subsection{Connection with concatenation hierarchies}
We now compare the two kinds of hierarchies. We start with a preliminary easy result.

\begin{restatable}{proposition}{tlpol} \label{prop:tlpol}
  For every class \Cs, $\tlc{\Cs} = \tlc{\pol{\Cs}}$.
\end{restatable}

\begin{proof}
  Clearly, $\tlc{\Cs} \subseteq \tlc{\pol{\Cs}}$ as $\Cs \subseteq \pol{\Cs}$. For the converse inclusion, we show that the modalities $\textup{F}_L$ and $\textup{P}_L$ with $L \in \pol{\Cs}$ can be simulated using only modalities $\textup{F}_K$ and $\textup{P}_K$ with $K \in \Cs$ . By symmetry, we only consider the former. For all $\varphi \in \tla{\Cs}$, we build $\langle \varphi \rangle_L \in \tla{\Cs}$ such that for all $w \in A^*$ and $i \in \pos{w}$, $w,i \models \langle \varphi \rangle_L \Leftrightarrow w,i \models \finallyl{\varphi}$.

  By definition of \pol{\Cs}, $L \in \pol{\Cs}$ is built from the languages in \Cs using union and marked product. We perform an induction on this construction. If $L \in \Cs$, it suffices to define $\langle \varphi \rangle_L= \finallyl{\varphi}$, which is a \tla{\Cs} formula. If $L= L_0 \cup L_1$ with $L_0,L_1 \in \pol{\Cs}$, we let $\langle\varphi\rangle_L=\langle\varphi\rangle_{L_0} \vee  \langle\varphi\rangle_{L_1}$. Finally, if $L= L_0aL_1$ with $L_0,L_1 \in \pol{\Cs}$ and $a \in A$, we let $\langle\varphi\rangle_L = \langle a \wedge \langle \psi\rangle_{L_1} \rangle_{L_0}$. This completes the proof.
\end{proof}

We complete Proposition~\ref{prop:tlpol} with the following corollary.

\begin{restatable}{corollary}{ctlpol} \label{cor:tlpol}
  For any class \Cs and $n\!\in\!\nat$, $\bpoln{\Cs}\!\subseteq\! \tlhn{\Cs}$.
\end{restatable}

\begin{proof}
  We show the inclusion  $\bpol{\Ds}\!\subseteq\! \tlc{\Ds}$ for every class~\Ds. The statement will then follow by induction.  Since $\tlc{\Ds} = \tlc{\pol{\Ds}}$ by Proposition~\ref{prop:tlpol}, it is immediate that $\pol{\Ds} \subseteq \tlc{\Ds}$. Therefore, $\bpol{\Ds} \subseteq \bool{\tlc{\Ds}}$. Finally, since \tlc{\Ds} is a Boolean algebra by definition, we have $\bool{\tlc{\Ds}}=\tlc{\Ds}$, completing the~proof.
\end{proof}

\begin{remark}\label{rem:tlvsipol}
  If \Cs is a \pvari, one can strengthen Corollary~\ref{cor:tlpol}: one can show $\ipolp{n+1}\Cs\subseteq\tlhn\Cs$ from the characterizations of $\ipolp2\Cs$ and $\tlc{\Cs}$ (Theorems~\ref{thm:capolbp} and \ref{thm:tlcar}).
\end{remark}

In general, the converse inclusion $\tlhn\Cs\subseteq \bpoln{\Cs}$ fails. For instance, it is shown in~\cite{between} that for the basis \at, level one, \tlc{\at}, already intersects \emph{all} levels \bpolp{n}{\at}. However, the union of all levels is the same: the star-free closure. Indeed, since $\sfp{\sfp{\Cs}} = \sfp{\Cs}$ for every \Cs, it suffices to show that $\tlc{\Cs} \subseteq \sfp{\Cs}$. This follows from known results. As mentioned in Remark~\ref{rem:fo2}, it is shown in~\cite[Theorem~9.9]{pzupol2} that for every Boolean algebra \Cs, we have $\tlc{\Cs} = \fod(\infsigc)$ where $\fod(\infsigc)$ is two-variable first-order logic equipped with a set of predicates built from \Cs. Clearly, $ \fod(\infsigc) \subseteq \fo(\infsigc)$ (with \fo as full first-order logic). Finally, it is shown in~\cite[Theorem~73]{PZ:generic18} that $\sfp{\Cs} = \fo(\infsigc)$ when \Cs is a \vari.
\begin{restatable}{proposition}{sfunion} \label{prop:sfunion}
	If \Cs is a \vari, $\sfp{\Cs}\!=\! \bigcup_{n \in \nat} \tlhn{\Cs}$.
\end{restatable}

\section{Group bases}
\label{sec:groups}
In the literature, concatenation hierarchies use a few specific bases consisting solely of group languages, which we also adopt for \nesths. In this section, we introduce and motivate such ``group bases''. We prove that for the resulting hierarchies, membership at level \emph{two} reduces to separation for the basis, show that the \nesths are strict, and compare them to concatenation hierarchies.

\subsection{Definitions}

\emph{Group languages} are those that can be recognized by a morphism into a finite \emph{group}. A \emph{group \vari} is a \vari \Gs consisting of group languages only. We also look at a slight extension: for each class \Gs, we let $\Gs^+$ be the class consisting of all languages $\{\veps\} \cup L$ and $A^+ \cap L$ for $L \in \Gs$. One can verify that if \Gs is a \vari, then so is $\Gs^+$. 

The simplest group \vari is $\stzer = \{\emptyset,A^*\}$. We let $\dotzer = \stzer^+ = \{\emptyset,\{\veps\}, A^+,A^*\}$. These are the bases of the two most common concatenation hierarchies: the Straubing-Thérien hierarchy~\cite{StrauConcat,TheConcat} and the dot-depth~\cite{BrzoDot}. They both classify the original class \sfr of star-free languages~\cite{schutzsf}. By results of McNaughton-Papert~\cite{mnpfosf} and  Kamp~\cite{kltl}, \sfr also consists of languages that can be defined in first-order logic \fow, or in full linear temporal logic \ltl. Thus, the \nesths of bases \stzer and \dotzer classify the languages that can be defined in \ltl. The modalities $\textup{F}_{A^*}$ and $\textup{P}_{A^*}$ have the semantics of the classic ones $\textup{F}$ and $\textup{P}$ of unary temporal logic (for instance $w,i \models \finally{\varphi}$ when there exists $j \in \pos{w}$ such that $i < j$ and $w,j \models \varphi$). Similarly, $\textup{F}_{\{\veps\}}$ and $\textup{P}_{\{\veps\}}$ have the semantics of the classic modalities $\textup{X}$ and $\textup{Y}$ (for instance, $w,i \models \nex{\varphi}$ when $i+1 \in \pos{w}$ and $w,i+1 \models \varphi$). This means that $\tlc{\stzer}$ and $\tlc{\dotzer}$ are exactly the standard variants of unary temporal logic, which we denote by \tls and \tlxs (see \emph{e.g.},~\cite{evwutl}).  The levels \emph{two} (\emph{i.e.}, \tlhp{2}{\stzer} and \tlhp{2}{\dotzer}) have also been investigated~\cite{between}. Indeed, these are also the levels \emph{one} for the bases \at and~\lt.

\begin{restatable}{proposition}{leveltwo} \label{prop:leveltwo}
	\!\!$\tlhp{2}{\stzer} = \tlc{\at}$\! and $\tlhp{2}{\dotzer} = \tlc{\lt}$.
\end{restatable}
\begin{proof}
  One may verify that $\at \!\subseteq\! \tlc{\stzer}$ and $\lt \!\subseteq\! \tlc{\dotzer}$ which yields the right to left inclusions.  For the \mbox{converse}, results of \mbox{Etessami, Vardi, Thérien, Wilke, Pin and Weil} yield $\tlc{\stzer} \!= \!\ipol{\at}$ and $\tlc{\dotzer} \!= \!\ipol{\lt}$~\cite{evwutl,twfo2,pwdelta2}. Thus, \mbox{$\tlhp{2}{\stzer}\!\subseteq\! \tlc{\pol{\at}}$ and $\tlhp{2}{\dotzer}\!\subseteq\! \tlc{\pol{\lt}}$}. This completes the proof since $\tlc{\at}\!=\!\tlc{\pol{\at}}$ and $\tlc{\lt} \!=\! \tlc{\pol{\lt}}$ by Proposition~\ref{prop:tlpol}.
\end{proof}

Proposition~\ref{prop:leveltwo} establishes a connection with the work of Krebs, Lodaya, Pandya and Straubing~\cite{between}, who investigated \tlc{\at} and \tlc{\lt}. They mostly work with  definitions in terms of two-variable first-order logic: $\fod(\infsig{\at})$ and $\fod(\infsig{\lt})$ ({strictly speaking, they use specialized sets of predicates, not the generic ones \infsig{\at} and \infsig{\lt}. However, it is immediate that they yield the same classes of languages}). They prove that both classes have decidable \emph{membership}. In the paper, we prove decidability for separation and covering~as~well.

\smallskip

The bases \stzer and $\dotzer = \stzer^+$ give rise to hierarchies that fall within the class \sfr of star-free languages. More generally, if the only group languages in \Cs are $\emptyset$ and $A^*$, then all~levels of the concatenation or \nesth of basis \Cs are subclasses of~\sfr. In other words, to go beyond star-free languages, we need bases that include nontrivial group languages. This makes group \varis particularly appealing choices as bases for both types of hierarchies.

\smallskip Besides the \vari \stzer, there are three other standard examples of group \varis: the classes \md of \emph{modulo languages} (membership of a word depends only on its length modulo a fixed integer), \abg of \emph{alphabet modulo testable languages} (membership of a word depends only on the number of occurrences of each letter modulo a fixed~integer), and \grp of all group languages. The concatenation hierarchies of bases \md and $\md^+$ are important~\cite{MACIEL2000135,ChaubardPS06,KufleitnerW15}. They classify first-order logic with \emph{modular predicates}, $\fo(<,MOD)$. Level one in the \nesth of basis \md (\emph{i.e.}, \tlc{\md}) has also been investigated~\cite{DartoisP13}. The concatenation hierarchies of bases \abg and $\abg^+$ are also natural: they classify variants of first-order logic equipped with predicates counting the number of occurrences of a specific letter occurring before a position, modulo an integer.
Finally, the class \grp of \emph{all} group languages defines a concatenation hierarchy introduced by Margolis and Pin~\cite{MargolisP85}: the \emph{group hierarchy}. For these three classes, separation is decidable~\cite{Ash91,abelianp,pzgr}.

\subsection{Membership for levels of \nesths}

Membership for the levels \emph{one} was handled in~\cite{pzupol2}: if \Gs is a group \vari \Gs with decidable \emph{separation}, then \tlc{\Gs} and \tlc{\Gs^+} have decidable \emph{membership} (this is also immediate from Theorem~\ref{thm:tlcar}). We prove the same for level \emph{two} using concatenation hierarchies. Since we are dealing with group based hierarchies, we have the following proposition.

\begin{restatable}{proposition}{tl2bpol} \label{prop:tl2bpol}
	Let \Gs be a group \vari and $\Cs\! \in\! \{\Gs,\Gs^+\}$. Then, $\tlhp{2}{\Cs}\!=\! \tlc{\bpol{\Cs}}$.
\end{restatable}

\begin{proof}
	We have $\tlc{\bpol{\Cs}} \subseteq \tlhp{2}{\Cs}$ as $\bpol{\Cs} \subseteq \tlc{\Cs}$ by  Corollary~\ref{cor:tlpol}. For the converse inclusion, it is shown in~\cite[Theorems~10.1 and 10.4]{pzupol2} that $\tlc{\Cs} = \ipol{\bpol{\Cs}}$ (this requires the hypothesis that \Gs consists of group languages). In particular, $\tlc{\Cs}\subseteq \pol{\bpol{\Cs}}$ which yields $\tlhp{2}{\Cs} \subseteq \tlc{\pol{\bpol{\Cs}}}$. Finally, Proposition~\ref{prop:tlpol} yields $\tlc{\pol{\bpol{\Cs}}} = \tlc{\bpol{\Cs}}$, completing the proof.
\end{proof}

Proposition~\ref{prop:tl2bpol} can be combined with Corollary~\ref{cor:tlcar}. We obtain that for every group \vari \Gs, if $\Cs \in \{\Gs,\Gs^+\}$, then \tlhp{2}{\Cs}-\emph{membership} reduces to \bpol{\Cs}-\emph{separation}. The latter problem has been extensively studied for group classes. If $\Cs \in \{\Gs,\Gs^+\}$, it is shown in~\cite{pzconcagroup,PlaceZ22} that \bpol{\Cs}-\emph{separation} reduces to \Cs-\emph{separation}. This yields the following result.

\begin{restatable}{corollary}{tl2group} \label{cor:nth:tl2group}
  Let \Gs be a group \vari with decidable separation, \tlhp{2}{\Gs} and \tlhp{2}{\Gs^+} have decidable membership.
\end{restatable}

Corollary~\ref{cor:nth:tl2group} applies to the bases \stzer, \md, \abg and~\grp, as they have decidable separation~\cite{Ash91,abelianp,pzgr}. We push~this for \stzer and $\dotzer = \stzer^+$ in Section~\ref{sec:carav}: we prove that \tlhp{2}{\stzer} and  \tlhp{2}{\dotzer} have decidable \emph{separation}. Thus, by Corollary~\ref{cor:tlcar}, \tlhp{3}{\stzer} and \tlhp{3}{\dotzer} have decidable membership.

\subsection{Strictness and strict intertwining of hierarchies}

We prove that for \emph{every group \vari \Gs}, the \nesths of bases \Gs and $\Gs^+$ are \emph{strict} and strictly intertwined if the alphabet contains at least two letters. Thus, we assume that the alphabet $A$ contains at least \emph{two} letters, written $a$ and~$b$.

We start from languages introduced by Brzozowski and Knast~\cite{BroKnaStrict} to prove that the \emph{dot-depth} hierarchy (\emph{i.e.}, the concatenation hierarchy of basis \dotzer) is strict. We first recall these languages, before adapting them to the generic setting.~Let,
\[
H_0 = \{\veps\} \quad \text{and} \quad H_n = (aH_{n-1}b)^* \text{ for $n \geq 1$.}
\]
Brzozowski and Knast proved that $H_n \in \bpolp{n}{\dotzer}$ for all $n \in \nat$ and $H_{n}\not\in\bpolp{n-1}{\dotzer}$ for all $n \geq 1$. Actually, these languages imply that the \emph{\nesth} of basis \dotzer is strict, as one can show that $H_{n}\not\in\tlhp{n-1}{\dotzer}$. We generalize them to handle \emph{all bases} \Gs and $\Gs^+$ \emph{simultaneously}.

\begin{remark}
	The languages $H_n$ themselves cannot be used in the general case. Indeed, one can show that for all $n \in \nat$, $H_n$ belongs to level \textbf{one} of the group hierarchy (\emph{i.e.}, \bpol{\grp}).
\end{remark}

We prove that \emph{for every \vari \Gs}, the \nesths of bases \Gs and $\Gs^+$ are \emph{strict} and \emph{strictly intertwined}.
Since~$\{\veps\}, A^+ \in \bpol{\Gs}$, we have $\Gs^+ \subseteq \bpol{\Gs}$. Hence, $\Gs \subseteq \Gs^+ \subseteq \bpol{\Gs} \subseteq \tlc{\Gs}$. Therefore, a simple induction yields the following inclusions:
\begin{equation}\label{eq:inclusions}
  \begin{array}{lclcl}
	\tlhp{n}{\Gs} &\subseteq& \tlhp{n}{\Gs^+} &\subseteq &\tlhp{n+1}{\Gs}, \\
	\bpolp{n}{\Gs} &\subseteq &\bpolp{n}{\Gs^+} &\subseteq &\bpolp{n+1}{\Gs}.
  \end{array}
\end{equation}
Our goal is to show that the four inclusions of~\eqref{eq:inclusions} are \emph{strict}. Since $\bpoln{\Cs} \subseteq \tlhn{\Cs}$ for all classes \Cs by Corollary~\ref{cor:tlpol}, it suffices to exhibit two languages $K_n,L_n$ such that,
\begin{itemize}
  \item $K_n\in\bpolp{n}{\Gs^+}$, but $K_n\notin\tlhp{n}\Gs$.
  \item $L_n\in\bpolp{n+1}{\Gs}$, but $L_n\notin\tlhp{n}{\Gs^{+}}$.
\end{itemize}

We modify the languages $H_n$ by replacing the letters ``$a$'' and ``$b$'' by more involved languages.

Consider the word $x_i = ab^i$ and the language $Y_i =a^+b^i $ for all $i \geq 1$. Moreover, let $Q =  x_1^+x_2^+x^{}_1$, $R = x_3^+x_4^+x^{}_3$, $S = Y_1^+Y_2^+Y^{}_1$ and $T = Y_3^+Y_4^+Y^{}_3$. Finally, we define,
\[
\begin{array}{lcl}
	K_0  = \{\veps\} & \text{and}& K_{n}  = (QK_{n-1}R)^*  \text{ for $n \geq 1$},\\
	L_0  =   a^* &\text{and} & L_{n}  =  (a + SL_{n-1}T)^*   \text{ for $n \geq 1$}.
\end{array}
\]

The strictness proof is based on two lemmas. The first one considers the operator \bpolo. We prove it in Appendix~\ref{app:groups}.

\begin{restatable}{lemma}{strictpos} \label{lem:strictpos}
	For every $n\in \nat$, we have $K_n \in \bpolp{n}{\dotzer}$ and $L_n \in  \bpolp{n+1}{\stzer}$.
\end{restatable}

Since $\stzer \subseteq \Gs$ and $\dotzer = \stzer^+ \subseteq \Gs^+$,  Lemma~\ref{lem:strictpos} implies that $K_n \in \bpolp{n}{\Gs^+}$ and $L_n \in  \bpolp{n+1}{\Gs}$ for every group \vari~\Gs. Since $\Gs \subseteq \grp$ and $\Gs^+ \subseteq \grp^+$,  the next lemma implies that $K_n \not\in \tlhp{n}{\Gs}$ and $L_n \not\in \tlhp{n}{\Gs^+}$. Do note that in the proof of this lemma, separation plays a central role to conduct the induction, as explained in the introduction.

\begin{restatable}{lemma}{strictneg} \label{lem:strictneg}
	For every $n \in \nat$, we have $K_n \not\in \tlhn{\grp}$ and $L_n \not\in \tlhp{n}{\grp^+}$.
\end{restatable}

\begin{proof}
  We sketch the proof that $K_n\! \not\in\! \tlhn{\grp}$ for all $n$ using Lemma~\ref{lem:tlcinduc} (for details and the similar proof that $L_n\! \not\in\! \tlhp{n}{\grp^+}$, see Appendix~\ref{app:groups}). Let $Q_0 = Q$ and $R_0 = R$. For $n \geq 1$, let $Q_n = K_nQ_{n-1}K_n$ and $R_n = K_n R_{n-1} K_n$. We use induction on $n$ to show that $K_n$ is \emph{not} \tlhn{\grp}-separable from either $Q_n$ and $R_n$. This yields $K_n \not\in \tlhn{\grp}$, since $K_n$ does not intersect $Q_n$ nor $R_n$.
  
  When $ n = 0$, we have to prove that $\{\veps\}$ is  \emph{not} \grp-separable from either $Q = x_1^+x_2^+x^{}_1$ and $R = x_3^+x_4^+x^{}_3$. By symmetry, we only prove the former. Let $U \in \grp$ such that $\veps \in U$. We show that $U\cap x_1^+ x_2^+ x^{}_1 \neq \emptyset$. By hypothesis, $U$ is recognized by a morphism $\alpha: A^*\to G$ into a finite group $G$. Let $p=\omega(G)$ and $w=x_1^{p-1} x_2^p x_1 \in x_1^+x_2^+x^{}_1$. Since $G$ is a group, $\alpha(w) = 1_G = \alpha(\veps)$. Therefore, $w \in U\cap x_1^+ x_2^+ x^{}_1$.

  Assume that $n \geq 1$. By induction, $K_{n-1}$ is  \emph{not} \tlhp{n-1}{\grp}-separable from either $Q_{n-1}$ or $R_{n-1}$.  One can check from the~definition that $K_{n-1}^* = K_{n-1}  \subseteq K_n$. By definition of $K_n$, this yields $K_n =  K_{n-1}^*(Q^{}_{n-1}K_{n-1}^*R^{}_{n-1}K_{n-1}^*)^*$. Thus, Lemma~\ref{lem:tlcinduc} implies that $K_n$ is not \tlhn{\grp}-separable from either $Q_n = K_nQ_{n-1}K_n$ or $R_n = K_nR_{n-1}K_n$.
\end{proof}

Altogether, Lemmas~\ref{lem:strictpos} and~\ref{lem:strictneg} yield the desired theorem.
\begin{restatable}{theorem}{strict} \label{thm:thmstrict}
	Let \Gs a be group \vari. If $|A| \geq 2$, then,
	\vspace{-0.1cm}
	\begin{itemize}
		\item The \nesths of bases \Gs and $\Gs^+$ are both strict and strictly intertwined.
		\item The concatenation hierarchies of bases \Gs and $\Gs^+$ are both strict and strictly intertwined.
	\end{itemize}
\end{restatable}

\subsection{Comparing concatenation and \nesths}

We compare the two kinds of hierarchies. If $\Cs \in \{\Gs,\Gs^+\}$ for a group \vari \Gs, Lemma~\ref{lem:strictpos} and Lemma~\ref{lem:strictneg} yield that for all $n \in \nat$, \bpolp{n+1}{\Cs} contains a language outside of \tlhp{n}{\Cs}. Here, we prove that for all $n \in \nat$, \tlhp{2}{\Cs} contains a language outside of \bpolp{n}{\Cs}. Altogether, it follows that for all $m\geq 2$ and all $k >m$,  \tlhp{m}{\Cs} and \bpolp{k}{\Cs} are \emph{incomparable}. We depict the situation for $m=2$ in Figure~\ref{fig:hiera2}.

\begin{figure}[H]
  \centering
  \includegraphics[width=8.5cm]{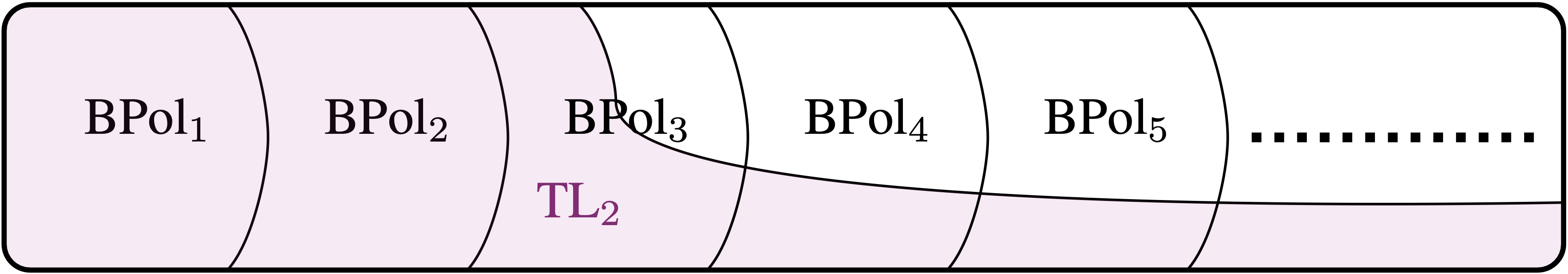}
% \begin{tikzpicture}[xscale=0.65,yscale=0.24, every node/.style={scale=0.8},fill between/on layer=main]

% 	\begin{scope}
% 	  \clip[rounded corners] (0,0) rectangle (13,6);

% 	  \draw[name path=Z] (0,0) to (0,6) {};

% \node at (1.2,3.4) {$\textup{BPol}_1$};
% 	  \draw[name path=C] (1.5,0) to [out=-70,in=-90] (2.2,3) to [out=90,in=70] (1.5,6) ;

% 	  \node at (3.2,3.4) {$\textup{BPol}_2$};
% 	  \draw[name path=E] (3.5,0) to [out=-70,in=-90] (4.2,3) to [out=90,in=70] (3.5,6) ;

% 	  \node at (4.8,1.5) {\textcolor{byzantine!70!black}{$\textup{TL}_2$}};
% 	  \draw[name path=D] (14,0.0) to (14,1.6) to [out=6,in=-90]  (5.1,3.8) to [out=90,in=70] (4.5,6);
% \tikzfillbetween[of=Z and D]{byzantine, opacity=0.11}

% 	  \node at (5.2,3.4) {$\textup{BPol}_3$};

% \draw[name path=E] (5.5,0) to [out=-70,in=-90] (6.2,3) to [out=90,in=70] (5.5,6) ;

% 	  \node at (7.2,3.4) {$\textup{BPol}_4$};

% \draw[name path=E] (7.5,0) to [out=-70,in=-90] (8.2,3) to [out=90,in=70] (7.5,6) ;

% 	  \node at (9.2,3.4) {$\textup{BPol}_5$};
% 	  \draw[name path=E] (9.5,0) to [out=-70,in=-90] (10.2,3) to [out=90,in=70] (9.5,6) ;

% 	  \draw[line width=1.5pt,dotted] (10.4,3.0) to (12.8,3.0) ;
% 	\end{scope}

% 	\draw[thick,rounded corners] (0,0) rectangle (13,6);
%   \end{tikzpicture}
  \caption{Level 2 (filled) in a \nesth (bases \Gs/$\Gs^{+}$)}
  \label{fig:hiera2}
  \vspace{-2ex}
\end{figure}

The precise statement differs slightly depending on the basis. We start with \stzer which requires arbitrarily large alphabets.

\medskip
\noindent
\textbf{The basis \stzer.} It is shown in~\cite{between} that \tlc{\at} (\emph{i.e.}, \tlhp{2}{\stzer} by Proposition~\ref{prop:leveltwo}) contains a language outside of $\bpolp{n}{\stzer}$ for all $n \in \nat$. We reprove this result using \emph{distinct languages}, as we shall reuse them to handle other bases.

For each $k \in \nat$, we define an alphabet $A_k = \{\ell_0,\dots,\ell_k\}$ of size $k+1$ and two disjoint languages $U_k,V_k \subseteq A_k^*$. We let $U_0 = \{\veps\}$ and $V_0 = A_0^+$.  For $k \geq 1$, we let,
\[
U_k = (\ell_kV_{k-1})^* \quad \text{and} \quad V_k = (\ell_kV_{k-1})^* \ell_kU_{k-1} (\ell_kV_{k-1})^*.
\]

\begin{restatable}{lemma}{langtl2st}\label{lem:langtl2st}
	For every $k \in \nat$, we have $U_k,V_k \in  \tlhp{2}{\stzer}$.
\end{restatable}

\newcommand{\modif}[1]{\ensuremath{\left\langle#1\right\rangle}\xspace}
\begin{proof}
	We use induction on $k$. If $k=0$, then $U_0 = \{\veps\}$ is defined by $\neg \finally{\neg max}$ and $V_0 = \ell_0^+$ by $\finally{\ell_0} \wedge \neg \finally{(\neg (\ell_0 \vee max))}$. Hence, $U_0,V_0 \in \tlhp{2}{\stzer}$. Assume now that $k \geq 1$. Induction yields $U_{k-1},V_{k-1} \in  \tlhp{2}{\stzer}$. We get two formulas $\varphi_U,\varphi_V \in   \tlhp{2}{\stzer}$ defining $U_{k-1}$ and $V_{k-1}$, respectively. We use them to build new formulas defining $U_k$ and $V_k$. First, we modify $\varphi_U,\varphi_V$ to constrain their evaluation to the infixes in $A_{k-1}^*$. Let \modif{\varphi_U} and \modif{\varphi_V} be the formulas obtained from $\varphi_U$ and $\varphi_V$ by applying the two following kinds of modification:
	\begin{enumerate}
		\item We replace every atomic formula ``$min$'' by ``$min \vee \ell_k$'' and every atomic formula ``$max$'' by ``$max \vee \ell_k$''.
		\item We recursively replace every subformula of the form \finallyp{L}{\varphi'} by \finallyp{L \cap A_{k-1}^*}{\varphi'}, and every sub-formula of the form \finallymp{L}{\varphi'} by \finallymp{L \cap A_{k-1}^*}{\varphi'}.
	\end{enumerate}
	These remain \tlhp{2}{\stzer} formulas since $A_k^* \in \tlc{\stzer}$. Consider the two following formulas of  \tlhp{2}{\stzer}:
	\[
	\begin{array}{lll}
		\psi_U & = & \finallyp{\{\veps\}}{(\ell_k \vee max)} \wedge \neg \finally{\left(\ell_k \wedge \neg \modif{\varphi_V}\right)}, \\
		\psi_V & = & \finallyp{\{\veps\}}{(\ell_k)} \wedge \finally{(\ell_k \wedge \modif{\varphi_U})}\\
		&& \wedge \neg \finally{\bigr(\ell_k \wedge \neg \modif{\varphi_V} \wedge \finally{\left(\ell_k \wedge \neg \modif{\varphi_V}\right)}\bigl)}. \\
	\end{array}
	\]
	One can verify that $\psi_U$ and $\psi_V$ define $U_k$ and $V_k$, respectively. Hence, $U_{k},V_{k} \in \tlhp{2}{\stzer}$, which completes the proof.
\end{proof}

Moreover, we have the following lemma. As for Lemma~\ref{lem:strictneg}, separation plays a central role to conduct its inductive proof.

\begin{restatable}{lemma}{langtl2stlow}\label{lem:langtl2stlow}
	For every $k \in \nat$, we have $U_k\notin \polp{k+1}{\stzer}$.
\end{restatable}

\begin{proof}
  We show by induction on $k$ that $U_k$ is not \polp{k+1}\stzer-separable from~$V_k$, for all $k\in\nat$. Since we have $U_k \cap V_k = \emptyset$, this yields $U_k \not\in \polp{k+1}\stzer$, as desired. For $k = 0$, $U_0 = \{\veps\}$ is not \pol{\stzer}-separable from $V_0 = \ell_0^+$, since the only language in \pol{\stzer} containing $\veps$ is $A^*$. Assume now that $k \geq 1$. By induction, $U_{k-1}$ is not \polp{k}\stzer-separable from $V_{k-1}$. By Lemma~\ref{lem:nsepconcat}, $\ell_kU_{k-1}$ is not \polp{k}\stzer-separable from $\ell_kV_{k-1}$, which means that $\ell_kV_{k-1}$ is not co-\polp{k}\stzer-separable from $\ell_kU_{k-1}$. By Lemmas~\ref{lem:polinduc} and~\ref{lem:polcopol}, it follows that $U_k=(\ell_kV_{k-1})^*$ is not  \polp{k+1}\stzer-separable from $V_k=(\ell_kV_{k-1})^* \ell_kU_{k-1} (\ell_kV_{k-1})^*$.
\end{proof}

From Lemmas~\ref{lem:strictpos},~\ref{lem:strictneg},~\ref{lem:langtl2st} and~\ref{lem:langtl2stlow}, we get the next theorem.

\begin{restatable}{theorem}{thmtlstcomp}\label{thm:tlstcomp}
	Let $m,k \in \nat$ such that $2 \leq m < k < |A|$. Then, the classes \tlhp{m}{\stzer} and $\bpolp{k}{\stzer}$ are incomparable.
\end{restatable}

Theorem~\ref{thm:tlstcomp} involves arbitrarily large alphabets. It is shown in~\cite{between} that when $|A| = \mathbf{2}$, we have $\tlhp{2}{\stzer}  \subseteq \bpolp{3}{\stzer}$. The authors conjecture that for a \emph{fixed} alphabet, \tlhp{2}{\stzer} is included in \emph{some} level \bpolp{n}{\stzer}. We prove a strengthening of this conjecture. We consider the operator ``$\textup{IPol}$''. Recall that $\ipolp{n}{\Cs}\!=\! \polp{n}{\Cs}\! \cap \!\cocl{\polp{n}{\Cs}}\! \subseteq\! \bpolp{n}{\Cs}$ for every class \Cs.

\begin{restatable}{theorem}{thmtlst} \label{thm:tlstinc}
	We have $\ipolp{3}{\stzer} \!\subseteq\! \tlhp{2}{\stzer} \!\subseteq\!  \ipolp{|A|+1}{\stzer}$.
\end{restatable}

\begin{proof}
  The inclusion $\ipolp{3}{\stzer} \subseteq \tlhp{2}{\stzer}$ is shown in~\cite{between} using a result of Weil~\cite{Weil_1993} (we also provide a simple proof in Appendix~\ref{app:groups}).  In the main text, we only prove that $\tlhp{2}{\stzer} \subseteq \ipolp{|A|+1}{\stzer}$ in the special case when $|A| = 2$. The full proof is available in Appendix~\ref{app:groups}.

  Let us consider an alphabet of size two, say $A = \{a,b\}$. Let $L\in A^*$ with $L \in \tlhp{2}{\stzer}$. We prove that $L \in \ipolp{3}{\stzer}$. Let $\Ds = \ipolp{2}\stzer$. One can check that $\bpol{\Ds} = \bpolp{2}{\stzer}$, so that $\ipolp{2}{\Ds} = \ipolp{3}{\stzer}$. Therefore, it suffices to prove that $L \in \ipolp{2}{\Ds}$.
  Let $\alpha: A^* \to M$ be the syntactic morphism of $L$.  By Theorem~\ref{thm:capolbp}, we have to prove that if $(e,s)\in E(M) \times M$ is a \Ds-pair and $t\in M$, then $(esete)^{\omega+1} = (esete)^{\omega} et e(esete)^{\omega}$.

  Let $H = A^*aA^* \cap A^*bA^*$. Clearly, $A^* = H \cup a^* \cup b^*$. By definition of \Ds-pairs, $\alpha\inv(e)$ is not \Ds-separable from $\alpha\inv(s)$. Hence, as $H,a^*,b^* \in \Ds$ (since $\bpol{\stzer} \subseteq \Ds$ by definition), one can verify that one of the three following conditions holds:
  \begin{enumerate}
	\item $H \cap \alpha\inv(e)$ is not \Ds-separable from $H \cap \alpha\inv(s)$,
	\item $a^* \cap \alpha\inv(e)$ is not \Ds-separable from $a^* \cap \alpha\inv(s)$,
	\item $b^* \cap \alpha\inv(e)$ is not \Ds-separable from $b^* \cap \alpha\inv(s)$.
  \end{enumerate}
  We consider three cases depending on which property holds. Assume first that $H \cap \alpha\inv(e)$ is not \Ds-separable from $H \cap \alpha\inv(s)$. In particular, both languages are nonempty. We get $u \in H \cap \alpha\inv(e)$ and $v \in H \cap \alpha\inv(s)$. Let $x \in \alpha\inv(t)$ and $w = uxu$. By definition of $H$, the two letters $a$ and $b$ of $A$ occur in $u,v$ and $w$. One can verify that this implies that $\{u\}$ is neither \at-separable from $\{v\}$ nor $\{w\}$. Thus, $(\alpha(u),\alpha(v)) = (e,s)$ and $(\alpha(u),\alpha(w)) = (e,ete)$ are \at-pairs. Since $L \in \tlhp{2}{\stzer} = \tlc{\at}$, it then follows from Theorem~\ref{thm:tlcar} that $(esete)^\omega = (esete)^\omega ete (esete)^\omega$. Finally, since $L \in \tlc{\at}$, we know that $L$ is \emph{star-free}. Thus, its syntactic monoid $M$ is aperiodic by Schützenberger's theorem~\cite{schutzsf}. In particular, $(esete)^{\omega} = (esete)^{\omega+1}$ and we obtain $(esete)^{\omega+1} = (esete)^{\omega} et e(esete)^{\omega}$ as desired.

  It remains to handle the second case (the third one is symmetrical). Assume that  $a^* \cap \alpha\inv(e)$ is not \Ds-separable from $a^* \cap \alpha\inv(s)$. In particular, both languages are nonempty. We prove that $e = ese$, which yields the desired equality. There are two sub-cases depending on whether $a^{*} \cap \alpha\inv(e) = \{\veps\}$ or $a^+ \cap \alpha\inv(e) \neq \emptyset$. In the first case, since $\{\veps\} \in \Ds$ (this is because $\bpol{\stzer} \subseteq \Ds$), the hypothesis that $a^* \cap \alpha\inv(e)$ is not \Ds-separable from $a^* \cap \alpha\inv(s)$ yields $\veps \in a^* \cap \alpha\inv(s)$ and we get $e = s = 1_M$, whence $ese = e$, as desired. Otherwise, $a^+ \cap \alpha\inv(e) \neq \emptyset$. This yields $k \geq 1$ such that $\alpha(a^k) = e$. Moreover, since $a^* \cap \alpha\inv(s) \neq \emptyset$, we have $\ell \in \nat$ such that $\alpha(a^\ell) = s$. Let $r = \alpha(a)$. We have $e = r^k$ and since $e$ is idempotent, it follows that $e = r^{\omega}$. Moreover, $ese = r^{\omega+\ell}$. Thus, since $M$ is aperiodic, we obtain again~$e = ese$.
\end{proof}

\begin{remark}
  Theorem~\ref{thm:tlstinc} is specific to the level \textbf{two}. Indeed, the level \tlhp{3}{\stzer} contains a language outside of \bpolp{n}{\stzer} for every $n \in \nat$. This is implied by Theorem~\ref{thm:gpualt} below, as one can verify that $\tlhp{2}{\dotzer} \subseteq \tlhp{3}{\stzer}$.
\end{remark}

Theorem~\ref{thm:tlstinc} has a surprising corollary for binary alphabets.

\begin{restatable}{corollary}{delta3} \label{cor:delta3}
	When $|A| = 2$, we have $\tlhp{2}{\stzer} = \ipolp{3}{\stzer}$.
\end{restatable}

Corollary~\ref{cor:delta3} is reminiscent of a famous theorem by Thérien and Wilke~\cite{twfo2} for the previous levels: $\tlc{\stzer}=\ipolp{2}{\stzer}$ (for every alphabet). Naturally, Corollary~\ref{cor:delta3} fails in the general case (see Lemma~\ref{lem:langtl2stlow}). Nonetheless, \tlc{\at} over binary alphabets is the first special case considered by Krebs, Lodaya, Pandya and Straubing in their original paper~\cite{betweenlics}.

\medskip
\noindent
\textbf{Groups bases.} Consider an arbitrary group basis \Gs containing the modulo languages (\md). In this case, two letters suffice. We reuse the above languages $U_k,V_k$, which we encode over a \emph{binary alphabet}. Fix such an alphabet $A = \{a,b\}$. For every $k \in \nat$, we define a morphism $\beta_k: A_k^* \to A^*$ as follows:
\[
\begin{array}{llll}
	\beta_k: & A_k^* & \to     & A^*, \\
	& \ell_i & \mapsto & a^i b a^{k-i}.
\end{array}
\]
The following proposition (proved in Appendix~\ref{app:groups}) shows that this encoding maps \tlhp2\stzer to \tlhp2\md. Combined with Lemma~\ref{lem:langtl2st}, it implies that $\beta_k(U_k) \in \tlhp{2}{\md}$ for all $k \in \nat$.

\begin{restatable}{proposition}{modproj}\label{prop:modproj}
	Let $k \in \nat$ and $L \subseteq A_k^*$. If $L \in \tlhp2{\stzer}$, then $\beta_k(L) \in \tlhp{2}{\md}$.
\end{restatable}

The next statement serves as the counterpart of Lemma~\ref{lem:langtl2stlow}. Its proof is similar and given in Appendix~\ref{app:groups}.

\begin{restatable}{lemma}{tlmodunc}\label{lem:tlmodunc}
	Let \Gs be a group \vari. For all $k \in \nat$, we have $\beta_k(U_k) \notin \polp{k+1}{\Gs}$.
\end{restatable}

Altogether, we get the following theorem.

\begin{restatable}{theorem}{gualt} \label{thm:gualt}
	\mbox{Let \Gs be a group \vari with $\md\! \subseteq\! \Gs$.} Assume that $|A|\! \geq\! 2$. For all $m,k\! \in\! \nat$ such that $2\! \leq\! m\! <\! k$, the classes \tlhp{m}{\Gs} and \bpolp{k}{\Gs} are incomparable.
\end{restatable}

\medskip
\noindent
\textbf{\Wsuit extensions of groups bases.}  Consider the basis $\Gs^+$ for an arbitrary group \vari \Gs. We again reuse the languages $U_k,V_k$ and encode them over a \emph{binary alphabet} $A = \{a,b\}$. For every $k \in \nat$, let $B_k =  A_k \cup \{b\}$. We define two morphisms $\gamma_k: B_k^* \to A_k^*$ and $\delta_k: B_k^* \to A^*$ as follows:
\[
\begin{array}{llll}
	\gamma_k: & B_k^* & \to     & A_k^* \\
	& b                  & \mapsto & \veps \\
	& \ell_i & \mapsto & \ell_i 
\end{array} \quad \text{and}\quad  \begin{array}{llll}
	\delta_k: & B_k^* & \to     & A^* \\
	& b                  & \mapsto & b \\
	& \ell_i & \mapsto & ba^{i+1} 
\end{array}
\]
By definition, if $L \subseteq A_k^*$, then $\gamma_k\inv(L)$ is the language consisting of all words built from a word in $L$ by inserting arbitrarily many ``$b$'''s at any position. The language $\delta^{}_k(\gamma_k\inv(L))$ consists of all words built from those in $\gamma_k\inv(L)$ by replacing each letter $\ell_i$ with the word $ba^{i+1}$. The following proposition  (proved in Appendix~\ref{app:groups}) is analogous to Proposition~\ref{prop:modproj}: it shows that this encoding maps $\tlhp2{\stzer}$ to $\tlhp{2}{\dotzer}$. Combined with Lemma~\ref{lem:langtl2st}, it yields $\delta^{}_k(\gamma_k\inv(U_k)) \in \tlhp{2}{\dotzer}$ for all $k \in \nat$.

\begin{restatable}{proposition}{ddproj} \label{prop:ddproj}
	Let $k \in \nat$ and $L \subseteq A_k^*$. If $L \in \tlhp{2}{\stzer}$, then $\delta^{}_k(\gamma_k\inv(L)) \in \tlhp{2}{\dotzer}$.
\end{restatable}

The following statement serves as the counterpart of Lemma~\ref{lem:langtl2stlow}. Its proof is similar and given in Appendix~\ref{app:groups}.

\begin{restatable}{lemma}{tlddunc}\label{lem:tlddunc}
Let \Gs be a group \vari. For all $k \in \nat$, we have $\delta^{}_k(\gamma_k\inv(U_k)) \notin \polp{k+1}{\Gs^+}$.
\end{restatable}

Altogether, we get the following theorem.

\begin{restatable}{theorem}{gpualt} \label{thm:gpualt}
	Let \Gs be a group \vari and assume that $|A| \geq 2$. For all $m,k \in \nat$ such that $2 \leq m < k$, the classes \tlhp{m}{\Gs^+} and \bpolp{k}{\Gs^+} are incomparable.
\end{restatable}

\section{Rating maps}
\label{sec:ratms}
We use the framework of~\cite{pzcovering2} to handle separation and covering. It relies on objects called \emph{\ratms}. They ``rate'' covers. For each lattice \Cs, each \ratm $\rho$ and each language~$L$, we define the optimal \Cs-covers of $L$ for $\rho$. We reduce \Cs-covering to the computation of optimal \Cs-covers.

\subsection{\Mratms}

A \emph{semiring} is a tuple $(R,+,\cdot)$ where $R$ is a set and ``$+$'' and ``$\cdot$'' are two binary operations called addition and multiplication, which satisfy the following axioms:
\begin{itemize}
  \item $(R,{+})$ is a commutative monoid. Its identity is denoted by $0_R$.
  \item $(R,{\cdot})$ is a monoid. Its identity is denoted by $1_R$.
  \item Distributivity: for $r,s,t \in R$, $r \cdot (s + t) = (r \cdot s) + (r \cdot t)$  and $(r + s) \cdot t = (r \cdot t) + (s \cdot t)$.
\item $0_R$ is a zero for $(R,{\cdot})$: $0_R \cdot r = r \cdot 0_R = 0_R$ for all $r \in R$.
\end{itemize}
A semiring $R$ is \emph{idempotent} when $r + r = r$ for all $r \in R$ (there is no additional constraint on the multiplication). If $(R,+,\cdot)$ is an idempotent semiring, we define a relation $\leq$ over $R$: for all $r,s \in R$, we let $r\leq s$ if and only if $r+s=s$. One can check that~$\leq$ is a partial order compatible with addition and multiplication and that every morphism between two idempotent semirings is increasing for this ordering.

\begin{example}\label{ex:bgen:semiring}
  A key example is the set of all languages $(2^{A^*},\cup,\cdot)$, with union as addition and language concatenation as multiplication ($\{\varepsilon\}$ is neutral). The ordering is inclusion. \end{example}

When dealing with subsets of an idempotent semiring $R$, we shall often apply a \emph{downset operation}. Given $S \subseteq R$, we write $\dclosr S = \{r \in R \mid r \leq s \text{ for some $s \in S$}\}$. We extend this notation to Cartesian products of arbitrary sets with $R$. Given some set $X$ and $S \subseteq X \times R$, we write $\dclosr S = \{(x,r) \in X \times R \mid \text{$\exists s \in R$ s.t. $r \leq s$ and $(x,s) \in S$}\}$.

\medskip
\noindent
{\bf\Mratms.} We define a \emph{\mratm} as a semiring morphism $\rho: (2^{A^*},\cup,\cdot) \to (R,+,\cdot)$ where $(R,+,\cdot)$ is a \emph{finite} idempotent semiring. That is,
\begin{enumerate}
  \item\label{itm:bgen:ford} $\rho(\emptyset) = 0_R$ and for every $K_1,K_2 \subseteq A^*$, we have $\rho(K_1\cup K_2)=\rho(K_1)+\rho(K_2)$.
  \item\label{itm:bgen:fmult} $\rho(\{\varepsilon\}) = 1_R$ and for every $K_1,K_2 \subseteq A^*$, we have $\rho(K_1K_2) = \rho(K_1) \cdot \rho(K_2)$.
\end{enumerate}

For the sake of improved readability, for all words $w \in A^*$, we write $\rho(w)$ for $\rho(\{w\})$. Also, we write $\rho_*: A^* \to R$ for the restriction of $\rho$ to $A^*$: for all $w \in A^*$, we have $\rho_*(w) = \rho(w)$ (this notation is useful when referring to the language $\rho_*\inv(r) \subseteq A^*$ consisting of all \emph{words} $w \in A^*$ such that $\rho(w) = r$). It is a morphism into the finite monoid $(R,\cdot)$.

\begin{remark}
  As the adjective ``\tame'' suggests, there exists a more general notion of ``\ratm''. We do not use this notion, see~\cite{pzcovering2} for the general framework.
\end{remark}

Now, a \emph{\fratm} is a \mratm $\rho: 2^{A^*} \to R$ such that for every language $K \subseteq A^*$, there exist finitely many words $w_1,\dots,w_n \in K$ such that $\rho(K) = \rho(w_1) + \cdots + \rho(w_k)$. In this case, $\rho$ is fully defined by the morphism $\rho_*: A^* \to R$. Indeed, for each $K$, $\rho(K)$ is equal to the sum of all elements $\rho(w)$ for $w\in K$ (although this sum may be infinite, it simplifies to a finite one because addition is commutative and idempotent). Hence, \fratms are finitely representable. In particular, we can discuss algorithms that take \fratms as input.

\smallskip
\noindent
\textbf{Optimal \imprints.} Let $\rho: 2^{A^*} \to R$ be a \mratm. Given a finite set of languages~\Kb, the \emph{$\rho$-\imprint of~\Kb} is the subset of~$R$ defined by $\prin{\rho}{\Kb} = \dclosr \big\{\rho(K) \mid K \in\Kb\big\}$. Intuitively, when \Kb is a cover of some language $L$, this object measures the ``quality'' of \Kb (smaller is better).

We now define optimality. Consider a lattice~\Cs. For a language $L$, an \emph{optimal} \Cs-cover of $L$ for $\rho$ is a \Cs-cover \Kb of $L$ with the least possible $\rho$-\imprint, \emph{i.e.}, such that $\prin{\rho}{\Kb} \subseteq \prin{\rho}{\Kb'}$ for every \Cs-cover $\Kb'$ of $L$. In general, there may be infinitely many optimal \Cs-covers for a given \mratm $\rho$. However, at least one always exists if \Cs is a lattice. We state this property in the following lemma (proved in~\cite[Lemma~4.15]{pzcovering2}).

\begin{restatable}{lemma}{lemopt} \label{lem:opt}
  Given a lattice \Cs, a \mratm  $\rho$ and a language  $L$, there exists an optimal \Cs-cover of $L$ for~$\rho$.
\end{restatable}

By definition, given a lattice \Cs, a \mratm~$\rho$ and a language $L$, all optimal \Cs-covers of $L$ for $\rho$ have the same $\rho$-\imprint. Hence, this unique $\rho$-\imprint is a \emph{canonical} object for \Cs, $\rho$ and $L$. We call it the \emph{\Cs-optimal $\rho$-\imprint on~$L$} and we denote it by $\opti{\Cs}{L,\rho}$. That is, $\opti{\Cs}{L,\rho} = \prin{\rho}{\Kb}$ for every optimal \Cs-cover \Kb of $L$  for $\rho$.

\subsection{Application to covering}

Deciding \Cs-covering for a fixed class \Cs reduces to finding an algorithm that computes \Cs-optimal \imprints. If \Cs is a \emph{Boolean algebra}, we have the following statement of~\cite{pzcovering2}.

\begin{restatable}{proposition}{propbreduc} \label{prop:breduc}
  Let \Cs be a Boolean algebra. There exists an effective reduction from \Cs-covering to the following problem: ``given as input a \fratm $\rho: 2^{A^*} \to R$ and $F \subseteq R$, decide whether $\opti{\Cs}{A^*,\rho} \cap F = \emptyset$''.
\end{restatable}

\begin{proof}[Proof]
  We sketch the reduction (see~\cite{pzcovering2} for details). Let $(L_0,\{L_1,\dots,L_n\})$ be an input of the \Cs-covering problem. One can compute a morphism $\alpha: A^*\to M$ recognizing all \emph{regular} languages $L_i$: for each $i\leq n$, there exists $F_i\subseteq M$ such that $L_i=\alpha\inv(F_i)$. We build a \fratm $\rho: 2^{A^*} \to 2^{M}$. Its rating set is the idempotent semiring $(2^{M},\cup,\cdot)$ whose  multiplication is obtained by lifting the one of $M$. For all $K\subseteq A^*$, we define $\rho(K) =\alpha(K)$. Now, let $F \subseteq 2^M$ be the set of all subsets $X \subseteq M$ such \mbox{that $X \cap F_i \neq \emptyset$ for all $i \leq n$}. Using the fact that \Cs is a Boolean algebra, one can check that $(L_0,\{L_1,\dots,L_n\})$ is \Cs-coverable if and only if $\opti{\Cs}{A^*,\rho} \cap F = \emptyset$.
\end{proof}

The converse is also true: if \Cs-covering is decidable, then one can compute the \Cs-optimal \imprints (see Appendix~\ref{app:ratms}).

\begin{restatable}{proposition}{covreduc} \label{prop:covreduc}
	Let \Cs be a Boolean algebra with decidable covering. Given as input a regular language $L$ and a \fratm $\rho: 2^{A^*} \to R$, one can compute the set \copti{L,\rho}.
\end{restatable}

\noindent
{\bf Pointed optimal \imprints.} By Proposition~\ref{prop:breduc}, an algorithm computing \opti{\Cs}{A^*,\rho} from a \fratm $\rho$ yields a procedure for \Cs-covering. In Section~\ref{sec:carav}, we consider $\Cs = \tlc{\at}$. The procedure itself involves an object carrying \emph{more information}. Let \Cs be a Boolean algebra, $\eta: A^* \to N$ be a morphism into a finite monoid and $\rho: 2^{A^*} \to R$ be a \mratm. The \emph{$\eta$-pointed \Cs-optimal $\rho$-\imprint} is the following set: 
\[
  \popti{\Cs}{\eta}{\rho} = \bigl\{(t,r) \in N \times R \mid r \in \opti{\Cs}{\eta\inv(t),\rho}\bigl\}.
\]
Note that $\popti{\Cs}{\eta}{\rho}\subseteq N \times R$. It will be convenient to use the a functional notation here: given a subset $S \subseteq N \times R$ and $t \in N$, we write $S(t) \subseteq R$ for the set $S(t) = \{r \in R \mid (t,r) \in S\}$.

The set \popti{\Cs}{\eta}{\rho} encodes the \Cs-optimal $\rho$-\imprint on $A^*$, \emph{i.e.}, the subset $\opti{\Cs}{A^*,\rho}$ of $R$. Indeed, the following easy lemma is an immediate corollary of~\cite[Fact~4.7]{pzcovering2}.

\begin{restatable}{lemma}{pointgen}\label{lem:pointgen}
  Let \Cs be a Boolean algebra, $\eta: A^* \to N$ be a morphism and $\rho: 2^{A^*} \to R$ be a \mratm. Then, $\opti{\Cs}{A^*,\rho} = \bigcup_{t\in N}  \opti{\Cs}{\eta\inv(t),\rho} = \bigcup_{t\in N}  \popti{\Cs}{\eta}{\rho}(t)$.
\end{restatable}

Our goal is to compute $\popti{\tlc{\at}}{\etaat}{\rho} \subseteq 2^A \times R$, in order to decide \tlc\at-covering. Recall that $\etaat: A^* \to 2^A$ is the canonical \at-morphism recognizing \emph{all} languages of~\at.

\section{Optimal \imprints for \tlc{\at}}
\label{sec:carav}
\newcommand{\ptlatopti}{\popti{\tlc{\at}}{\etaat}{\rho}}
\newcommand{\auxrq}{\ensuremath{\pi_{R,Q}}\xspace}

We turn to our main theorem. It characterizes the \tlc{\at}-optimal \imprints. The characterization yields a least fixpoint algorithm for computing the set $\ptlatopti$ from a \fratm $\rho$. By Proposition~\ref{prop:breduc} and Lemma~\ref{lem:pointgen}, this yields the decidability of \tlc{\at}-covering. Note that the algorithm  uses a subprocedure from~\cite{pzupol2,PlaceZ20} deciding \emph{covering} for the simpler class $\tlxs = \tlc{\dotzer}$, which is a nontrivial~question.

\begin{remark}
  Considering \tlc{\at}-covering is mandatory even if one is solely interested in separation, as mentioned in the introduction. There does not appear to be a simpler, specialized algorithm for\/ \tlc{\at}-separation. This is unsurprising, as performing a fixpoint computation often requires extra information (see, for example the case of star-free languages~\cite{pzfoj}). Finally, the framework itself seems indispensable: there is no clear way to formulate \tlc{\at}-operation without \ratms.
\end{remark}

\smallskip
\noindent
{\bf Preliminaries.} We use a construction taking as input an idempotent semiring $R$ and an arbitrary subset $Q \subseteq R$ to build a \fratm $\auxrq: 2^{Q^*} \to R$ on the alphabet $Q$.

\begin{restatable}{remark}{remnot} \label{rem:notation}
  Given $q_1,\dots,q_n \in Q$, the notation $q_1 \cdots q_n$ is ambiguous: it could either denote a word of length $n$ in $Q^*$ or the multiplication of $n$ elements in $Q$. For the sake of avoiding confusion, given an element $q \in Q$, we shall write $(q)$ when we use it as a letter: we denote by $(q_1)(q_2)\cdots (q_n) \in Q^n$ the word made of the $n$ letters $q_1,\dots,q_n \in Q$.
\end{restatable}

Since we are defining a \fratm, it suffices to specify the image of single letters $(q)\in Q$. Since $Q\subseteq R$, we simply define $\auxrq((q)) = q$ for all $q \in Q$. Our characterization uses \tlxs-optimal \auxrq-\imprints (see \tlc{\at}-operation below).

\smallskip
\noindent
{\bf Characterization.} Consider a \fratm $\rho: 2^{A^*} \to R$. We define the \tlc{\at}-saturated subsets of $2^A \times R$ for $\rho$. We shall prove that \ptlatopti is the least such set (with respect to inclusion). Let $S \subseteq 2^A \times R$. We say that $S$ is \emph{\tlc{\at}-saturated} for $\rho$ when the following conditions hold:
\begin{enumerate}
  \item {\bf Trivial elements.} $(\etaat(w),\rho(w)) \in S$ for all $w \in A^*$.
  \item {\bf Downset.} We have $\dclosr S = S$.
  \item {\bf Multiplication.} For every $(B_1,r_1),(B_2,r_2) \in S$, we have $(B_1 \cup B_2,r_1r_2) \in S$.
  \item {\bf \tlc{\at}-operation.} For all $B \in 2^A$ and $Q \subseteq S(B)$, we have $\opti{\tlxs}{Q^+,\auxrq} \subseteq S(B)$.
\end{enumerate}

\begin{restatable}{theorem}{main} \label{thm:main}
  Let $\rho: 2^{A^*} \to R$ be a \fratm,  Then, \ptlatopti is the least \tlc{\at}-saturated subset of $2^A \times R$.
\end{restatable}

Theorem~\ref{thm:main} provides an algorithm computing the set \ptlatopti from an input \fratm $\rho: 2^{A^*} \to R$.  Indeed, the least \tlc{\at}-saturated subset of $2^A \times R$ for $\rho$ can be computed using a least fixpoint procedure. It starts from the set of trivial elements and saturates it with the three operations in the definition. Clearly, downset and multiplication can be implemented. For \tlc{\at}-operation, one needs to compute the set $\opti{\tlxs}{Q^+,\auxrq}$ from the \fratm \auxrq. In view of Proposition~\ref{prop:covreduc}, this boils down to deciding \tlxs-covering and it is shown in~\cite{pzupol2,PlaceZ20} that this problem is decidable.

In view of Lemma~\ref{lem:pointgen}, this yields a procedure that computes $\opti{\tlc{\at}}{A^*,\rho}\!\subseteq\! R$ from an input \fratm $\rho: 2^{A^*} \to R$. Thus, \tlc{\at} has decidable covering by  Proposition~\ref{prop:breduc}.

\begin{restatable}{corollary}{cmain} \label{cor:main}
  Separation and covering are decidable for the class $\tlhp{2}{\stzer} = \tlc{\at}$.
\end{restatable}

Corollary~\ref{cor:main} also implies the decidability of separation and covering for the class \tlc{\lt}. It is shown in~\cite[Theorem 8.1]{between} that the classes \tlc{\at} and \tlc{\lt} are connected by an operator called ``enrichment'' or ``wreath product''. It combines \emph{two} classes into a larger one: $\tlc{\lt} = \tlc{\at} \circ \su$. Here, \su stands for the class of ``suffix languages'': it consists of all Boolean combinations of languages $A^*w$ with $w \in A^*$.

\begin{restatable}{remark}{algebrarghhhh}
  The presentation in~\cite{between} differs from the one above. The authors consider the algebraic varieties characterizing \tlc{\at} and \su rather than the classes themselves. Their theorem states that the variety characterizing \tlc{\lt} is exactly the wreath product of the variety $\textbf{M}_e\textbf{DA}$ (which characterizes \tlc{\at}) with the variety $\textbf{D}$ (which characterizes~\su).
\end{restatable}

It is known that the operator $\Cs \mapsto \Cs \circ \su$ preserves the decidability of both separation and covering~\cite{PlaceZ20}. Therefore, we obtain the following additional corollary.

\begin{restatable}{corollary}{cmainlt} \label{cor:mainlt}
  Separation and covering are decidable for the class $\tlhp{2}{\dotzer} = \tlc{\lt}$.
\end{restatable}

Finally, Corollary~\ref{cor:tlcar} yields that level three in both hierarchies has decidable membership.

\begin{restatable}{corollary}{cmemb} \label{cor:level3}
	Membership is decidable for the classes \tlhp{3}{\stzer} and \tlhp{3}{\dotzer}.
\end{restatable}

\begin{proof}[Proof of Theorem~\ref{thm:main}]
	We fix a \fratm $\rho: 2^{A^*} \to R$. There are two directions corresponding to soundness and completeness in the fixpoint procedure computing \ptlatopti. 

	\begin{restatable}{proposition}{sound} \label{prop:sound}
		$\ptlatopti$ is \tlc{\at}-saturated for $\rho$.
	\end{restatable}

	\vspace{-0.2cm}

	\begin{restatable}{proposition}{complete} \label{prop:comp}
		Let $S \subseteq 2^A \times R$ be \tlc{\at}-saturated for $\rho$. For all $B \subseteq A$, there exists a \tlc{\at}-cover $\Kb_B$ of $\etaat\inv(B)$ such that $\prin{\rho}{\Kb_B} \subseteq S(B)$.
	\end{restatable}

	We first complete  the main proof. By Proposition~\ref{prop:sound}, \ptlatopti is \tlc{\at}-saturated. We need to prove that it the least such set. Let $S\subseteq 2^A \times R$ which is \tlc{\at}-saturated for $\rho$. We prove that $\ptlatopti \subseteq S$. Clearly, it suffices to show that $\ptlatopti(B) \subseteq S(B)$ for all $B \in 2^A$. Proposition~\ref{prop:comp} yields a \tlc{\at}-cover $\Kb_B$ of $\etaat\inv(B)$ such that $\prin{\rho}{\Kb_B}\subseteq S(B)$. This completes the proof as $\ptlatopti(B) = \opti{\tlc{\at}}{\etaat\inv(B),\rho}$ and $\opti{\tlc{\at}}{\etaat\inv(B),\rho} \subseteq \prin{\rho}{\Kb_B}$ by definition.

  We now prove Proposition~\ref{prop:comp}, which is the most interesting. Proposition~\ref{prop:sound} is easier and proved in Appendix~\ref{app:carav}. We fix a \tlc{\at}-saturated set  for $\rho$, say $S \subseteq 2^A \times R$.

  We use induction to prove a more general statement involving auxiliary alphabets. A \emph{good triple} is a tuple $(\frA,\beta,\tau)$ where \frA is an alphabet, $\beta: \frA^* \to 2^A$ is a morphism such that $\beta\inv(\emptyset)=\{\veps\}$ and $\tau: 2^{\frA^*} \to R$ is a \fratm such that for every letter $\ell \in \frA$, we have $(\beta(\ell),\tau(\ell)) \in S$. Moreover, we write \cbeta for the class consisting of all languages over \frA recognized by~$\beta$. A \tlc{\cbeta}-\emph{partition} is a \tlc{\cbeta}-cover that is a partition. Our goal is to build a specific \tlc{\cbeta}-partition of $\frA^*$, which we do in the following lemma.

  \begin{restatable}{lemma}{compmain} \label{lem:compmain}
    Let $(\frA,\beta,\tau)$ be a good triple. There is a \tlc{\cbeta}-partition \Kb of $\frA^*$ such that,
    \begin{equation} \label{eq:compmain}
    	\begin{array}{c}
    		\text{for all $K \in \Kb$, there exists $B \in 2^A$ such that} \\
    		\text{$K \subseteq \beta\inv(B)$ and $(B,\tau(K)) \in S$}.
    	\end{array}
    \end{equation}
  \end{restatable}

  We first use Lemma~\ref{lem:compmain} to complete the main proof. Since $S$ is \tlc{\at}-saturated for $\rho$, we have $(\etaat(a),\rho(a)) \in S$ for all $a \in A$ (this is a trivial element). Thus, $(A,\etaat,\rho)$ is a good triple. Also, $\Cs_{\etaat}=\at$ by definition of \etaat. Thus, Lemma~\ref{lem:compmain} yields a \tlc{\at}-partition \Kb of $A^*$ such that for all $K\in\Kb$, there exists $B \in 2^A$ such that $K \subseteq \etaat\inv(B)$ and $(B,\tau(K)) \in S$. For each $B \in 2^A$, we let $\Kb_B = \{K \in \Kb \mid K \subseteq \etaat\inv(B)\}$. It follows from our hypotheses on \Kb~that $\Kb_B$ is a \tlc{\at}-partition of $\etaat\inv(B)$ such that $\tau(K) \in S(B)$ for all $K \in \Kb_B$. By closure under downset, we get $\prin{\rho}{\Kb_B}\subseteq S(B)$, completing the proof.

  \smallskip

  We now prove  Lemma~\ref{lem:compmain}. We fix a good triple $(\frA,\beta,\tau)$ and use induction on the size of the set $\{\beta(w) \mid w \in \frA^+\} \subseteq 2^A$ to build a \tlc{\cbeta}-partition $\Kb$ of $\frA^*$ satisfying~\eqref{eq:compmain}.

  Assume first that $\frA = \emptyset$. Since $\frA^* = \{\veps\}$, it suffices to define $\Kb = \{\{\veps\}\}$. This is clearly a \tlc{\cbeta}-partition of $\{\veps\}$ (the language $\{\veps\}$ is defined by the \tla{\cbeta} formula ``$\neg \finally{(\neg max)}$''). Moreover, $\{\veps\} \subseteq \beta\inv(\emptyset)$ and $(\emptyset,\tau(\veps)) = (\emptyset,1_R) \in S$ since $S$ is \tlc{\at}-saturated for~$\rho$ (this is the trivial element $(\etaat(\veps),\rho(\veps)) \in S$).  Hence, $\Kb$ satisfies~\eqref{eq:compmain}.

  We assume from now on that $\frA \neq \emptyset$. Consider the sets $D \in 2^A$ such that $D = \beta(\ell)$ for a letter $\ell \in \frA$.  We fix such a set $D \in 2^A$ which is \emph{minimal} with respect to inclusion: there is no  $\ell \in \frA$ such that $\beta(\ell) \subsetneq D$. Let $\frAd = \{\ell \in \frA \mid \beta(\ell) = D\}$ and $\frAnd =  \{\ell \in \frA \mid \beta(\ell) \neq D\}$. Since $\frA = \frAd^{} \cup \frAnd$, we can split each word in $\frA^*$ into two parts (each possibly empty): a prefix in $\frAd^*$ and a suffix in $(\frAnd\frAd^*)^*$ (starting at the first letter in \frAnd). We partition $\frAd^*$ and $(\frAnd\frAd^*)^*$ independently and combine the partitions to build that of $\frA^*$.

  \smallskip
  \noindent
  {\bf Step 1: partitioning \boldmath$\fakebold{\frA}^{\boldsymbol*}_{\boldsymbol D}$.} We construct a \tlc{\at}-partition $\Ub$ of $\frAd^*$ satisfying~\eqref{eq:compmain}.  We rely on \tlc{\at}-operation (we cannot use induction as it may happen that $\frAd= \frA$). Let $Q=S(D)$. We use $Q$ as an alphabet. Let \Hb be an optimal \tlxs-cover of $Q^+$ for the \fratm $\auxrq: 2^{Q^*} \to R$: we have $\prin{\auxrq}{\Hb}=\opti{\tlxs}{Q^+,\auxrq}$. Since \tlxs is a Boolean algebra and $Q^+\in\tlxs$, we may assume without loss of generality that \Hb is a \emph{partition} of~$Q^+$. Moreover, since $S$ is \tlc{\at}-saturated, \tlc{\at}-operation yields $\prin{\auxrq}{\Hb}\subseteq S(D)$. We use \Hb to build the desired \tlc{\cbeta}-partition $\Ub$ of $\frAd^*$. To this aim, we define a morphism $\delta: \frAd^* \to Q^*$. Let $\ell \in \frAd$. By definition of \frAd, we have $\beta(\ell) = D$. Since the triple~$(\frA,\beta,\tau)$ is good, this yields~$(D,\tau(\ell)) \in S$. Thus, $\tau(\ell) \in S(D) = Q$ and we may define $\delta(\ell) = (\tau(\ell)) \in Q$ (here, $(\tau(\ell))$ is viewed as a letter of the alphabet $Q$). The fact is proved in Appendix~\ref{app:carav}.

  \begin{restatable}{fact}{fctq} \label{fct:fctq}
    Let $H \subseteq Q^*$. If $H \in \tlxs$, then $\delta\inv(H) \in \tlc{\cbeta}$.
  \end{restatable}

  Let $\Ub = \{\delta\inv(H) \mid H \in \Hb\} \cup \{\{\veps\}\}$. Since \Hb is a \tlxs-partition of $Q^+$ and $\delta: \frAd^* \to Q^*$ is a morphism, Fact~\ref{fct:fctq} implies that $\Ub$ is $\tlc{\cbeta}$-partition of $\frAd^*$. We prove that it satisfies~\eqref{eq:compmain}. Let $U \in \Ub$. If $U = \{\veps\}$, then $\{\veps\} \subseteq \beta\inv(\emptyset)$ and $(\emptyset,\tau(\veps)) = (\emptyset,1_R) \in S$, as desired. Otherwise, $U = \delta\inv(H)$ for some $H \in \Hb$. Since $H \subseteq Q^+$, this yields $U \subseteq \frAd^+$ and since $\beta(\ell) = D$ for all $\ell \in \frAd$, this implies that $U \subseteq \beta\inv(D)$. It now suffices to prove that $(D,\tau(U)) \in S$. We show that $\tau(U) \leq \auxrq(H)$. By definition, this implies that $\tau(U) \in \prin{\auxrq}{\Hb}$ and since we already know that $\prin{\auxrq}{\Hb}\subseteq S(D)$, this yields $(D,\tau(U)) \in S$ as desired. Since $\tau$ is \full, it suffices to prove that for every $w \in U$, we have $\tau(w) \leq \auxrq(H)$. Since $U = \delta\inv(H)$, we have $\delta(w) \in H$ which yields $\auxrq(\delta(w)) \leq \auxrq(H)$. Moreover, it is immediate from the definition of $\delta$ that $\auxrq(\delta(w)) = \tau(w)$. Hence, we obtain $\tau(w) \leq \auxrq(H)$, as desired.

  \smallskip
  \noindent
  {\bf Step 2: partitioning {$\boldsymbol(\fakebold{\frAnd\frA}^{\boldsymbol*}_{\boldsymbol D}\boldsymbol)^{\boldsymbol*}$}.} We reuse the \tlc{\cbeta}-partition $\Ub$ of $\frAd^*$. We build a new good triple $(\frB,\gamma,\zeta)$ and abstract the words in $(\frAnd\frAd^*)^*$ by words in $\frB^*$. Let $\frB = \frAnd \times \Ub$.  We define a morphism $\gamma: \frB^* \to 2^A$ and a \fratm $\zeta: 2^{\frB^*} \to R$. Consider a letter $(\ell,U) \in \frAnd \times \Ub = \frB$. Since $\Ub$ satisfies~\eqref{eq:compmain} by definition, there exists $B_U \in 2^A$ such that $U \subseteq\beta\inv(B_U)$ and $(B_U,\tau(U)) \in S$. We define $\gamma((\ell,U))=\beta(\ell)\cup B_U$ and $\zeta((\ell,U)) = \tau(\ell)\tau(U)$. Observe that since $(\beta(\ell),\tau(\ell)) \in S$ (recall that $(\frA,\beta,\tau)$ is a good triple) and $(B_U,\tau(U)) \in S$, closure under multiplication for $S$ yields $(\gamma((\ell,U)),\zeta((\ell,U))) \in S$. Hence, $(\frB,\gamma,\zeta)$ is a good triple. The next fact can be verified using the hypothesis that $D$ is a \emph{minimal} set such that $D \in \{\beta(\ell) \mid \ell \in \frA\}$ (the details are available in Appendix~\ref{app:carav}).

  \begin{restatable}{fact}{inducpar} \label{fct:inducpar}
    We have $\{\gamma(w)\! \mid\! w \in \frB^+\}\!\subseteq\! \{\beta(w)\! \mid\! w \in \frA^+\} \setminus \{D\}$.
  \end{restatable}

  Fact~\ref{fct:inducpar} yields $|\{\gamma(w) \mid w \in \frB^+\}| < |\{\beta(w) \mid w \in \frA^+\}|$. Therefore, induction in the proof of Lemma~\ref{lem:compmain} yields a \tlc{\cgamma}-partition $\Kb_\frB$ of $\frB^*$ satisfying~\eqref{eq:compmain}.

  We use $\Kb_\frB$ to build a \tlc{\cbeta}-partition \Vb of $(\frAnd\frAd^*)^*$ satisfying~\eqref{eq:compmain}. We first define a map $\mu: (\frAnd\frAd^*)^*\to\frB^*$. Let $w\in(\frAnd\frAd^*)^*$. We have $w = \ell_1 u_1 \cdots \ell_n u_n$ with $\ell_1,\dots,\ell_n \in \frAnd$ and $u_1,\dots,u_n \in \frAd^*$. Recall that $\Ub$ is a \emph{partition} of $\frAd^*$. Thus, for each $h \leq n$, there exists a \emph{unique} $U_h \in \Ub$ such that $u_h \in U_h$. We let $\mu(w)\! =\! (\ell_1,U_1) \cdots (\ell_n,U_n)\! \in\! \frB^*$. The fact can be verified. 

  \begin{restatable}{fact}{preserve} \label{fct:preserve}
    For every $w \in (\frAnd\frAd^*)^*$, we have $\beta(w) = \gamma(\mu(w))$ and $\tau(w) \leq \zeta(\mu(w))$.
  \end{restatable}

  We say that a position $i \in \pos{w}$ is \emph{distinguished} if $\wpos{w}{i} \in \frAnd \cup \{min,max\}$. There exists a natural bijection between the distinguished positions of $w$ and \emph{all} positions of $\mu(w)$. We abuse terminology and also write $\mu$ for the bijection: we have $\mu(0) = 0$, $\mu(|w|+1) = |w|+1$ and if $i \in \pos{w}$ is the distinguished position of $w$ corresponding to the highlighted letter $\ell_h \in \frAnd$, then $\mu(i) =h$. The following fact can be verified using the hypothesis that $D$ is a \emph{minimal} set such that $D \in \{\beta(\ell) \mid \ell \in \frA\}$ (we present a proof in Appendix~\ref{app:carav}).

  \begin{restatable}{fact}{constrain} \label{fct:constrain}
    Let $w \in (\frAnd\frAd^*)^*$ and $i,j \in \pos{w}$ such that $i < j$. Then, there exists no distinguished position $k \in \pos{w}$ such that $i < k < j$ if and only if $\beta(\infix{w}{i}{j}) \in \{D,\emptyset\}$.
  \end{restatable}

  For every letter $(\ell,U) \in \frB = \frAnd \times \Ub$, Fact~\ref{fct:constrain} can be used to construct a formula $\psi_{\ell,U} \in \tlc{\cbeta}$ such that for every $w \in (\frAnd\frAd^*)^*$ and every distinguished position $i \in \pos{w}$, we have $w,i \models \psi_{\ell,U} \Leftrightarrow \mu(i) \text{ has label } (\ell,U)$. Indeed, $\psi_{\ell,U}$ has to check that $i$ has label $\ell$ (this is achieved by the \tla{\cbeta} formula ``$\ell$'') and that $\infix{w}{i}{j} \in U$ for the next distinguished position $j$ of $w$. By hypothesis on $\Ub$, $U$ is defined by a \tla{\cbeta} formula $\varphi_U$. In view of Fact~\ref{fct:constrain}, we can modify $\varphi_U$ in order to constrain its evaluation so that it stays in the infix $\infix{w}{i}{j} \in \frAd^*$. Roughly speaking, this boils down to modifying all occurrences of $\textup{F}_L$ and $\textup{P}_L$ for $L \in \cbeta$: one replaces $L$ by $L\cap\beta\inv(\{D,\emptyset\}) \in \cbeta$. We detail this construction in Appendix~\ref{app:carav}, where it is used to prove the following lemma.

  \begin{restatable}{lemma}{tlatabs} \label{lem:tlatbas}
    For all $K \in \tlc{\cgamma}$, $\mu\inv(K) \in \tlc{\cbeta}$.
  \end{restatable}

  We let $\Vb = \{\mu\inv(K) \mid K \in \Kb_\frB\}$. Since $\Kb_\frB$ is a \tlc{\cgamma}-partition of $\frB^*$,  \Vb is a \tlc{\cbeta}-partition of $(\frAnd\frAd^*)^*$ by Lemma~\ref{lem:tlatbas}. Let us verify that it satisfies~\eqref{eq:compmain}. Let $V \in \Vb$. By definition, $V = \mu\inv(K)$ for some $K \in \Kb_\frB$. Hence, since $\Kb_\frB$ satisfies~\eqref{eq:compmain}, we get $B \in 2^A$ such that $K \subseteq \gamma\inv(B)$ and $(B,\zeta(K)) \in S$. Since $K \subseteq \gamma\inv(B)$, Fact~\ref{fct:preserve} implies that $V = \mu\inv(K) \subseteq \beta\inv(B)$. Moreover, Fact~\ref{lem:tlatbas} also implies that $\tau(V) = \tau(\mu\inv(K)) \leq \zeta(K)$. Altogether, we obtain  $(B,\tau(V)) \in S$ by closure under downset. Hence, \eqref{eq:compmain} holds.

  \smallskip
  \noindent
  {\bf Step 3: partitioning $\fakebold{\frA}^{\boldsymbol*}$.} We have two \tlc{\cbeta}-partitions $\Ub$ and \Vb of $\frAd^*$ and $(\frAnd \frAd^*)^*$ respectively satisfying~\eqref{eq:compmain}. We combine them using a standard concatenation principle for unary temporal logic (see Appendix~\ref{app:carav} for the proof).

  \begin{restatable}{lemma}{theutlconc} \label{lem:theutlconcat}
    Let $U,V\! \in\! \tlc{\cbeta}$ such that $U\! \subseteq\! \frAd^*$ and $V\! \subseteq\! (\frAnd \frAd^*)^*$. Then, $UV\! \in \!\tlc{\cbeta}$.
  \end{restatable}

  We now build the \tlc{\cbeta}-partition \Kb of~$\frA^*$ described in Lemma~\ref{lem:compmain}. We let $\Kb = \{UV \mid U \in \Ub \text{ and } V \in \Vb\}$. Since $\Ub$ and \Vb are partitions of $\frAd^*$ and $(\frAnd \frAd^*)^*$ and $\frAd \cap \frAnd = \emptyset$, one can verify that \Kb is a partition of $\frA^* = \frAd^*(\frAnd \frAd^*)^*$. Thus, it is a \tlc{\cbeta}-partition of $\frA^*$ by Lemma~\ref{lem:theutlconcat}.

  It remains to prove that \Kb satisfies~\eqref{eq:compmain}. Let $K \in \Kb$. We have $K = UV$ for $U \in \Ub$ and $V \in \Vb$.  Since $\Ub,\Vb$ satisfy~\eqref{eq:compmain}, we get $B,C \in 2^A$ such that $U \subseteq \beta\inv(B)$, $V \subseteq \beta\inv(C)$, $(B,\tau(U)) \in S$ and $(C,\tau(V)) \in S$. Consequently, we obtain $K = UV \subseteq \beta\inv(B \cup C)$ and since $S$ is closed under multiplication, we get $(B \cup C,\tau(UV)) \in S$ as desired.
\end{proof}

\section{Conclusion}
\label{sec:conc}
We defined \nesths using the ``unary temporal logic'' operator $\Cs \mapsto \tlc{\Cs}$. We focused on bases of the form $\Cs=\Gs$ or $\Gs^+$, where \Gs is a group \vari. In this context, we proved that the \nesths of basis~\Cs are \emph{strict} and \emph{strictly intertwined}, and we compared them with the associated concatenation hierarchies. Additionaly, we proved that if \Gs has decidable \emph{separation}, then membership is decidable for levels \emph{one} and \emph{two} in both hierarchies. Our main theorem provides a stronger result for the particular bases $\stzer$ and $\dotzer=\stzer^+$. Specifically, it implies that level \emph{two} in both hierarchies has decidable \emph{covering}, which in turn ensures that level \emph{three} has decidable membership.

Several follow-up questions arise. A natural one is whether our results can be extended to higher levels in \nesths. Addressing this would require studying the operator $\Cs \mapsto \tlc{\Cs}$ in a broader context, as our main theorem heavily relies on the specific properties of the class \at. Another question is to generalize the link existing between the hierarchies of bases \stzer and \dotzer. As seen in Section~\ref{sec:carav}, we have $\tlhp{2}{\dotzer} = \tlhp{2}{\stzer} \circ \su$ and we used this connection to get the decidability of covering for \tlhp{2}{\dotzer}. It turns out that this can be generalized: we have $\tlhn{\dotzer}=\tlhn{\stzer} \circ \su$ for all $n \geq 1$. We leave the proof for further work.

\bibliographystyle{abbrv}
\bibliography{main}

\begin{thebibliography}{10}

\bibitem{arfi87}
M.~Arfi.
\newblock Polynomial operations on rational languages.
\newblock In {\em Proceedings of the 4th Annual Symposium on Theoretical
  Aspects of Computer Science, {STACS'87}}, volume 247 of {\em Lect. Notes
  Comp. Sci.}, pages 198--206. Springer, 1987.

\bibitem{arfi91}
M.~Arfi.
\newblock Op{\'{e}}rations polynomiales et hi{\'{e}}rarchies de
  concat{\'{e}}nation.
\newblock {\em Theoretical Computer Science}, 91(1):71--84, 1991.

\bibitem{Ash91}
C.~J. Ash.
\newblock Inevitable graphs: a proof of the type {II} conjecture and some
  related decision procedures.
\newblock {\em International Journal of Algebra and Computation},
  1(1):127--146, 1991.

\bibitem{BrzoDot}
J.~A. Brzozowski and R.~S. Cohen.
\newblock Dot-depth of star-free events.
\newblock {\em Journal of Computer and System Sciences}, 5(1):1--16, 1971.

\bibitem{BroKnaStrict}
J.~A. Brzozowski and R.~Knast.
\newblock The dot-depth hierarchy of star-free languages is infinite.
\newblock {\em Journal of Computer and System Sciences}, 16(1):37--55, 1978.

\bibitem{bslt73}
J.~A. Brzozowski and I.~Simon.
\newblock Characterizations of locally testable events.
\newblock {\em Discrete Mathematics}, 4(3):243--271, 1973.

\bibitem{ChaubardPS06}
L.~Chaubard, J.-{\'{E}}. Pin, and H.~Straubing.
\newblock First order formulas with modular predicates.
\newblock In {\em Proceedings of the 21th {IEEE} Symposium on Logic in Computer
  Science ({LICS}'06)}, pages 211--220, 2006.

\bibitem{DartoisP13}
L.~Dartois and C.~Paperman.
\newblock Two-variable first order logic with modular predicates over words.
\newblock In {\em Proceedings of the 30th International Symposium on
  Theoretical Aspects of Computer Science, {STACS'13}}, volume~20 of {\em
  Leibniz International Proceedings in Informatics (LIPIcs)}, pages 329--340.
  Schloss Dagstuhl--Leibniz-Zentrum fuer Informatik, 2013.

\bibitem{abelianp}
M.~Delgado.
\newblock Abelian poinlikes of a monoid.
\newblock {\em Semigroup Forum}, 56(3):339--361, 1998.

\bibitem{small_fragments08}
V.~Diekert, P.~Gastin, and M.~Kufleitner.
\newblock A survey on small fragments of first-order logic over finite words.
\newblock {\em International Journal of Foundations of Computer Science},
  19(3):513--548, 2008.

\bibitem{evwutl}
K.~Etessami, M.~Y. Vardi, and T.~Wilke.
\newblock First-order logic with two variables and unary temporal logic.
\newblock {\em Information and Computation}, 179(2):279--295, 2002.

\bibitem{iwfo2alt}
N.~Immerman and P.~Weis.
\newblock Structure theorem and strict alternation hierarchy for fo2 on words.
\newblock In {\em Proceedings of the 16th International Conference on Computer
  Science Logic}, CSL'07, pages 343--357, 2007.

\bibitem{kltl}
H.~W. Kamp.
\newblock {\em Tense Logic and the Theory of Linear Order}.
\newblock Phd thesis, Computer Science Department, University of California at
  Los~Angeles, USA, 1968.

\bibitem{knast83}
R.~Knast.
\newblock A semigroup characterization of dot-depth one languages.
\newblock {\em RAIRO - Theoretical Informatics and Applications},
  17(4):321--330, 1983.

\bibitem{betweenlics}
A.~Krebs, K.~Lodaya, P.~K. Pandya, and H.~Straubing.
\newblock Two-variable logic with a between relation.
\newblock In {\em Proceedings of the 31st Annual {ACM/IEEE} Symposium on Logic
  in Computer Science, {LICS}'16}, pages 106--115, 2016.

\bibitem{betweenconf}
A.~Krebs, K.~Lodaya, P.~K. Pandya, and H.~Straubing.
\newblock An algebraic decision procedure for two-variable logic with a between
  relation.
\newblock In {\em 27th {EACSL} Annual Conference on Computer Science Logic,
  {CSL'18}}, Leibniz International Proceedings in Informatics (LIPIcs). Schloss
  Dagstuhl--Leibniz-Zentrum fuer Informatik, 2018.

\bibitem{between}
A.~Krebs, K.~Lodaya, P.~K. Pandya, and H.~Straubing.
\newblock Two-variable logics with some betweenness relations: Expressiveness,
  satisfiability and membership.
\newblock {\em Logical Methods in Computer Science}, 16(3), 2020.

\bibitem{ksfo2alt}
A.~Krebs and H.~Straubing.
\newblock {An effective characterization of the alternation hierarchy in
  two-variable logic}.
\newblock In {\em Proceedings of the 32sd Annual Conference on Foundations of
  Software Technology and Theoretical Computer Science}, FSTTCS'12, pages
  86--98, 2012.

\bibitem{KufleitnerW15}
M.~Kufleitner and T.~Walter.
\newblock One quantifier alternation in first-order logic with modular
  predicates.
\newblock {\em {RAIRO} Theorerical Informatics and Applications}, 49(1):1--22,
  2015.

\bibitem{kwfo2alt3}
M.~Kufleitner and P.~Weil.
\newblock {The FO2 alternation hierarchy is decidable}.
\newblock In {\em Proceedings of the 21st International Conference on Computer
  Science Logic}, CSL'12, pages 426--439, 2012.

\bibitem{MACIEL2000135}
A.~Maciel, P.~P{\'{e}}ladeau, and D.~Th{\'{e}}rien.
\newblock Programs over semigroups of dot-depth one.
\newblock {\em Theoretical Computer Science}, 245(1):135--148, 2000.

\bibitem{MargolisP85}
S.~W. Margolis and J.~Pin.
\newblock Products of group languages.
\newblock In {\em {FCT}'85}. Springer, 1985.

\bibitem{mnpfosf}
R.~McNaughton and S.~A. Papert.
\newblock {\em Counter-Free Automata}.
\newblock {MIT} Press, 1971.

\bibitem{jep-intersectPOL}
J.-E. Pin.
\newblock An explicit formula for the intersection of two polynomials of
  regular languages.
\newblock In {\em Proceedings of the 17th International Conference on
  Developments in Language Theory, {DLT'13}}, volume 7907 of {\em Lect. Notes
  Comp. Sci.}, pages 31--45. Springer, 2013.

\bibitem{pin-straubing:upper}
J.-{\'{E}}. Pin and H.~Straubing.
\newblock Monoids of upper triangular boolean matrices.
\newblock In {\em Semigroups. Structure and Universal Algebraic Problems},
  volume~39, pages 259--272. North-Holland, 1985.

\bibitem{pwdelta2}
J.-E. Pin and P.~Weil.
\newblock Polynomial closure and unambiguous product.
\newblock {\em Theory of Computing Systems}, 30(4):383--422, 1997.

\bibitem{pmixed}
T.~Place.
\newblock The amazing mixed polynomial closure and its applications to
  two-variable first-order logic.
\newblock In {\em Proceedings of the 37th Annual {ACM/IEEE} Symposium on Logic
  in Computer Science, {{LICS}'22}}, 2022.

\bibitem{pzsucc}
T.~Place and M.~Zeitoun.
\newblock Separation and the successor relation.
\newblock In {\em 32nd International Symposium on Theoretical Aspects of
  Computer Science, {STACS'15}}, Leibniz International Proceedings in
  Informatics (LIPIcs), pages 662--675. Schloss Dagstuhl--Leibniz-Zentrum fuer
  Informatik, 2015.

\bibitem{pzfoj}
T.~Place and M.~Zeitoun.
\newblock Separating regular languages with first-order logic.
\newblock {\em Logical Methods in Computer Science}, 12(1), 2016.

\bibitem{pzcovering2}
T.~Place and M.~Zeitoun.
\newblock The covering problem.
\newblock {\em Logical Methods in Computer Science}, 14(3), 2018.

\bibitem{PZ:generic18}
T.~Place and M.~Zeitoun.
\newblock Generic results for concatenation hierarchies.
\newblock {\em Theory of Computing Systems (ToCS)}, 63(4):849--901, 2019.
\newblock Selected papers from CSR'17.

\bibitem{pzjacm19}
T.~Place and M.~Zeitoun.
\newblock Going higher in first-order quantifier alternation hierarchies on
  words.
\newblock {\em Journal of the {ACM}}, 66(2):12:1--12:65, 2019.

\bibitem{pzconcagroup}
T.~Place and M.~Zeitoun.
\newblock Separation and covering for group based concatenation hierarchies.
\newblock In {\em Proceedings of the 34th Annual {ACM/IEEE} Symposium on Logic
  in Computer Science, {LICS}'19}, pages 1--13. {IEEE} Computer Society, 2019.

\bibitem{PlaceZ20}
T.~Place and M.~Zeitoun.
\newblock Adding successor: {A} transfer theorem for separation and covering.
\newblock {\em {ACM} Transactions on Computational Logic}, 21(2):9:1--9:45,
  2020.

\bibitem{PlaceZ22}
T.~Place and M.~Zeitoun.
\newblock A generic polynomial time approach to separation by first-order logic
  without quantifier alternation.
\newblock In {\em 42nd {IARCS} Annual Conference on Foundations of Software
  Technology and Theoretical Computer Science, {FSTTCS} 2022, December 18-20,
  2022, {IIT} Madras, Chennai, India}, volume 250 of {\em LIPIcs}, pages
  43:1--43:22. Schloss Dagstuhl - Leibniz-Zentrum f{\"{u}}r Informatik, 2022.

\bibitem{pzupol2}
T.~Place and M.~Zeitoun.
\newblock All about unambiguous polynomial closure.
\newblock {\em \mbox{TheoretiCS}}, 2(11):1--74, 2023.

\bibitem{pzgr}
T.~Place and M.~Zeitoun.
\newblock Group separation strikes back.
\newblock In {\em 38th Annual {ACM/IEEE} Symposium on Logic in Computer
  Science, {LICS'23}}, pages 1--13. {IEEE} Computer Society, 2023.

\bibitem{PZ-Level3}
T.~Place and M.~Zeitoun.
\newblock Dot-depth three, return of the {J}-class.
\newblock In P.~Sobocinski, U.~D. Lago, and J.~Esparza, editors, {\em
  Proceedings of the 39th Annual {ACM/IEEE} Symposium on Logic in Computer
  Science, {LICS} 2024, Tallinn, Estonia, July 8-11, 2024}, pages 64:1--64:15.
  {ACM}, 2024.

\bibitem{pztl}
T.~Place and M.~Zeitoun.
\newblock A generic characterization of generalized unary temporal logic and
  two-variable first-order logic.
\newblock In {\em Proceedings of the 32nd EACSL Annual Conference on Computer
  Science Logic}, CSL'24, 2024.

\bibitem{schutzsf}
M.~P. Sch{ü}tzenberger.
\newblock On finite monoids having only trivial subgroups.
\newblock {\em Information and Control}, 8(2):190--194, 1965.

\bibitem{schul}
M.~P. Sch{ü}tzenberger.
\newblock Sur le produit de concat{\'{e}}nation non ambigu.
\newblock {\em Semigroup Forum}, 13:47--75, 1976.

\bibitem{simonthm}
I.~Simon.
\newblock Piecewise testable events.
\newblock In {\em Proceedings of the 2nd GI Conference on Automata Theory and
  Formal Languages}, pages 214--222. Springer, 1975.

\bibitem{StrauConcat}
H.~Straubing.
\newblock A generalization of the sch{\"u}tzenberger product of finite monoids.
\newblock {\em Theoretical Computer Science}, 13(2):137--150, 1981.

\bibitem{StrauVD}
H.~Straubing.
\newblock Finite semigroup varieties of the form {V {\textasteriskcentered} D}.
\newblock {\em Journal of Pure and Applied Algebra}, 36:53--94, 1985.

\bibitem{TheConcat}
D.~Th{\'e}rien.
\newblock Classification of finite monoids: The language approach.
\newblock {\em Theoretical Computer Science}, 14(2):195--208, 1981.

\bibitem{twfo2}
D.~Th{{\'{e}}}rien and T.~Wilke.
\newblock Over words, two variables are as powerful as one quantifier
  alternation.
\newblock In {\em Proceedings of the 30th Annual ACM Symposium on Theory of
  Computing, {STOC'98}}, pages 234--240, New York, NY, USA, 1998. {ACM}.

\bibitem{permauto}
G.~Thierrin.
\newblock Permutation automata.
\newblock {\em Theory of Computing Systems}, 2(1):83--90, 1968.

\bibitem{ThomEqu}
W.~Thomas.
\newblock Classifying regular events in symbolic logic.
\newblock {\em Journal of Computer and System Sciences}, 25(3):360--376, 1982.

\bibitem{Weil_1993}
P.~Weil.
\newblock Some results on the dot-depth hierarchy.
\newblock {\em Semigroup Forum}, 46(1):352--370, 1993.

\bibitem{zalclt72}
Y.~Zalcstein.
\newblock Locally testable languages.
\newblock {\em Journal of Computer and System Sciences}, 6(2):151--167, 1972.

\end{thebibliography}

\newpage
\appendices

\section{Additional preliminaries}
\label{app:appdefs}
We present additional terminology used in proof arguments throughout the appendix. First, we introduce ordered monoids, which are used to handle classes that may not be closed under complement (we need them to deal with the classes \pol{\Cs}). Then, we introduce the notion of \emph{\Cs-morphism} for a \pvari \Cs, which is a key mathematical tool. We also prove Lemma~\ref{lem:nsepconcat}, which is a simple application of these notions.

\subsection{Ordered monoids}

An \emph{ordered monoid} is a pair $(M,\leq)$ where $M$ is a monoid and $\leq$ is a partial order on $M$, compatible with multiplication: for all $s,s',t,t' \in M$, if $s \leq t$ and $s' \leq t'$, then $ss'\leq tt'$. A subset $F \subseteq M$ is an \emph{upper set} (for~$\leq$) if it is upward closed for $\leq$: for all $s,t\in M$, if $s \in F$ and $s \leq t$, then $t \in F$.   By convention, we view every \emph{unordered} monoid $M$ as the ordered one $(M,=)$ whose ordering is \emph{equality}. In this case, \emph{all} subsets of $M$ are upper sets. Thus, every definition that we present for ordered monoids makes sense for unordered ones.

Ordered monoids are used to refine the notion of recognizability by a morphism. We may consider morphisms $\alpha: A^* \to (M,\leq)$ into an arbitrary \emph{ordered} monoid~$(M,\leq)$. We say that a language $L \subseteq A^*$ is \emph{recognized} by such a  morphism~$\alpha$ when there is an \emph{upper set} $F \subseteq M$ for $\leq$ such that $L= \alpha\inv(F)$. 

\begin{restatable}{remark}{convequ}
  Since we view any unordered monoid $M$ as the ordered one $(M,=)$, this definition generalizes the one presented for morphisms into unordered monoids in the body of the paper. A language $L \subseteq A^*$ is recognized by a morphism $\alpha: A^* \to M$ if there exists an \emph{arbitrary} set $F \subseteq M$ (as all subsets are upper sets for ``$=$'') such that $L = \alpha\inv(F)$.
\end{restatable}

\subsection{\Cs-morphisms for \pvaris}

In this section, we define \Cs-morphisms for \pvaris, in order to obtain a result (Lemma~\ref{lem:cpairs} below) that we will use as a tool. Given a \pvari \Cs, a morphism $\eta: A^*\to (N,\leq)$ into a finite ordered monoid is a \emph{\Cs-morphism} if it is \emph{surjective} and all languages recognized by $\eta$ belong to \Cs. The next lemma, proved in~\cite{pzupol2}, states that if \Cs is a \vari (\emph{i.e.},~\Cs is also closed under complement), it suffices to consider \emph{unordered} monoids (\emph{i.e.}, viewed as ordered by equality).

\begin{restatable}{lemma}{cmorphbool} \label{lem:cmorphbool}
	Let \Cs be a \vari and let $\eta: A^* \to (N,\leq)$ be a morphism. The two following conditions are equivalent:
	\begin{enumerate}
		\item $\eta: A^* \to (N,\leq)$ is a \Cs-morphism.
		\item $\eta: A^* \to (N,=)$ is a \Cs-morphism.
	\end{enumerate}
\end{restatable}

\begin{restatable}{remark}{cmbool} \label{rem:cmbool}
  Another formulation is that when \Cs is a \vari, whether a morphism $\eta: A^* \to (N,\leq)$ is a \Cs-morphism does \emph{not} depend on the ordering $\leq$. Thus, when dealing~with \varis, we implicitly use Lemma~\ref{lem:cmorphbool} and drop the~ordering.
\end{restatable}

We have the following key proposition (proved in~\cite[Proposition~2.13]{pzupol2}) that we shall rely on to build \Cs-morphisms.

\begin{restatable}{proposition}{genocm2}\label{prop:genocm}
	Let \Cs be a \pvari and let $L_1,\dots,L_k \in \Cs$. There exists a \Cs-morphism $\eta: A^* \to (N,\leq)$ recognizing $L_1,\dots,L_k$.
\end{restatable}

We shall often use \Cs-morphisms to manipulate the \Cs-pair relation introduced in Section~\ref{sec:prelims}. Let $\alpha: A^* \to M$ and $\eta: A^* \to (N,\leq)$ be morphism. For every pair $(s,t) \in M^2$, we say that $(s,t)$ is an $\eta$-pair (for $\alpha$) when there exist $u,v \in A^*$ such that $\eta(u) \leq \eta(v)$, $\alpha(u) = s$ and $\alpha(v) = t$. The following lemma is proved in~\cite[Lemma~5.11]{pzupol2}.

\begin{restatable}{lemma}{cpairs} \label{lem:cpairs}
	Let \Cs be a \pvari and $\alpha: A^* \to M$ be a morphism. The two following properties hold:
	\begin{itemize}
		\item If $\eta: A^* \to (N,\leq)$ is a \Cs-morphism, then every \Cs-pair for $\alpha$ is an $\eta$-pair for $\alpha$.
		\item There exists a \Cs-morphism $\eta: A^* \to (N,\leq)$ such that the \Cs-pairs for $\alpha$ are exactly the $\eta$-pair for $\alpha$.
	\end{itemize}
\end{restatable}

\subsection{Proof of Lemma~\ref{lem:nsepconcat}}

With Proposition~\ref{prop:genocm} in hand, we are now able to prove Lemma~\ref{lem:nsepconcat}. Let us first recall its statement.

\nsepconcat*

\begin{proof}
  We argue by contradiction: assume that there exists a language $K \in \Cs$ such that $L_1H_1 \subseteq K$ and $K \cap L_2H_2=\emptyset$. Proposition~\ref{prop:genocm} yields a \Cs-morphism $\alpha:A^* \to (M,\leq)$ recognizing~$K$. Since $\alpha$ is a \Cs-morphism, our hypothesis implies that no language recognized by $\alpha$ can separate $L_1$ from $L_2$ or $H_1$ from $H_2$. This yields $u_1 \in L_1$, $u_2 \in L_2$, $v_1 \in H_1$ and $v_2 \in H_2$ such that $\alpha(u_1) \leq \alpha(u_2)$ and $\alpha(v_1) \leq \alpha(v_2)$. It follows that $\alpha(u_1v_1) \leq \alpha(u_2v_2)$. Clearly, $u_1v_1\in L_1H_1 \subseteq K$ and since $K$ is recognized by $\alpha$, it follows that $u_2v_2 \in  K$ as well. This is a contradiction since $u_2v_2 \in L_2H_2$ and $K \cap L_2H_2=\emptyset$.
\end{proof}

\section{Missing proof in Section~\ref{sec:concat}}
\label{app:concat}
We prove Lemma~\ref{lem:polinduc}. Let us first recall the statement.

\polinduc*

\begin{proof}
	By contradiction, we assume that $H_1^*$ is \pol{\Cs}-separable from $H_1^*H_2H_1^*$. This yields a language $K \in \pol{\Cs}$ such that $H_1^*\subseteq K$ and $K\cap H_1^*H_2H_1^* = \emptyset$. Proposition~\ref{prop:genocm} yields a \pol{\Cs}-morphism $\alpha: A^* \to (M,\leq)$ recognizing $K$ (note that we need the ordering here since \pol{\Cs} need not be closed under complement). Moreover, Lemma~\ref{lem:cpairs} yields a \Cs-morphism $\eta: A^* \to (N,\leq)$ such that the \Cs-pairs for $\alpha$ are exactly the $\eta$-pairs. Since  $H_1$ is not \Cs-separable from $H_2$, no language recognized by $\eta$ can be a separator. This yields, $u \in H_1$ and $v \in H_2$ such that $\eta(u) \leq \eta(v)$. Let $s = \alpha(u)$ and $t = \alpha(v)$. By definition of $\eta$, we know that $(s,t) \in M^2$ is a \Cs-pair. Hence, since $\alpha$ is a \pol{\Cs}-morphism, it follows from~\cite[Theorem~54]{PZ:generic18} that $s^{\omega+1} \leq s^\omega t s^\omega$. Let $p = \omega(M)$. Since $u \in H_1$ and $v \in H_2$, we have $u^{p+1} \in H_1^* \subseteq K$ and $u^p v u^p \in H_1^*H_2H_1^*$. Moreover, $\alpha(u^p) =  s^{\omega+1} \leq s^\omega t s^\omega = \alpha(u^p v u^p)$. Hence, since $u^p \in K$ and $K$ is recognized by $\alpha$, we get $u^p v u^p \in K$. This is a contradiction since $u^p v u^p \in H_1^*H_2H_1^*$ and $K\cap H_1^*H_2H_1^* = \emptyset$ by hypothesis.
\end{proof} 

\section{Missing proofs in Section~\ref{sec:utl}}
\label{app:utl}
This appendix is devoted to proving  Lemma~\ref{lem:tlcinduc} and Proposition~\ref{prop:tlvar}. We start with the former whose statement is as follows.

\tlcinduc*

\begin{proof}
  We prove that $P$ is not \tlc{\Cs}-separable from $PVP$ (that $P$ is not \tlc{\Cs}-separable from $PUP$ is proved symmetrically). We fix $K \in \tlc{\Cs}$  such that $P \subseteq K$, we have to prove that $K \cap PVP \neq \emptyset$.  Let $\alpha: A^* \to M$ be the syntactic morphism of $K$. Lemma~\ref{lem:cpairs} yields a \Cs-morphism $\eta: A^* \to N$ such that the \Cs-pairs for $\alpha$ are exactly the $\eta$-pairs for $\alpha$. By hypothesis, $L$ is not \Cs-separable from either $U$ and $V$. Since $\eta$ is a \Cs-morphism, it follows that no language recognized by $\eta$ can be a separator. Hence, there exists $x,x' \in L$, $u \in U$ and $v \in V$ such that $\eta(x) = \eta(u)$ and $\eta(x') = \eta(v)$. Let $p = \omega(M)$, $y = (xx')^p$, $w_1 = ux'(xx')^{p-1}$ and $w_2 = (xx')^{p-1}xv$. Moreover, let $e = \alpha(z)$, $s = \alpha(w_1)$ and $t = \alpha(w_2)$. Clearly, $\eta(z) =  \eta(w_1) = \eta(w_2)$. Hence, $(e,s)$ and $(e,t)$ are \Cs-pairs for $\alpha$ by definition of $\eta$. Moreover, $e \in E(M)$ by definition of $p$. Altogether, since $L \in \tlc{\Cs}$ and $\alpha$ is the syntactic morphism of $L$, Theorem~\ref{thm:tlcar} yields,
  \[
  (esete)^\omega = (esete)^\omega ete(esete)^\omega.
  \]
  Let $z = (yw_1yw_2y)^p$. Clearly, $y \in L^*$, $w_1 \in UL^*$ and $w_2 \in L^*V$. Hence, $z \in P$ and $zw_2z \in PVP$ as $P = L^*(U L^*VL^*)^*$. Since $P \subseteq K$, we get $z \in K$. Moreover, $\alpha(z) = (esete)^\omega$ and $\alpha(zw_2z) = (esete)^\omega ete(esete)^\omega$ by definition. Hence, the above yields $\alpha(z) = \alpha(zw_2z)$ and since $z \in K$ which is recognized by $\alpha$, we obtain $zw_2z \in K$. Altogether, it follows that $zw_2z \in K$ and $zw_2z \in PVP$ which yields $K \cap PVP \neq \emptyset$ as desired.
\end{proof}

We turn to Proposition~\ref{prop:tlvar}. In order to carry out the proof, we first need to adapt Theorem~\ref{thm:tlcar} into a characterization of the \tlc{\Cs}-morphisms rather than of the languages in \tlc{\Cs}. We do so in the following proposition.

\begin{restatable}{proposition}{tlcmor} \label{prop:tlcmor}
  Let \Cs be a \vari and $\alpha: A^* \to M$ be a surjective morphism. Then, $\alpha$ is a \tlc{\Cs}-morphism if and only if it satisfies the following property:
  \begin{equation*} \tag{\ref{eq:tlcar}}
  	\begin{array}{c}
  		\hspace*{-.5ex}(esete)^\omega = (esete)^\omega ete (esete)^\omega \quad \text{for all $e \in E(M)$}\\
  		\hspace*{-.5ex}\text{and $s,t \in M$ such that $(e,s),(e,t)$ are \Cs-pairs for $\alpha$.}
  	\end{array}
  \end{equation*}
\end{restatable}

\begin{proof}
  First, assume that $\alpha$ satisfies~\eqref{eq:tlcar}. We prove that every language recognized by $\alpha$ belongs to \tlc{\Cs}. Since \tlc{\Cs} is closed under union (by definition), it suffices to show that $\alpha\inv(p) \in \tlc{\Cs}$ for all $p \in M$. We fix $p \in M$ for the~proof. Let $\beta: A^* \to Q$ be the syntactic morphism of $\alpha\inv(p)$. By Theorem~\ref{thm:tlcar}, it suffices to show that $\beta$ satisfies~\eqref{eq:tlcar}. Let $f \in E(Q)$ and $q,r \in Q$ such that $(f,q),(f,r)$ are \Cs-pairs for~$\beta$. We have to show that,
  \[
  (fqfrf)^\omega = (fqfrf)^\omega frf (fqfrf)^\omega.
  \]
  Lemma~\ref{lem:cpairs} yields a \Cs-morphism $\eta: A^*\to N$ such that the \Cs-pairs for $\alpha$ are exactly the $\eta$-pairs for $\alpha$. Since $(f,q)$ and $f,r$ are \Cs-pairs for~$\beta$ and $\beta$ is a \Cs-morphism, we obtain $u_f,u_q,v_f,v_r\in A^*$ such that $\eta(u_f)=\eta(u_q)$, $\eta(v_f) = \eta(v_r)$, $\beta(u_f) = \beta(v_f) = f$, $\beta(u_q) = q$ and $\beta(v_r) = r$. Let $k = \omega(M) \times \omega(Q)$. We also let $w_f = (u_fv_f)^k$, $w_q = u_qv_f(u_fv_f)^{k-1}$ and $w_r = (u_fv_f)^{k-1} u_fv_r$. Clearly, $\beta(w_f) = f$, $\beta(w_q) = qf$ and $\beta(w_r) = fr$. Moreover, $\eta(w_f) = \eta(w_q) = \eta(w_r)$. Let $e = \alpha(w_f)$, $s = \alpha(w_q)$ and $t = \alpha(w_r)$. Since $\eta(w_f) = \eta(w_q) = \eta(w_r)$, the definition of $\eta$ yields that $(e,s)$ and $(e,t)$ are \Cs-pairs for $\alpha$. Hence, since $e \in E(M)$ (by definition of $k$) and $\alpha$ satisfies~\eqref{eq:tlcar}, we get,
  \[
  (esete)^\omega = (esete)^\omega ete (esete)^\omega.
  \]
  Let $z = w_fw_qw_fw_rw_f$. For all $x,y \in A^*$, the above yields,
  \[
	\alpha(x z^k y) = \alpha(x z^k w_fw_rw_fz^k y).
  \]
  In particular, $x z^k y\in\alpha\inv(p)\Leftrightarrow x z^k w_fw_rw_fz^k y\in\alpha\inv(p)$. Therefore, $z^k$  and $z^k w_fw_rw_fz^k$ have the same image under the syntactic morphism $\beta$ of $\alpha\inv(p)$. Since $k$ is a multiple of $\omega(Q)$, this yields $(fqfrf)^\omega = (fqfrf)^\omega frf (fqfrf)^\omega$, as desired.

  \medskip

  We now prove the converse implication. Assume that $\alpha$ is a \tlc{\Cs}-morphism. Fix an idempotent $e \in E(M)$ and  $s,t \in M$ such that $(e,s)$ and $(e,t)$ are \Cs-pairs for~$\alpha$. We prove that $(esete)^\omega = (esete)^\omega ete (esete)^\omega$. Let $p = (esete)^\omega \in E(M)$. By hypothesis, $\alpha\inv(p) \in \tlc{\Cs}$. Let $\beta: A^* \to Q$ be the syntactic morphism of $\alpha\inv(p)$. By Theorem~\ref{thm:tlcar}, we know that $\beta$ satisfies~\eqref{eq:tlcar}. Moreover, Lemma~\ref{lem:cpairs} yields a \Cs-morphism $\eta: A^*\to N$ such that the \Cs-pairs for $\beta$ are exactly the $\eta$-pairs for $\beta$. Since $(e,s)$ and $(e,t)$ are \Cs-pairs for $\alpha$, Lemma~\ref{lem:cpairs} implies that they are also $\eta$-pairs for $\alpha$. We get $u_e,u_s,v_e,v_t\in A^*$ such that $\eta(u_e)=\eta(u_s)$, $\eta(v_e) = \eta(v_t)$, $\alpha(u_e) = \alpha(v_e) = e$, $\alpha(u_s) = s$ and $\alpha(v_t) = t$. Let $h = \omega(Q) \times \omega(M)$. We define $w_e = (u_ev_e)^h$, $w_s=u_sv_e(u_ev_e)^{h-1}$ and $w_t=(u_ev_e)^{k-1}u_ev_t$. Clearly, $\alpha(w_e) = e$, $\alpha(w_s) = se$ and $\alpha(w_t) = et$. Moreover, $\eta(w_e) = \eta(w_s) = \eta(w_t)$. Let $f = \alpha(w_e)$, $q = \alpha(w_s)$ and $r = \alpha(w_t)$. Since $\eta(w_e) = \eta(w_s) = \eta(w_t)$, $(f,q)$ and $(f,r)$ are \Cs-pairs for $\beta$ by definition of $\eta$. Hence, since $f \in E(Q)$ and $\beta$ satisfies~\eqref{eq:tlcar}, we get $(fqfrf)^\omega = (fqfrf)^\omega frf (fqfrf)^\omega$. Thus, we have the following for $z = w_ew_sw_ew_tw_e$,
  \[
	\beta(z^h) = \beta(z^h w_ew_tw_ez^h).
  \]
  Since $h$ is a multiple of $\omega(M)$, $\alpha(z^h) = (esete)^\omega = p$. Hence, $z^h \in \alpha\inv(p)$ and since $\alpha\inv(p)$ is recognized by $\beta$, the above equality yields $z^h w_ew_tw_ez^h \in \alpha\inv(p)$. Thus, $(esete)^\omega ete (esete)^\omega = p = (esete)^\omega$ as desired.
\end{proof}

We may now prove Proposition~\ref{prop:tlvar}. We rely on Theorem~\ref{thm:tlcar} and Proposition~\ref{prop:tlcmor} above. We first recall the statement.

\tlvar*

\begin{proof}
  It suffices to prove that if \Cs is a \vari, then so is \tlc{\Cs}. The result will then follow by induction. We already know from the definition that \tlc{\Cs} is a Boolean algebra. We have to prove closure under quotients. Let $L\in \tlc{\Cs}$ and $u \in L$. We have to prove that $Lu\inv \in \tlc{\Cs}$ and $u\inv L \in\tlc{\Cs}$.  Let $\alpha: A^* \to M$ be the syntactic morphism of $L$ and $F \subseteq M$ be the accepting set such that $L = \alpha\inv(F)$. One can verify from the definitions that,
  \[
  \begin{array}{lll}
  	Lu\inv &= &\alpha\inv\left( \{s \in M \mid s\alpha(u) \in F\}\right), \\
  	u\inv L &=& \alpha\inv\left( \{s \in M \mid \alpha(u)s \in F\}\right).
  \end{array}
  \]
  Hence, $Lu\inv$ and $u\inv L$ are both recognized by $\alpha$. Since $L \in \tlc{\Cs}$, it follows from Theorem~\ref{thm:tlcar} that $\alpha$ satisfies~\eqref{eq:tlcar}. By Proposition~\ref{prop:tlcmor}, this implies that $\alpha$ is a \tlc{\Cs}-morphism. We obtain $Lu\inv \in \tlc{\Cs}$ and $u\inv L \in\tlc{\Cs}$, as desired.
\end{proof}

\section{Missing proofs in Section~\ref{sec:groups}}
\label{app:groups}
We present the proofs of the statements of Section~\ref{sec:groups}. We start with the strictness proofs.

\subsection{Strictness}

Let us first recall the definitions leading up to the statements. We use two families of languages $(K_n)_{n\in \nat}$ and $(L_n)_{n \in \nat}$ to prove that the \nesths of bases \Gs and $\Gs^+$ are \emph{strict} and \emph{strictly intertwined}. Roughly, we replace the letters ``$a$'' and ``$b$'' in the definition of $H_n$ by languages which are not \grp-separable from $\{\veps\}$.

Let $x_i = ab^i$ and $Y_i = a^+b^i$ for all $i \geq 1$.  Moreover, let $Q =  x_1^+x_2^+x^{}_1$, $R = x_3^+x_4^+x^{}_3$, $S = Y_1^+Y_2^+Y^{}_1$ and $T = Y_3^+Y_4^+Y^{}_3$. Finally, we define,
\[
\begin{array}{lcl}
	K_0  = \{\veps\} & \text{and}& K_{n}  = (QK_{n-1}R)^*  \text{ for $n \geq 1$},\\
	L_0  =   a^* &\text{and} & L_{n}  =  (a+ SL_{n-1}T)^*  \text{ for $n \geq 1$}.
\end{array}
\]

We may now prove the two lemmas involved in the strictness proof.

\strictpos*

\begin{proof}
  We actually prove that $K_n \in \cocl{\polp{n}{\dotzer}}$ and $L_n \in \cocl{\polp{n+1}{\stzer}}$ for all $n \in \nat$. This yields the desired result. Both properties are proved by induction on $n$. The arguments are similar. When $n=0$, $K_0 = \{\veps\} \in \dotzer$. Also, $L_0 = a^* = A^* \setminus \bigl(\bigcup_{b \in A \setminus\{a\}} A^*bA^*\bigr) \in \copol{\stzer}$. We now assume that $n \geq 1$.

  By definition of the mapping $\Ds \mapsto \cocl{\Ds}$, it suffices to prove that $A^* \setminus K_n \in \polp{n}{\dotzer}$ and $A^* \setminus L_n \in \polp{n+1}{\stzer}$. We decompose these languages as the union of simpler languages. We write $H = aA^* \cup \{\veps\}$. One can verify that,
  \begin{equation} \label{eq:dyck}
	\begin{array}{llll}
	  A^* \setminus K_n & = & & A^* \setminus(Q +  R)^* \\
						&  & \cup &  A^*x_4x_3 (R K_{n-1})^{n-1} x_3 H \\
						&& \cup &  K_{n-1}x_3 H \\
						& & \cup &  A^* x_2x_1 (K_{n-1} Q)^{n-1} x_1 H \\
						&& \cup &  A^*x_2x_1 K_{n-1}, \\[1ex]
	  A^* \setminus	L_n & = & &A^* \setminus (a+ S +  T)^* \\
						&  & \cup &  A^*Y_4Y_3 (T L_{n-1})^{n-1} Y_3 H \\
						&& \cup &  L_{n-1}Y_3 H \\
						& & \cup &  A^* Y_2Y_1 (L_{n-1} S)^{n-1} Y_1 H \\
						&& \cup &  A^*Y_2Y_1 L_{n-1}.
	\end{array}
  \end{equation}
  By definition,  \polp{n}{\dotzer} and \polp{n+1}{\stzer} are closed under union. Hence, it suffices to prove that all the languages in these unions belong to \polp{n}{\dotzer} and \polp{n+1}{\stzer} respectively.  We first show that $(Q +  R)^* \in \cocl{\pol{\dotzer}}$ and $(a + S +  T)^* \in \cocl{\polp{2}{\stzer}}$.

  Let first $U = (Q +  R)^*$. Since $Q = x_1^+x_2^+x_1$ and $R = x_3^+x_4^+x_3$, one can verify that membership in $U$ is characterized by a finite number of forbidden patterns, namely,
  \[
	\begin{array}{lllll}
	  A^*\setminus U  & = & & A^* \setminus (x_1+x_2+x_3 +x_4)^* \\
		 & & \cup & \{x_1\} \cup \{x_3\}\\
		 & & \cup &  x_2H \cup x_4H \cup A^*x_2 \cup A^* x_4\\
		 & & \cup & x_1x_3H \cup x_3x_1H  \cup  A^*x_1x_3 \cup A^* x_3x_1\\
		 & & \cup & A^* x_1x_4  H \cup A^* x_4x_1H \\
		 & & \cup & A^* x_2x_3  H \cup A^* x_3x_2 H \\
		 & & \cup &	A^* x_4x_2H \cup A^* x_2x_4H \\
		 & & \cup & A^*x_1x_1x_3H \cup A^*x_3x_1x_3H \\
		 & & \cup & A^*x_3x_3x_1H \cup A^*x_1x_3x_1H \\
		 & & \cup & A^*x_2x_1x_2H \cup A^*x_4x_3x_4H.
	\end{array}
  \]
  Observe that by definition of $x_1,x_2,x_3$ and $x_4$, we have,
  \[
	A^* \setminus (x_1\!+\!x_2\!+\!x_3\! +\!x_4)^* \!=\! bA^*\cup A^*a \cup A^*a^2A^* \cup A^*b^5A^*.
  \]
  This is a language in \pol{\dotzer}. Moreover, since it is clear that $x_1,x_2,x_3,x_4,H \in \pol{\dotzer}$, we get $A^*\setminus U  \in \pol{\dotzer}$ by closure under concatenation, and we conclude therefore that $U=(Q+R)^*\in \copol{\dotzer}$, as announced.

  Similarly, let $V=(a+S+T)^*$. Since $S = Y_1^+Y_2^+Y_1$ and $T = Y_3^+Y_4^+Y_3$, one can verify that,
  \[
	\begin{array}{lllll}
	  A^*\setminus V  & = & & A^* \setminus (a + Y_1+Y_2+Y_3 +Y_4)^*  \\
		 & & \cup & Y_1 \cup Y_3\\
		 & & \cup & Y_2H \cup Y_4H \cup A^*Y_2 \cup A^* Y_4 \\
		 & & \cup & Y_1Y_3H \cup Y_3Y_1H \cup A^*Y_1Y_3 \cup A^* Y_3Y_1\\
		 & & \cup & A^* Y_1Y_4  H \cup A^* Y_4Y_1H \\
		 & & \cup & A^* Y_2Y_3  H \cup A^* Y_3Y_2 H \\
		 & & \cup & A^* Y_4Y_2H \cup A^* Y_2Y_4H \\
		 & & \cup & A^*Y_1Y_1Y_3H \cup A^*Y_3Y_1Y_3H \\
		 & & \cup & A^*Y_3Y_3Y_1H \cup A^*Y_1Y_3Y_1H \\
		 & & \cup & A^*Y_2Y_1Y_2H \cup A^*Y_4Y_3Y_4H.
	\end{array}
  \]
  Observe that,
  \[
	A^* \setminus (a+ Y_1+ Y_2+Y_3 +Y_4)^*  = bA^*  \cup A^*b^5A^*,
  \]
  which is clearly a language in \polp{2}{\stzer}.
  Since it is clear that $Y_1,Y_2,Y_3,Y_4 \in \polp{2}{\stzer}$ (this is because $a^* \in \bpol{\stzer}$, as seen above), we get $A^*\setminus V \in \polp{2}{\stzer}$ by closure under concatenation, and we conclude therefore that $V=(S+T+a)^*\in\cocl{\polp{2}{\stzer}}$, as announced.

  It remains to handle the other languages in~\eqref{eq:dyck}. This is where we use induction. Let us start with $A^* \setminus K_n$.  Induction yields $K_{n-1} \in \copol{\polp{n-1}{\dotzer}}$. Hence, we obtain $K_{n-1}x_3 H \in \polp{n}{\dotzer}$ and $A^*x_2x_1 K_{n-1} \in \polp{n}{\dotzer}$. We now consider the languages $A^*x_4x_3 (R K_{n-1})^{n-1} x_3 H$ and $A^* x_2x_1 (K_{n-1} Q)^{n-1} x_1 H$. There are two cases. When $n = 1$, these are the languages $A^*x_4x_3x_3 H$ and $A^* x_2x_1x_1 H$: they clearly belong to \pol{\dotzer}. Assume now that $n \geq 2$. Since $Q = x_1^+x_2^+x_1$ and $R = x_3^+x_4^+x_3$, it suffices to prove that $x_1^+,x_2^+, x_3^+,x_4^+ \in \bpol{\dotzer}$. By closure under concatenation, this will imply that $A^*x_4x_3 (R K_{n-1})^{n-1} x_3 H$ and $A^* x_2x_1 (K_{n-1} Q)^{n-1} x_1 H$ both belong to $\polp{n}{\dotzer}$. By definition, $x_k^+ = (ab^k)^+$ for all $k \geq 1$. Hence,
  \[
	x_k^+  =   A^* \setminus \Bigl(\{\veps\} \cup bA^*\cup A^*b^{k+1}A^* \cup \bigcup_{0 \leq i < k}A^*ab^iH\Bigr).
  \]
  Thus, $x_k^+ \in \copol{\dotzer}$ as desired.

  It remains to handle $A^* \setminus L_n$. By induction, we have $L_{n-1} \in \cocl{\polp{n}{\stzer}}$. Hence, $L_{n-1}Y_3 H \in \polp{n+1}{\stzer}$ and $A^*Y_2Y_1 K_{n-1} \in \polp{n}{\stzer}$. We now consider the languages $A^*Y_4Y_3 (T K_{n-1})^{n-1} Y_3 H$ and $A^* Y_2Y_1 (K_{n-1} S)^{n-1} Y_1 H$. If $n = 1$, these are the languages $A^*Y_4Y_3Y_3 H$ and $A^* Y_2Y_1Y_1 H$ which clearly belong to \polp{2}{\stzer}. Assume now that $n \geq 2$. Since $S = Y_1^+Y_2^+Y_1$ and $T = Y_3^+Y_4^+Y_3$, it suffices to prove that $Y_1^+,Y_2^+, Y_3^+,Y_4^+ \in \bpolp{2}{\stzer}$. By closure under concatenation, this will imply that $A^*Y_4Y_3 (T K_{n-1})^{n-1} Y_3 H$ and $A^* Y_2Y_1 (K_{n-1} S)^{n-1} Y_1 H$ both belong to \polp{n+1}{\stzer}. By definition $Y_k^+ =(a^+b^k)^+$ for all $k \geq 1$. Hence,
  \[
	Y_k^+  = A^* \setminus \Bigl(\{\veps\} \cup bA^* \cup A^*b^{k+1}A^* \cup \bigcup_{1 \leq i < k} (A^*ab^iaA^*)\Bigr).
  \]
  This completes the proof.
\end{proof}

The second lemma is based on Lemma~\ref{lem:tlcinduc}.

\strictneg*

\begin{proof}
  We first show that $K_n \not\in \tlhn{\Gs}$ for all $n \in \nat$. Let $Q_0 = Q$ and $R_0 = R$. For $n \geq 1$, let $Q_n = K_nQ_{n-1}K_n$ and $R_n = K_n R_{n-1} K_n$. We have the following simple fact.
  
  \begin{fact} \label{fct:tlcinduc1}
	For all $n \in \nat$, $K_n \cap Q_n = \emptyset$ and $K_n \cap R_n = \emptyset$.
  \end{fact}

  \begin{proof}
	By definition, each word $w \in K_n \cup Q_n \cup R_n$ admits a \emph{unique} decomposition $w= w_1 \cdots w_\ell$ such that each factor $w_i$ belongs to either $Q$ or $R$ for each $i \leq \ell$. Moreover, if $w \in K_n$, then one can verify that there are as many factors $w_i \in Q$ as there are factors $w_j \in R$. On the other hand, if $w \in Q_n$, then one can verify that the number of factors $w_i \in Q$ is equal to the number of factors $w_j \in R$ \emph{plus one}. Similarly, if  $w \in R_n$, then one can verify that the number of factors $w_j \in R$ is equal to the number of factors $w_i \in Q$ \emph{plus one}. The fact follows.
  \end{proof}

  We now use induction on $n$ to show that $K_n$ is \emph{not} \tlhn{\grp}-separable from either $Q_n$ and $R_n$. This yields $K_n \not\in \tlhn{\grp}$ by Fact~\ref{fct:tlcinduc1}.

	When $n = 0$, we have to prove that $\{\veps\}$ is  \emph{not} \grp-separable from either $Q = x_1^+x_2^+x^{}_1$ and $R = x_3^+x_4^+x^{}_3$. By symmetry, we only prove the former. Let $U \in \grp$ such that $\veps \in U$. We show that $U \cap x_1^+ x_2^+ x^{}_1 \neq \emptyset$. By hypothesis, $U$ is recognized by a morphism $\alpha: A^*\to G$ into a finite group $G$. Let $p=\omega(G)$ and $w=x_1^{p-1} x_2^p x_1 \in x_1^+x_2^+x^{}_1$. Since $G$ is a group, $\alpha(w) = 1_G = \alpha(\veps)$. Therefore, $w \in U \cap x_1^+ x_2^+ x^{}_1$.

	Assume that $n \geq 1$. By induction, $K_{n-1}$ is  \emph{not} \tlhp{n-1}{\grp}-separable from either $Q_{n-1}$ or $R_{n-1}$.  One can check from the~definition that $K_{n-1}^* = K_{n-1}  \subseteq K_n$. By definition of $K_n$, this yields $K_n =  K_{n-1}^*(Q^{}_{n-1}K_{n-1}^*R^{}_{n-1}K_{n-1}^*)^*$. Thus, Lemma~\ref{lem:tlcinduc} implies that $K_n$ is not \tlhn{\grp}-separable from either $Q_n = K_nQ_{n-1}K_n$ or $R_n = K_nR_{n-1}K_n$.

  \medskip

  We now prove that $L_n \not\in \tlhn{\Gs^+}$ for all $n \in \nat$. Let $S_0 = S$ and $T_0 = T$. For $n \geq 1$, let $S_n = L_nS_{n-1}L_n$ and $T_n = L_n T_{n-1} L_n$. We have the following simple fact.

  \begin{fact} \label{fct:tlcinduc2}
	For all $n \in \nat$, $L_n \cap S_n = \emptyset$ and $L_n \cap T_n = \emptyset$.
  \end{fact}

  \begin{proof}
	By definition, each word $w \in L_n \cup S_n \cup T_n$ that contains at least one letter ``$b$'' admits a \emph{unique} decomposition $w= w_1 \cdots w_\ell$ such that each factor $w_i$ belongs to either $S$ or $T$ for each $i \leq \ell$. Moreover, if $w \in L_n$, then one can verify that there are as many factors $w_i \in S$ as there are factors $w_j \in T$. On the other hand, if $w \in S_n$, then one can verify that the number of factors $w_i \in S$ is equal to the number of factors $w_j \in T$ \emph{plus one}. Similarly, if  $w \in T_n$, then one can verify that the number of factors $w_j \in T$ is equal to the number of factors $w_i \in S$ \emph{plus one}. The fact follows.
  \end{proof}

  We now use induction on $n$ to show that $L_n$ is \emph{not} \tlhn{\grp^+}-separable from either $S_n$ and $T_n$. This yields $L_n \not\in \tlhn{\grp^+}$ by Fact~\ref{fct:tlcinduc2}.
  
  	When $n = 0$, we have to prove that $a^*$ is  \emph{not} \grp-separable from either $S = Y_1^+Y_2^+Y^{}_1$ and $T = Y_3^+Y_4^+Y^{}_3$. By symmetry, we only prove the former. Let $U \in \grp^+$ such that $a^* \subseteq \in U$. We show that $U \cap Y_1^+ Y_2^+ Y^{}_1 \neq \emptyset$. Since $U \in \grp^+$, there exists $V \in \grp$ such that $U = V \cap A^+$ or $U = V\cup \{\veps\}$. By definition of \grp, $V$ is recognized by a morphism $\alpha: A^*\to G$ into a finite group $G$. Let $p=\omega(G)$ and $w = x_1^{p-1} x_j2^p x_1\in x_2^+x_1^+x^{}_2 \subseteq  Y_1^+ Y_2^+ Y^{}_1$. It is clear that $a^p \in U \cap A^+$. Hence, $a^p \in V$ (since $V$ and~$U$ agree on nonempty words). Moreover, since $G$ is a group, $\alpha(w) = 1_G = \alpha(a^p)$. Since $V$ is recognized by $\alpha$, we get $w \in V$. Since $w \in A^+$, this yields $w \in U$ by definition of $V$, and we get  $U \cap Y_i^+ Y_j^+ Y^{}_i \neq\emptyset$, completing the~proof.
  	
The inductive step is based on Lemma~\ref{lem:tlcinduc}. Assume that $n \geq 1$. By induction hypothesis, $L_{n-1}$ is  \emph{not} \tlhp{n-1}{\grp^+}-separable from either $S_{n-1}$ or $T_{n-1}$.  One can check from the~definition that $L_{n-1}^* = L_{n-1}  \subseteq L_n$. By definition of $L_n$, this implies that $L_n =  L_{n-1}^*(S_{n-1}L_{n-1}^*S_{n-1}L_{n-1}^*)^*$. Hence, it follows from Lemma~\ref{lem:tlcinduc} that $L_n$ is not \tlhn{\grp^+}-separable from either $S_n = L_nS_{n-1}L_n$ or $T_n = L_nT_{n-1}L_n$.
\end{proof}

\subsection{Comparison with concatenation hierarchies}

We first prove Theorem~\ref{thm:tlstinc}. We recall the statement below.

\thmtlst*

\begin{proof}
	Recall that $\tlhp{2}{\stzer} = \tlc{\at}$ by Proposition~\ref{prop:leveltwo}. Moreover, it is standard that
	$\bpolp2\stzer=\bpol\at$ (see~\cite{pin-straubing:upper,PZ:generic18}). Hence, we have to prove that,
	\[
	\ipolp2\at \subseteq \tlc{\at} \subseteq \ipolp{|A|+1}{\stzer}.
	\]
	Let us start with $\ipolp2\at \subseteq \tlc{\at}$. Let $L \in \ipolp2\at$ and let $\alpha: A^* \to M$ be the syntactic morphism of $L$. In view of Theorem~\ref{thm:tlcar}, it suffices to prove that for every $e \in E(M)$ and every $s,t \in M$ such that $(e,s)$ and $(e,t)$ are \at-pairs, we have $(esete)^\omega = (esete)^\omega ete (esete)^\omega$. Since $L \in \ipolp2\at$, Theorem~\ref{thm:capolbp} yields $(esete)^{\omega+1} = (esete)^\omega ete (esete)^\omega$. Moreover, $L$ is \emph{star-free} since $L \in \ipolp2\at$. Thus, Schützenberger's theorem~\cite{schutzsf} implies that its syntactic monoid is aperiodic: we have $q^{\omega+1} = q^{\omega}$ for every $q \in M$.  In particular, we have $(esete)^\omega = (esete)^{\omega+1}$. Altogether, we get $(esete)^\omega = (esete)^\omega ete (esete)^\omega$, as desired.
	
	\medskip
	
	The second inclusion is harder.  Let $L \in \tlc{\at}$, we prove that $L \in \ipolp{|A|+1}{\stzer}$. We proceed by induction on the size of~$A$. We write $A = \{a_1,\dots,a_{n}\}$ for the proof. The base case and the inductive step are handled simultaneously. Let $\Ds = \ipolp{n-1}\stzer$. One can check that,
	\[
	\ipolp{n}{\stzer} = \ipol{\bpol{\Ds}} = \ipolp{2}{\Ds}.
	\]
	It remains to prove that $L \in \ipolp{2}{\Ds}$. We use Theorem~\ref{thm:capolbp}. Let $\alpha: A^* \to M$ be the syntactic morphism of $L$. Given a \Ds-pair $(e,s)\in E(M) \times M$ is a \Ds-pair and $t\in M$, we have to show  $(esete)^{\omega+1} = (esete)^{\omega} ete (esete)^{\omega}$. Let $H = \bigcap_{a \in A} A^*aA^*$ (in other words $H =  \{w \in A^* \mid \etaat(w)= A\}$). Moreover, let $A_i = A \setminus \{a_i\}$ for every $i \leq n-1$. Clearly, we have,
	\[
	A^* = H \cup \bigcup_{i \leq n-1} A_i^*.
	\]
	By definition of \Ds-pairs, $\alpha\inv(e)$ is not \Ds-separable from $\alpha\inv(s)$. Hence, since $H \in \Ds$ and $A_i^* \in \Ds$ for all $i \leq n-1$ (this is because $\bpol{\stzer} \subseteq \Ds$ by definition), one can verify that one of the two following conditions holds:
	\begin{enumerate}
		\item $H \cap \alpha\inv(e)$ is not \Ds-separable from $H \cap \alpha\inv(s)$,
		\item there exists $i \leq n-1$ such that $A_i^* \cap \alpha\inv(e)$ is not \Ds-separable from $A_i^* \cap \alpha\inv(s)$.
	\end{enumerate}
	We now consider two cases depending on which property holds. Assume first that $H \cap \alpha\inv(e)$ is not \Ds-separable from $H \cap \alpha\inv(s)$. In particular, this means that both languages are nonempty. We get $u \in H \cap \alpha\inv(e)$ and $v \in H \cap \alpha\inv(s)$. By definition of $H$, we have $\etaat(u) = \etaat(v) = A$. Moreover, let $x \in \alpha\inv(t)$ and $w = uxu$. Clearly, $\etaat(w) = A$ and $\alpha(w) = ete$. It follows that $\etaat(u) = \etaat(v) = \etaat(w) = A$, so that $\{u\}$ is neither \at-separable from $\{v\}$ nor $\{w\}$. Therefore, $(\alpha(u),\alpha(v)) = (e,s)$ and $(\alpha(u),\alpha(w)) = (e,ete)$ are \at-pairs.  Since $L \in \tlc{\at}$, it then follows from Theorem~\ref{thm:tlcar} that $(esete)^\omega = (esete)^\omega ete (esete)^\omega$. Moreover, since  $L \in \tlc{\at}$, we also know that $L$ is \emph{star-free}. Thus, Schützenberger's theorem~\cite{schutzsf} implies that its syntactic monoid is aperiodic: we have $q^{\omega+1} = q^{\omega}$ for every $q \in M$.  In particular, $(esete)^\omega = (esete)^{\omega+1}$. Altogether, we get,
	\[
	(esete)^{\omega+1} = (esete)^\omega ete (esete)^\omega.
	\]
	This completes the first case.

	It remains to handle the case where there exists $i \leq n-1$ such that $A_i^* \cap \alpha\inv(e)$ is not \Ds-separable from $A_i^* \cap \alpha\inv(s)$.  Observe that we may assume without loss of generality that $i = n-1$ (otherwise, we rename the letters). In this case, we prove that $e = ese$, which implies the desired equality $(esete)^{\omega+1} = (esete)^{\omega} ete(esete)^{\omega}$. We distinguish the cases $|A|=2$ (\emph{i.e.}, $n=3$) and~$|A| > 2$ (\emph{i.e.}, $n>3$).
	
	Let us start with the case $|A| = 2$. By definition, $A_2 = \{a_1\}$ is a single letter alphabet in this case. Thus, $a_1^* \cap \alpha\inv(e)$ is not \Ds-separable from $a_1^* \cap \alpha\inv(s)$. In particular, $a_1^* \cap \alpha\inv(e) \neq \emptyset$ and $a_1^* \cap \alpha\inv(s) \neq \emptyset$. There are two sub-cases depending on whether $a_1^* \cap \alpha\inv(e) = \{\veps\}$ or $a_1^+ \cap \alpha\inv(e) \neq \emptyset$. In the former case, since $\{\veps\} \in \Ds$ (this is because $\bpol{\stzer} \subseteq \Ds$), the hypothesis that $a_1^* \cap \alpha\inv(e)$ is not \Ds-separable from $a_1^* \cap \alpha\inv(s)$ yields $\veps \in a_1^* \cap \alpha\inv(s)$ and we get $e = s = 1_M$, whence $ese = e$, as desired. Otherwise, we have  $a_1^+ \cap \alpha\inv(e) \neq \emptyset$. This yields $k \geq 1$ such that $\alpha(a_1^k) = e$. Moreover, since $a_1^* \cap \alpha\inv(s) \neq \emptyset$, we have $\ell \in \nat$ such that $\alpha(a_1^\ell) = s$. Let $r = \alpha(a_1)$. We have $e = r^k$ and since $e$ is idempotent, it follows that $e = r^{\omega}$. Moreover, $ese = r^{\omega+\ell}$. Thus, since $M$ is aperiodic by Schützenberger's theorem~\cite{schutzsf}, we obtain again~$e = ese$.
	
	We turn to the case where $|A| > 2$. This is where we use induction on the size $n-1$ of $A$. Since $L \in \tlc{\at}$ and \tlc\at is a \vari by
	Proposition~\ref{prop:tlvar}, it is standard that all languages recognized by its syntactic morphism belong to \tlc{\at}. In particular, $\alpha\inv(e) \in \tlc{\at}$ and $\alpha\inv(s) \in \tlc{\at}$. Since it is clear that $A_{n-1}^* \in \tlc{\at}$ as well, we get $A_{n-1}^* \cap \alpha\inv(e) \in \tlc{\at}$ and $A_{n-1}^* \cap \alpha\inv(s) \in \tlc{\at}$. These two languages are defined over the alphabet $A_{n-1}$ of size $|A|-1$ and induction yields that over this restricted alphabet, \tlc{\at} is included in $\Ds = \ipolp{n-1}\stzer$. Consequently, the hypothesis that $A_{n-1}^* \cap \alpha\inv(e)$ is not \Ds-separable from $A_{n-1}^* \cap \alpha\inv(s)$ implies that  $A_{n-1}^* \cap \alpha\inv(e)$ is not \tlc{\at}-separable from $A_{n-1}^* \cap \alpha\inv(s)$. Since both languages belong to \tlc{\at}, it follows that they intersect. In particular, $\alpha\inv(e) \cap \alpha\inv(s) \neq \emptyset$ which yields $s = e$, concluding the proof.
\end{proof}

We now turn to the statements leading up to Theorem~\ref{thm:gualt}. Let us first recall the objects involved in these statements.  For each $k \in \nat$, we have an alphabet $A_k = \{\ell_0,\dots,\ell_k\}$ of size $k+1$ and two languages $U_k,V_k\subseteq A_k^*$. We have $U_0 = \{\veps\}$ and $V_0 = A_0^+$.  For $k \geq 1$, we have,
\[
U_k = (\ell_kV_{k-1})^* \quad \text{and} \quad V_k = (\ell_kV_{k-1})^* \ell_kU_{k-1} (\ell_kV_{k-1})^*.
\]
Moreover, a binary alphabet $A = \{a,b\}$ is fixed and we encode the languages $U_k,V_k$ over $A$. For every $k \in \nat$, a morphism $\beta_k: A_k^* \to A^*$ is defined as follows:
\[
\begin{array}{llll}
	\beta_k: & A_k^* & \to     & A^*, \\
	& \ell_i & \mapsto & a^i b a^{k-i}.
\end{array}
\]
We first prove Proposition~\ref{prop:modproj}. It states that this encoding maps \tlhp2\stzer to \tlhp2\md. Combined with Lemma~\ref{lem:langtl2st}, it implies that $\beta_k(U_k) \in \tlhp{2}{\md}$ for all $k \in \nat$.

\newcommand{\thecons}[1]{\ensuremath{\left\langle #1 \right\rangle}\xspace}

\modproj*

\begin{proof}
	We first show that $\beta^{}_k(A_k^*) \in \tlhp{2}{\md}$. Observe that since $(A^{k+1})^* \in \md$, the following formula $\Gamma$ belongs to \tla{\tlc{\md}}:
	\[
	\Gamma = min \vee \finallymp{(A^{k+1})^*}{min}.
	\]
	Clearly, given a word $x \in A^*$ and $j \in \pos{x}$, we have $x,j \models \Gamma$ if and only if $j = 0$ or $j \equiv 1 \bmod k+1$. We now associate a formula $\lambda_n \in \tla{\tlc{\md}}$ to every number $n \leq k$. Let us first present a~preliminary notation. If $\varphi \in \tla{\tlc{\md}}$ and $i \in \nat$, we write $\nexi{\varphi}{i}$ and $\nexmi{\varphi}{i}$ for the following formula: when $i = 0$, $\nexi{\varphi}{i} = \nexmi{\varphi}{i} = \varphi$ and when $i \geq 1$, $\nexi{\varphi}{i} = \finallyp{\{\veps\}}{\left(\nexi{\varphi}{i-1}\right)}$ and $\nexmi{\varphi}{i} = \finallymp{\{\veps\}}{\left(\nexmi{\varphi}{i-1}\right)}$. Observe that $\nexi{\varphi}{i}$ and $\nexmi{\varphi}{i}$ remain \tla{\tlc{\md}} formulas since it~is clear that $\{\veps\} \in \tlc{\md}$. For every $n \leq k$, we let,
	\[
	\lambda_n= \Biggl(\bigwedge_{0 \leq i < n} \nexi{a}{i}\Biggr) \wedge \nexi{b}{n} \wedge   \Biggl(\bigwedge_{n < i \leq k} \nexi{a}{i}\Biggr).
	\]
	Given a word $x \in A^*$ and $j \in \pos{x}$, we have $x,j \models \lambda_n$ if and only the word $a^n b a^{k-n} \in A^{k+1}$ is a prefix of the suffix starting at position $j$ (included) inside $x$. One can now verify that $\beta^{}_k(A_k^* )$ is defined by the following formula $\Lambda \in \tla{\tlc{\md}}$:
	\[
	\Lambda = \left(\finallyp{(A^{k+1})^*}{max}\right) \wedge \neg \finally{\Biggl(\Gamma \wedge \neg max \wedge \neg \bigvee_{0 \leq n \leq k} \lambda_n\Biggr)}.
	\]
	We now turn to the general case. Since $\tlhp{2}{\stzer} = \tlc{\at}$ by Proposition~\ref{prop:leveltwo}, it suffices to show that for every formula $\varphi \in \tla{\at}$ over the alphabet $A_k$, there exists a formula of \tla{\tlc{\md}} over the alphabet $A$ defining $\beta_k(L(\varphi))$. We start with the preliminary terminology that we need to specify the construction. Observe that for every letter $\ell \in A_k$, the image $\beta_k(\ell) \in A^*$ has length $k +1$.  Hence, given a word $w \in A_k^*$, we can define a natural mapping $\mu: \pos{w} \to \pos{\beta_k(w)}$. We let $\mu(0) = 0$ and $\mu(|w|+1) = |\beta_k(w)|+1 = (k+1) \times |w|+1$. Moreover, if $i \in \pos{w} \setminus \{0,|w|+1\}$, we let $\mu(i)=(i-1) \times (k+1)+1$ (which is the leftmost position in the infix of $\beta_k(w)$  that ``replaces'' the letter at position $i$ in $w$). In particular, we say that a position $j \in \pos{\beta_k(w)}$ is \emph{distinguished} if there exists $i\in\pos{w}$ such that $j=\mu(i)$. Observe that the formula $\Gamma\in \tla{\tlc{\md}}$ defined above can be used to pinpoint the distinguished positions:  given $w \in A_k^*$ and $j \in \pos{\beta(w)}$, we have $\beta(w),j \models \Gamma$ if and only if $j$ is distinguished.
	
	We turn to the construction. For each formula $\varphi \in \tla{\at}$ over $A_k$, we build a new one $\thecons{\varphi} \in \tla{\tlc{\md}}$ over $A$ satisfying the following property:
	\begin{equation} \label{eq:modcons}
		\begin{array}{c}
				w,i \models \varphi \Leftrightarrow \beta_k(w),\mu(i) \models \thecons{\varphi} \\
				 \text{for every $w \in A_k^*$ and $i \in \pos{w}$}.
		\end{array}
	\end{equation}
	Let us first explain why this completes the proof. Consider a language $L \in \tlc{\at}$. By hypothesis, $L$ is defined by a formula $\varphi \in \tla{\at}$. It now follows from~\eqref{eq:modcons} that $\beta_k(L)$ is the language $L(\thecons{\varphi}) \cap \beta^{}_k(A_k^*)$. Hence, $\beta_k(L)$ is defined by the formula $\thecons{\varphi} \wedge \Lambda \in \tla{\tlc{\md}}$.
	
	We turn to the construction itself.  We build \thecons{\varphi} using structural induction on $\varphi$. If $\varphi = min$, then $\thecons{\varphi} = min$. Similarly, if $\varphi = max$, then $\thecons{\varphi} = max$. If $\varphi = \ell$ for $\ell \in A_k$, then $\ell = \ell_i$ for some $i \leq k$.  It therefore suffices to define $\thecons{\varphi} = \lambda_i$ since we have $\beta_k(\ell_i) = a^i b a^{k-i}$.
	
	It remains to handle the temporal operators. First, assume that $\varphi = \finallyp{K}{\psi}$. By hypothesis on $\varphi$, we know that $K \in \at$. Hence, $K$ is a Boolean combination of languages $A_k^* \ell A_k^*$ where $\ell\in A_k$. For each $i \leq k$, we associate the following language $P_i \in \tlc{\md}$:
	\[
	P_i = (A^{k+1})^*A^i b A^* \text{ defined by $\finallyp{(A^{k+1})^*A^i}{b} \in \tla{\md}$}.
	\]
	We let $Q_K \in \tlc{\md}$ be the language over $A^*$ obtained by replacing every language $A_k^* \ell^{}_i A_k^*$ for $i \leq k+1$ in the Boolean combination defining $K$ by $P_i \in \tlc{\md}$. One can now verify that we can define,
	\[
	\thecons{\varphi} = \left(min \wedge  \finallyp{Q_K}{\left(\Gamma \wedge \thecons{\psi}\right)}\right) \vee \nexi{\left(\finallyp{Q_K}{\left(\Gamma \wedge \thecons{\psi}\right)\right)}}{k}.
	\]
	Similarly, if $\varphi = \finallymp{K}{\varphi'}$, then by hypothesis on $\varphi$, we know that $K \in \at$. Hence, we can associate a language $Q_K \in \tlc{\md}$ to $K$ as explained above. We now define,
	\[
	\thecons{\varphi} = \finallymp{Q_K}{\left(\nexmi{\left(\Gamma \wedge \thecons{\varphi}\right)}{k}\right) \vee \left(min \wedge\thecons{\varphi} \right)}.
	\]
	One can verify that this construction $\varphi \mapsto \thecons{\varphi}$ does satisfy~\eqref{eq:modcons}, which completes the proof.
\end{proof}

We now prove Lemma~\ref{lem:tlmodunc}. Let us first recall the statement.

\tlmodunc

\begin{proof}
	We fix $k$ for the argument.
	We use induction on $n$ to prove that $\beta_k(U_n)$ is not \polp{n+1}{\Gs}-separable from $\beta_k(V_n)$ for all $n\leq k$. Since $\beta_k(U_n) \cap \beta_k(V_n) = \emptyset$, this implies that $\beta_k(U_n) \not\in \polp{n+1}{\Gs}$ for all $n \leq k$ and the case $n = k$ yields the desired result. We start with the case $n = 0$. By definition, $\beta_k(U_0) = \{\veps\}$ and $\beta_k(V_0) = (ba^k)^+$. Since \Gs is a \emph{group} \vari, one can check that $(ba^k)^+$ is not \Gs-separable from $\{\veps\}$. Hence, Lemma~\ref{lem:polinduc} implies that $ \{\veps\}$ is not \pol{\Gs}-separable from $(ba^k)^+$ as desired. Assume now that $n \geq 1$. By induction hypothesis, we know that $\beta_k(U_{n-1})$ is not \polp{n}{\Gs}-separable from $\beta_k(V_{n-1})$. Hence, it follows from Lemma~\ref{lem:nsepconcat} that $\beta_k(\ell_nU_{n-1})$ is not \polp{n}{\Gs}-separable from $\beta_k(\ell_nV_{n-1})$. We can now apply Lemma~\ref{lem:polinduc} and Lemma~\ref{lem:polcopol},  which imply that $(\beta_k(\ell_nV_{n-1}))^*$ is not  \polp{n}{\Gs}-separable from $(\beta_k(\ell_nV_{n-1}))^* \beta_k(\ell_nU_{n-1}) (\beta_k(\ell_nV_{n-1}))^*$. This exactly says that  $\beta_k(U_n)$ is not \polp{n+1}{\Gs}-separable from $\beta_k(V_n)$, which completes the proof.
\end{proof}

\smallskip

It remains to prove the statements involved in the proof of Theorem~\ref{thm:gpualt}. Again, the proof is based on the languages $U_k,V_k$ which are encoded over a \emph{binary alphabet} $A = \{a,b\}$. For every $k \in \nat$, we defined $B_k =  A_k \cup \{b\}$. Moreover, we introduced two morphisms $\gamma_k: B_k^* \to A_k^*$ and $\delta_k: B_k^* \to A^*$ defined as follows:
\[
\begin{array}{llll}
	\gamma_k: & B_k^* & \to     & A_k^* \\
	& b                  & \mapsto & \veps \\
	& \ell_i & \mapsto & \ell_i 
\end{array} \quad \text{and}\quad  \begin{array}{llll}
	\delta_k: & B_k^* & \to     & A^* \\
	& b                  & \mapsto & b \\
	& \ell_i & \mapsto & ba^{i+1} 
\end{array}
\]
We may now prove Proposition~\ref{prop:ddproj}. It shows that this encoding maps $\tlhp2{\stzer}$ to $\tlhp{2}{\dotzer}$. Combined with Lemma~\ref{lem:langtl2st}, it yields $\delta^{}_k(\gamma_k\inv(U_k)) \in \tlhp{2}{\dotzer}$ for all $k \in \nat$.

\ddproj

\begin{proof}
	We first show that $\delta^{}_k(\gamma_k\inv(A_k^*)) \in \tlhp{2}{\dotzer}$. One can verify that,
	\[
	\delta^{}_k(\gamma_k\inv(A_k^*))  = A^* \setminus \left(aA^* \cup A^*a^{k+2}A^*\right).
	\]
	Clearly, it belongs to $\bpol{\dotzer} \subseteq \tlc{\dotzer} \subseteq \tlhp{2}{\dotzer}$. We now turn to the general case. Since $\tlhp{2}{\stzer} = \tlc{\at}$ by Proposition~\ref{prop:leveltwo}, it suffices to show that for every formula $\varphi \in \tla{\at}$ over the alphabet $A_k$, we have to construct a formula of \tla{\tlc{\dotzer}} over the alphabet $A$ defining $\delta^{}_k(\gamma_k\inv (L(\varphi)))$. We start with some preliminary terminology that we require to specify the construction.
	
	Consider a word $x \in A^*$. We say that a position $j \in \pos{x}$ is \emph{distinguished} if and only if either we have $j \in \{0,|x|+1\}$ or we have $\wpos{x}{j} = b$ and $\wpos{x}{j+1} = a$. Consider the following formula $\Gamma \in \tla{\tlc{\dotzer}}$:
	\[
	\Gamma = min \vee max \vee b \wedge \finallyp{\{\veps\}}{a}.
	\]
	Clearly, given a word $x \in A^*$ and $j \in \pos{x}$, we have $x,j \models \Gamma$ if and only if $j$ is distinguished. Now, consider a word $w \in A_k^*$ and let $x \in \delta^{}_k(\gamma_k\inv(w))$. By definition $x$ is built from $w$ by adding arbitrarily many extra letters ``$b$'' and replacing ever letter $\ell_h \in A_k$ for $h \leq k$ by the word $ba^{h+1}$. Hence, there exists a natural mapping $\mu: \pos{w} \to \pos{x}$ sending every position of $w$ to a \emph{distinguished} position of $x$. We let $\mu(0) = 0$, $\mu(|w|+1) = |x|+1$ and, for $1 < i < k$, we let $\mu(i)$ as the leftmost position in the unique infix $ba^{h+1}$ of $x$ corresponding to the letter $\ell_h$ that labels position $i$ in $w$.
	
	We are ready to describe the construction. For every  formula $\varphi \in \tla{\at}$ over $A_k$, we construct a new formula $\thecons{\varphi} \in \tla{\tlc{\dotzer}}$ over  $A$ satisfying the following property:
	\begin{equation} \label{eq:ddcons}
		\begin{array}{c}
			w,i \models \varphi \Leftrightarrow x,\mu(i) \models \thecons{\varphi} \\
			 \text{for every $w \in A_k^*$, $i \in \pos{w}$ and $x \in \delta^{}_k(\gamma_k\inv(w))$}.
		\end{array}
	\end{equation}
	We first complete the main proof. Let $L \in \tlhp{2}{\stzer}$. By Proposition~\ref{prop:leveltwo}, $L$ is defined by a formula $\varphi \in \tla{\at}$. It follows from~\eqref{eq:ddcons} that $L(\thecons{\varphi}) \cap \delta^{}_k(\gamma_k\inv(A_k^*)) = \delta^{}_k(\gamma_k\inv(L))$. Hence, $\delta^{}_k(\gamma_k\inv(L)) \in \tlhp{2}{\dotzer}$ since we already proved above that  $\delta^{}_k(\gamma_k\inv(A_k^*)) \in  \tlhp{2}{\dotzer}$.
	
	We turn to the construction itself.  We build \thecons{\varphi} using structural induction on $\varphi$. If $\varphi = min$, then $\thecons{\varphi} = min$. Similarly, if $\varphi = max$, then $\thecons{\varphi} = max$. If $\varphi = \ell$ for $\ell \in A_k$, then $\ell = \ell_h$ for some $h \leq k$.  Hence, $\delta_k(\ell_h) = ba^{h+1}$. For all $n \in \nat$, we write $\nexi{a}{n}$ for the following formula: when $n = 0$, $\nexi{a}{0} = a$ and when $n \geq 1$, $\nexi{a}{n} = \finallyp{\{\veps\}}{\left(\nexi{a}{n-1}\right)}$. Observe that $\nexi{a}{n}$ is \tla{\tlc{\dotzer}} since$ \{\veps\} \in \dotzer$. Recall that $\varphi = \ell_h$. We let,
	\[
	\thecons{\varphi} = b \wedge \Biggl(\bigwedge_{1 \leq n \leq h+1}\nexi{a}{n}\Biggr) \wedge \neg \nexi{a}{h+2}.
	\]
	It remains to handle the temporal operators. First, assume that $\varphi = \finallyp{K}{\psi}$. By hypothesis on $\varphi$, we know that $K \in \at$. Hence, $K$ is a Boolean combination of languages $A_k^* \ell A_k^*$ where $\ell\in A_k$. For each $h \leq k$, we associate the following language $P_h \in \bpol{\dotzer} \subseteq \tlc{\dotzer}$ to the letter $\ell_h \in A_k$:
	\[
	P_h = A^* ba^{h+1}bA^* \cup A^* ba^{h+1}.
	\]
	We let $Q_K \in \tlc{\dotzer}$ be the language over $A^*$ obtained by replacing every language $A_k^* \ell_h A_k^*$ for $h \leq k$ in the Boolean combination combination defining $K$ by $P_h \in \tlc{\dotzer}$. One can now verify that we can define,
	\[
	\thecons{\varphi} = \finallyp{Q_K}{\left(\Gamma \wedge \thecons{\psi}\right)}.
	\]
	Similarly, if $\varphi = \finallymp{K}{\psi}$, then by hypothesis on $\varphi$, we know that $K \in \at$. Hence, we can associate a language $Q_K \in \tlc{\dotzer}$ to $K$ as explained above. We now define,
	\[
	\thecons{\varphi} = \finallymp{Q_K}{\left(\Gamma \wedge \thecons{\varphi}\right)}.
	\]
	One can verify that this construction $\varphi \mapsto \thecons{\varphi}$ does satisfy~\eqref{eq:ddcons}, which completes the proof.
\end{proof}

It remains to prove Lemma~\ref{lem:tlddunc}. We first recall its statement.

\tlddunc

\begin{proof}
	For every $k \in \nat$, we let $X_k = \delta^{}_k(\gamma_k\inv(U_k))$ and $Y_k = \delta^{}_k(\gamma_k\inv(V_k))$. We use induction on $k$ to prove that $X_k$ is not \polp{k+1}{\Gs^+}-separable from $Y_k$ for every $k \in \nat$. Since $X_k \cap Y_k = \emptyset$, this implies that $X_k \not\in \polp{k+1}{\Gs^+}$ as desired. We start with the case $k = 0$. By definition $X_0 = b^*$ and $Y_0 = (b^+a)^+b^*$. Since \Gs is a \emph{group} \vari, one can verify that $(b^+a)^+$ is not $\Gs^+$-separable from $b^*$. Hence, Lemma~\ref{lem:polinduc} implies that $X_0 = b^*$ is not \pol{\Gs^+}-separable from $Y_0 = b^* (b^+a)^+b^*$. Assume now that $k \geq 1$. By induction hypothesis, we know that $X_{k-1}$ is not \polp{k}{\Gs^+}-separable from $Y_{k-1}$. Hence, it follows from Lemma~\ref{lem:nsepconcat} that $ab^{k+1}X_{k-1}$ is not \polp{k}{\Gs^+}-separable from $ab^{k+1}Y_{k-1}$. This implies that $Y_{k-1}$ is not \cocl{\polp{k}{\Gs^+}}-separable from $X_{k-1}$. We may now apply Lemma~\ref{lem:polinduc} and Lemma~\ref{lem:polcopol}  to obtain that $(b^+a^{k+1}Y_{k-1})^*$ is not  \polp{k+1}{\Gs^+}-separable from $(b^+a^{k+1}Y_{k-1})^* b^+a^{k+1}X_{k-1} (b^+a^{k+1}Y_{k-1})^*$. This exactly says that  $X_k$ is not \polp{k+1}{\Gs^+}-separable from $Y_k$, which completes the proof.
\end{proof}

\section{Missing proofs in Section~\ref{sec:ratms}}
\label{app:ratms}
We provide the missing proofs for Section~\ref{sec:ratms}. We first state a useful lemma. When \Cs is a \vari, we are able to characterize the optimal \imprints using \Cs-morphisms. Given a morphism $\eta: A^* \to N$ and a word $w \in A^*$, we write $\etyp{w} \subseteq A^*$ for the set of all words $w' \in A^*$ such that~$\eta(w) = \eta(w')$.

\begin{restatable}{lemma}{optdmorph} \label{lem:optdmorph}
  Let \Cs be a \vari, $\rho: 2^{A^*} \to R$ be a \mratm and $L\subseteq A^*$. For every $r \in R$, the two following properties are equivalent:
  \begin{enumerate}
	\item We have $r \in \copti{L,\rho}$.
	\item For every \Cs-morphism $\eta: A^* \to N$, there exists $w \in L$ such that $r \leq \rho(\etyp{w})$.
  \end{enumerate}
\end{restatable}

\begin{proof}
  We fix $r \in R$ for the proof. Assume first that $r \in \copti{L,\rho}$ and let $\eta: A^* \to N$ be a \Cs-morphism. Consider the finite set of languages $\Kb = \{\etyp{w} \mid w \in L\}$. By definition, \Kb is a \Cs-cover of $L$ since $\eta$ is a \Cs-morphism.  Hence, $\copti{L,\rho} \subseteq\prin{\rho}{\Kb}$ which yields $r \in \prin{\rho}{\Kb}$. By definition of \Kb, we get $w \in L$ such that $r \leq \rho(\etyp{w})$, as desired.

  Conversely, assume that the second assertion in the lemma holds. We prove that $r \in \copti{L,\rho}$. Let \Kb be an optimal \Cs-cover of $L$ for $\rho$. By definition, \Kb is a finite set of languages in \Cs. Hence, since \Cs is a \vari, Proposition~\ref{prop:genocm} yields a \Cs-morphism $\eta: A^* \to N$ such that every $K \in \Kb$ is recognized by $\eta$. The second assertion in the lemma yields $w \in L$ and $r \leq \rho(\etyp{w})$. Since $w \in L$ and \Kb is a cover of $L$, we get $K \in \Kb$ such that $w \in K$. Consequently, since $K$ is recognized by $\eta$, we have $\etyp{w} \subseteq K$. Finally, since $r \leq \rho(\etyp{w})$, we get $r \leq \rho(K)$, which implies that $r \in \prin{\rho}{\Kb}$. Since \Kb is an optimal \Cs-cover of $L$ for $\rho$, we obtain $r \in \copti{L,\rho}$, which completes the proof.
\end{proof}

We now prove Proposition~\ref{prop:covreduc}.

\covreduc*

\begin{proof}
	We prove that  \copti{L,\rho} is equal to the following set:
  \[
    \dclosr \left\{\sum_{q \in Q} q \mid \begin{array}{c}
    	\text{$Q \subseteq R$ such that } \left(L,\left\{\rho_*\inv(q) \mid q \in Q\right\}\right)\\
     \text{ is not \Cs-coverable}
    \end{array}\right\}.
  \]
  Clearly, this yields an algorithm for computing \copti{L,\rho}  using \Cs-covering as desired.

  We first prove the left to right inclusion. Let $r \in \copti{L,\rho}$. We exhibit a subset $Q \subseteq R$ such that $r \leq \sum_{q \in Q} q$ and $(L,\{\rho_*\inv(q) \mid q \in Q\})$ is not \Cs-coverable. Let $\tau: 2^{A^*} \to 2^R$ be the map defined by $\tau(K) = \{\rho(w)  \mid w \in K\}$. One can verify that $\tau$ is a \full \mratm. Let $\Kb_\tau$ be an optimal \Cs-cover of $L$ for $\tau$. Since $r \in \copti{L,\rho}$, we have $r \in \prin{\rho}{\Kb_\tau}$ and we get $K \in \Kb_\tau$ such that $r \leq \rho(K)$. Let $Q = \tau(K) \subseteq R$. Since $\rho$ is \full, one can verify that $\rho(K) = \sum_{q \in Q} q$. Thus, $r\leq \sum_{q \in Q} q$ and it remains to prove that $(L,\{\rho_*\inv(q) \mid q \in Q\})$ is not \Cs-coverable. We proceed by contradiction. Assume that there exists a \Cs-cover \Hb of $L$ which is separating for $\{\rho_*\inv(q) \mid q \in Q\}$. For every $H \in \Hb$, we know that there exists $q \in Q$ such that $H \cap \rho_*\inv(q) = \emptyset$. By definition of $\tau$, this implies that $Q \not\subseteq \tau(H)$ for every $H \in \Hb$. Consequently, $Q \not\in \prin{\tau}{\Hb}$ which yields $Q \not\in \copti{L,\tau}$. This is a contradiction since $Q = \tau(K)$ by definition and $K \in \Kb_\tau$ which is an \emph{optimal} \Cs-cover of $L$ for $\tau$.

  It remains to prove the right to left inclusion. Since \copti{L,\rho} is an \imprint, we have $\dclosr \copti{L,\rho}= \copti{L,\rho}$ by definition. Hence, it suffices to prove that for every $Q \subseteq R$ such that $(L,\{\rho_*\inv(q) \mid q \in Q\})$ is not \Cs-coverable, we have $\sum_{q \in Q} q \in \copti{L,\rho}$. Let $\Kb_\rho$ be an optimal \Cs-cover of $L$ for $\rho$. Since $(L,\{\rho_*\inv(q) \mid q \in Q\})$ is not \Cs-coverable, $\Kb_\rho$ cannot be separating for $\{\rho_*\inv(q) \mid q \in Q\}$ and we get $K \in \Kb_\rho$ such that $K \cap \rho_*\inv(q) \neq \emptyset$ for every $q \in Q$. It follows that $\sum_{q \in Q} q \leq \rho(K)$. We get $\sum_{q \in Q} q \in \prin{\rho}{\Kb_\rho} = \copti{L,\rho}$. This completes the~proof.
\end{proof}
 
\section{Missing proofs in Section~\ref{sec:carav}}
\label{app:carav}
This appendix is devoted to the missing arguments for the intermediary statements in the proof of Theorem~\ref{thm:main}. Recall that a \fratm $\rho: 2^{A^*} \to R$ is fixed. We start with the proof of Proposition~\ref{prop:sound}. Note that in its proof, we use the well-known algebraic characterization of \tlxs by Thérien and Wilke~\cite{twfo2} as a subresult. We first recall its statement.

\sound*

\begin{proof}
  As \tlc{\at} is a \vari by Proposition~\ref{prop:tlvar}, it is immediate from~\cite[Lemma~9.11]{pzcovering2} that the set \ptlatopti contains the trivial elements and is closed under both downset and multiplication (these are generic properties which hold for all \varis). Hence, we concentrate on \tlc{\at}-operation. Let $B \subseteq A$ and $Q \subseteq \ptlatopti(B) =\opti{\tlc{\at}}{\etaat\inv(B),\rho}$. We show that  $\opti{\tlxs}{Q^+,\auxrq} \subseteq \opti{\tlc{\at}}{\etaat\inv(B),\rho}$.

  We fix $r \in \opti{\tlxs}{Q^+,\auxrq}$ for the proof. We have to show that $r \in \opti{\tlc{\at}}{\etaat\inv(B),\rho}$. We use Lemma~\ref{lem:optdmorph}. Given a fixed \tlc{\at}-morphism $\alpha: A^* \to M$, we have to exhibit a word $w \in \etaat\inv(B)$ such that $r \leq \rho(\atyp{w})$. By hypothesis, $Q \subseteq \opti{\tlc{\at}}{\etaat\inv(B),\rho}$. Hence, Lemma~\ref{lem:optdmorph} implies that for every $q \in Q$, there exists $v_q \in \etaat\inv(B)$ such that $q \leq \rho(\atyp{v_q})$. We use these words to define a morphism $\beta: Q^* \to M$. For each $q \in Q$ (viewed as a letter), we define $\beta(q) = \alpha(v_q)$.

  We first prove that $\beta$ is a \tlxs-morphism using the algebraic characterization of \tlxs by Thérien and Wilke\footnote{Strictly speaking, Thérien and Wilke only consider the logic \fodws, the counterpart of \tlxs in two-variable first-order logic. The equality $\tlxs = \fodws$ is established in~\cite{evwutl}.}~\cite{twfo2}.  We define $T = \beta(Q^+) \subseteq M$. For every idempotent $e \in E(T)$ and $s,t \in T$, we prove that $(esete)^\omega = (esete)^\omega ete(esete)^\omega$. By the results of~\cite{twfo2}, this implies that $\beta$ is a \tlxs-morphism. Since $e,s,t \in T = \beta(Q^+)$, one can check that for every element $z \in \{e,s,t\}$ there exists $u_z \in A^*$ such that $z = \alpha(u_z)$ and $\etaat(u_z) = B$.  In particular, $\etaat(u_s) = \etaat(u_t) = \etaat(u_e)$ which implies that $\alpha\inv(e)$ is not \at-separable from either $\alpha\inv(e)$ or $\alpha\inv(e)$. In other words, the pairs $(e,s)$ and $(e,t)$ are \at-pairs for $\alpha$. Since $\alpha$ is an \at-morphism Proposition~\ref{prop:tlcmor} yields $(esete)^\omega = (esete)^\omega ete(esete)^\omega$ as desired. Hence, we conclude that $\beta$ is a \tlxs-morphism.

  Now, since $r \in \opti{\tlxs}{Q^+,\auxrq}$ and $\beta$ is \tlxs-morphism, Lemma~\ref{lem:optdmorph} yields a word $x \in Q^+$ such that $r \leq \auxrq(\btyp{x})$. We now use $x \in Q^+$ to construct a word $w \in \etaat\inv(B)$ such that  $\auxrq(\btyp{x}) \leq  \rho(\atyp{w})$. This will imply that $r \leq \rho(\atyp{w})$ and complete the proof.

  Let $n = |x|$. Moreover, let $q_1,\dots,q_n \in Q$ be the letters such that $x = (q_1) \cdots (q_n)$. We define $w =v_{q_1} \cdots v_{q_n}$. Since $v_q \in \etaat\inv(B)$ for every $q \in Q$, it is immediate that $w \in \etaat\inv(B)$. It remains to prove that $\auxrq(\btyp{x}) \leq \rho(\atyp{w})$. Since $\auxrq$ is \full, it suffices to prove that $\auxrq(y) \leq \rho(\atyp{w})$ for all $y\in\btyp{x}$. Let $m = |y|$ and $p_1,\dots,p_m \in Q$ the letters such that $y = (p_1) \cdots (p_m)$. Also, let $u = v_{p_1} \cdots v_{p_m}$. By definition, $\beta(q) = \alpha(v_q)$ for all $q \in Q$. Thus, 
  \[
  \alpha(u) = \beta((p_1) \cdots (p_m)) = \beta(y).
  \]
  Moreover, $\beta(y) = \beta(x)$ and,
  \[
  \beta(x) = \beta((q_1) \cdots (q_n)) = \alpha(w).
  \]
  Altogether, we get $\alpha(u) = \alpha(w)$ which yields $u \in\atyp{w}$. Since $u=v_{p_1}\cdots v_{p_m}$, this yields $\atyp{v_{p_1}}\cdots\atyp{v_{p_m}}\subseteq \atyp{w}$. We get $\rho(\atyp{v_{p_1}}\cdots\atyp{v_{p_m}}) \leq \rho(\atyp{w})$. Finally, we have $q \leq \rho(\atyp{v_q})$ for every $q \in Q$ by definition. Thus, we get $p_1 \cdots p_m  \leq \rho(\atyp{w})$. Since $y=(p_1)\cdots (p_m)$, this exactly says that $\auxrq(y) \leq \rho(\atyp{w})$ as desired.
\end{proof}

All statements that remain to be proved are subresults in the proof of Proposition~\ref{prop:comp}. Recall that a good triple $(\frA,\beta,\tau)$ is fixed: $\beta: \frA^* \to 2^A$ is a morphism and $\tau: 2^{\frA^*} \to R$ is a \fratm.  Moreover, $\beta\inv(\emptyset) = \{\veps\}$ and for every $\ell \in \frA$, we have $(\beta(\ell),\tau(\ell)) \in S$.

We start with the proof of Fact~\ref{fct:fctq}. Recall that a subalphabet \frAd is fixed. Moreover, we have a subset $Q \subseteq R$. We use the set $Q$ as an alphabet and a morphism $\delta: \frAd^* \to Q^*$ is defined. For each letter $\ell \in \frAd$, we have $\delta(\ell) = (\tau(\ell)) \in Q$.

\newcommand{\pangle}[1]{\ensuremath{\langle #1\rangle}\xspace}

\fctq*

\begin{proof}
	By hypothesis, we have a \tlxs formula $\varphi$ (over the alphabet $Q$) defining $H$. We modify $\varphi$ to build a new formula $\pangle{\varphi} \in \tla{\cbeta}$ (over the alphabet \frA). We replace every atomic formula $(q)$ for $q \in Q$ by $\bigvee_{\{\ell \in \frAd \mid \delta(\ell) = q\}} \ell$. We replace all modalities $\textup{F}$ and $\textup{P}$ by $\textup{F}_{\frA^*}$ and $\textup{P}_{\frA^*}$ (clearly, we have $\frA^* \in \cbeta$). Finally, we replace all modalities $\textup{X}$ and $\textup{Y}$ by $\textup{F}_{\{\veps\}}$ and $\textup{P}_{\{\veps\}}$ respectively. Note that we have $\{\veps\} \in \cbeta$ because \cbeta consists of all languages recognized by $\beta: \frA^* \to 2^A$ and $\beta\inv(\emptyset) = \{\veps\}$ by hypothesis (indeed, $(\frA,\beta,\tau)$ is a good triple).

	Since $\varphi$ defines $H$, one can verify that for all $w \in \frAd^*$, we have $w \in L(\pangle{\varphi}) \Leftrightarrow \delta(w) \in H$. In other words, $L(\pangle{\varphi})$ coincides with $\delta\inv(H) \subseteq \frAd^*$ over the words in $\frAd^*$. Hence, $\delta\inv(H)$ is defined by $\pangle{\varphi} \wedge \neg \finallyp{\frA^*}{\left(\neg \bigvee_{\ell \in \frAd} \ell \right)} \in \tla{\cbeta}$. This yields $\delta\inv(H) \in \tlc{\cbeta}$ as desired.
\end{proof}

We turn to the proof of Fact~\ref{fct:inducpar}. At this point, a new good triple $(\frB,\gamma,\zeta)$ has been built.

\inducpar*

\begin{proof}
	Given $w \in \frB^+$, we have to prove that $\gamma(w) \neq D$ and exhibit $w' \in \frA^+$ such that $\gamma(w) = \beta(w')$. Since $\frB = \frAnd \times \Ub_D$, we have $n \geq 1$, $\ell_1,\dots,\ell_n\in\frAnd$ and $U_1,\dots,U_n \in \Ub_D$ such that $w = (\ell_1,U_1) \cdots (\ell_n,U_n)$. For all $i \leq n$, we get $B_{U_i} \in 2^A$ such that $U_i \subseteq \beta\inv(B_{U_i})$. For each $i \leq n$, we fix a word $w_i \in U_i$ (we know that $w_i$ exists since $\Ub_D$ is a partition which means that the languages $U_i$ are nonempty). Clearly, we have $\beta(w_i) = B_{u_i}$. We define $w' = \ell_1w_1 \cdots \ell_nw_n \in \frA^+$. By definition of $\gamma$ it is immediate that $\beta(w') = \gamma(w)$.

	Moreover, since $\ell_1 \in \frAnd$, we know that $\beta(\ell_1) \neq D$ by definition of \frAnd. It is also clear that $\beta(\ell_1) \subseteq \beta(w')$. Hence, $\gamma(w) = \beta(w') \neq D$ since we chose $D$ as a \emph{minimal} alphabet among the set $\{\beta(\ell) \mid \ell \in \frA\}$ (\emph{i.e.}, there exists no  $\ell \in \frA$ such that $\beta(\ell) \subsetneq D$). This completes the proof.
\end{proof}

We turn to the proof of Fact~\ref{fct:constrain}. Recall that given a word $w\in(\frAnd\frAd^*)^*$, a position $i \in \pos{w}$ is  \emph{distinguished} if $\wpos{w}{i} \in \frAnd \cup \{min,max\}$.

\constrain*

\begin{proof}
	By definition, there exists no distinguished position $k \in \pos{w}$ such that $i < k < j$ if and only if $\infix{w}{i}{j} \in \frAd^*$. Thus, we have to prove that the latter conditions holds if and only if  $\beta(\infix{w}{i}{j}) \in \{D,\emptyset\}$.

	Assume first that $\infix{w}{i}{j} \in \frAd^*$. Since $\beta(\ell) = D$ for every $\ell \in \frAd$, it follows that either $\beta(\infix{w}{i}{j}) = \emptyset$ (if $\infix{w}{i}{j} = \veps$) or $\beta(\infix{w}{i}{j}) = D$ (if $\infix{w}{i}{j} \in \frAd^+$). Conversely, assume that $\beta(\infix{w}{i}{j}) \in \{D,\emptyset\}$. If $\beta(\infix{w}{i}{j}) = \emptyset$, then $\infix{w}{i}{j} = \veps$ since $\beta\inv(\emptyset) = \{\veps\}$ (this is because $(\frA,\beta,\tau)$ is a good triple). We now assume that $\beta(\infix{w}{i}{j}) = D$. We chose $D$ as a \emph{minimal} alphabet among the set $\{\beta(\ell) \mid \ell \in \frA\}$ (\emph{i.e.}, there exists no  $\ell \in \frA$ such that $\beta(\ell) \subsetneq D$). Hence, $\beta(\infix{w}{i}{j}) = D$ implies that $\beta(\infix{w}{i}{j}) \in \frAd^+$ which completes the proof.
\end{proof}

We now prove Lemma~\ref{lem:tlatbas}. Recall that we defined a mapping $\mu: (\frAnd\frAd^*)^*\to\frB^*$.

\newcommand{\crochi}[1]{\ensuremath{\langle #1\rangle_{min}}\xspace}
\newcommand{\crocha}[1]{\ensuremath{\langle #1\rangle_{max}}\xspace}

\tlatabs*

\begin{proof}
	We start with a preliminary construction. Consider a letter $(\ell,U) \in \frB = \frAnd \times \Ub_D$, we construct a formula $\psi_{\ell,U} \in \tlc{\cbeta}$ such that for every $w \in (\frAnd\frAd^*)^*$ and every distinguished position $i \in \pos{w}$, we have $w,i \models \psi_{\ell,U} \Leftrightarrow \mu(i) \text{ has label } (\ell,U)$. By hypothesis on $\Ub_D$, there exists a formula $\varphi \in \tlc{\cbeta}$ defining $U$. Using Fact~\ref{fct:constrain}, we modify $\varphi$ so that given $w \in (\frAnd\frAd^*)^*$, the evaluation of the resulting formula at a distinguished position $i$ is constrained to the next maximal infix in $\frAd^*$. More precisely, we use structural induction to build two formulas $\crochi{\varphi}$ and $\crocha{\varphi}$ such that given $w \in (\frAnd\frAd^*)^*$ and two distinguished positions $i,j \in \pos{w}$ such that $i < j$ and there is no distinguished positions in-between $i$ and $j$, the two following properties hold for every $k \in \pos{\infix{w}{i}{j}}$:
	\begin{itemize}
		\item If $k < |\infix{w}{i}{j}|+1$, then,
		\[
		w,i+k \models \crochi{\varphi} \Leftrightarrow \infix{w}{i}{j},k \models \varphi.
		\]
		\item If $1 < k$, then,
		\[
		w,i+k \models \crocha{\varphi} \Leftrightarrow \infix{w}{i}{j},k \models \varphi.
		\]
	\end{itemize}
	It will then suffice to define $\psi_{\ell,U} = \ell \wedge \crochi{\varphi}$. We only describe the construction. That it satisfies the above properties can be verified.  
	\begin{itemize}
		\item If $\varphi \in \frA \cup \{\top,\bot\}$, we let $\crochi{\varphi} = \crocha{\varphi} = \varphi$.
		\item If $\varphi = min$, we let $\crochi{\varphi} = \bigvee_{\ell \in \frAnd} \ell$ and $\crocha{\varphi} = \bot$.
		\item If $\varphi = max$, we let $\crochi{\varphi} = \bot$ and $\crocha{\varphi} = \bigvee_{\ell \in \frAnd} \ell$.
	\end{itemize}
	Assume now that $\varphi = \finallyl{\varphi'}$ for some $L \in \cbeta$. We let,
	\[
	\begin{array}{lll}
		\crochi{\varphi} & = &\finallyp{L \cap \beta\inv(\{D,\emptyset\})}{\crocha{\varphi'}} \\
		\crocha{\varphi} &= &\left(\bigvee_{\ell \in \frAd} \ell\right) \wedge \finallyp{L \cap \{D,\emptyset\}}{\crocha{\varphi'}}
	\end{array}
	\]
	Finally, assume that $\varphi = \finallymp{L}{\varphi'}$ for some $L \in \cbeta$. We let,
	\[
	\begin{array}{lll}
		\crochi{\varphi} & = &\left(\bigvee_{\ell \in \frAd} \ell\right) \wedge \finallyp{L \cap \beta\inv(\{D,\emptyset\})}{\crochi{\varphi'}} \\
		\crocha{\varphi} &= & \finallyp{L \cap \{D,\emptyset\}}{\crochi{\varphi'}}
	\end{array}
	\]
	This completes the construction.

	\medskip

	\newcommand{\cfinal}[1]{\ensuremath{\lfloor #1\rfloor}\xspace}

	We are ready to prove the lemma. Let $K \in \tlc{\cgamma}$. By hypothesis, there exists a formula $\Gamma\in\tla{\cgamma}$ defining $K$. Using induction on $\Gamma$, we define $\cfinal{\Gamma} \in \tla{\cbeta}$ such that for every $w \in (\frAnd\frAd^*)^*$ and distinguished position $i \in \pos{w}$, we have $w,i \models \cfinal{\Gamma} \Leftrightarrow \mu(w),\mu(i) \models \Gamma$. Since $(\frAnd\frAd^*)^* = \frAnd\frA^* \cup \{\veps\}$, it will follow that $\mu\inv(K) \subseteq \frA^*$ is defined by the formula,
	\[
	\finallyp{\{\veps\}}{\Biggl(max \vee \bigvee_{\ell \in \frAnd} \ell\Biggr)} \wedge \cfinal{\Gamma} \in \tla{\cbeta}
	\]
	This yields $\mu\inv(K) \in \tlc{\cbeta}$ which completes the proof (note that here $\{\veps\} \in \cbeta$ since $\{\veps\} = \beta\inv(\emptyset)$ because $(\frA,\beta,\tau)$ is a good triple).

	If $\Gamma \in \{min,max,\top,\bot\}$, we define $\cfinal{\Gamma} = \Gamma$. Assume now that $\Gamma = (\ell,U) \in \frB$. We let $\cfinal{\Gamma} = \psi_{\ell,U}$. It remains to handle the temporal modalities: there exists $L \in \cgamma$ such that $\Gamma =  \finallyl{\Gamma'}$ or $\Gamma =  \finallyml{\Gamma'}$. By definition $L \in \cgamma$ is recognized by $\gamma$: there exists $X \subseteq 2^A$ such that $L = \gamma\inv(X)$. We let $H = \beta\inv(X) \in \cbeta$. Abusing terminology, we also write ``$\frAd$'' and ``$\frAnd$'' for the formulas ``$\bigvee_{\ell \in \frAd} \ell$'' and ``$\bigvee_{\ell \in \frAnd} \ell$'' respectively. If $\Gamma = \finallyl{\Gamma'}$, we let $\cfinal{\finallyl{\Gamma'}}$ as the following formula:
	\[
	\begin{array}{ll}
		 &\left((\finallyp{\{\veps\}}{\frAnd}) \wedge \finallyp{H}{\cfinal{\Gamma'}}\right)\\
		 \vee& \finallyp{\beta\inv(\{D,\emptyset\})}{\left(\frAd \wedge (\finallyp{\{\veps\}}{\frAnd}) \wedge \finallyp{H}{\cfinal{\Gamma'}}\right)}.
	\end{array}
	\]
	Finally, when $\Gamma = \finallyml{\Gamma'}$, we let $	\cfinal{\finallyml{\Gamma'}}$ as the following formula:
		\[
	\begin{array}{ll}
	&	\finallymp{H}{\left((\finallyp{\{\veps\}}{\frAnd})\right)\\
		\wedge& \left(\left(\frAnd \wedge \cfinal{\Gamma'}\right) \vee \left(\frAd \wedge \finallyp{\beta\inv(\{D,\emptyset\})}{(\frAnd \wedge \cfinal{\Gamma'})}\right)\right)}.
	\end{array}
	\]
	
	One can check that for every word $w \in (\frAnd\frAd^*)^*$ and every distinguished position $i \in \pos{w}$, we have $w,i \models \cfinal{\Gamma} \Leftrightarrow \mu(w),\mu(i) \models \Gamma$ as desired. This completes the proof.
\end{proof}

Finally, we prove Lemma~\ref{lem:theutlconcat}. Recall that an alphabet \frA is fixed together with two subalphabets \frAd and \frAnd such that $\frAd \cap \frAnd = \emptyset$. Moreover, the \vari \cbeta is defined over the alphabet \frA.

\theutlconc*

\begin{proof}
	First, observe that we may assume without loss of generality that $\veps \not\in V$. Indeed, if we have $\veps \in V$, then $UV = U \cup U(V \cap \frA^+)$. Thus, since we already know that $U\in\tlc{\cbeta}$, it suffices to show that $U(V \cap \frA^+)\in\tlc{\cbeta}$. Since it is clear that $(V \cap \frA^+)\in\tlc{\cbeta}$, this boils down to the special case when $\veps \not\in V$.

	We assume from now on that $\veps \not\in V$. Hence, we have $V \subseteq (\frAnd\frAd^*)^+$ (in particular, every word in $V$ begins by a letter in $\frAnd$). Since $U \subseteq \frAd^*$, $V \subseteq (\frAnd\frAd^*)^+$ and $\frAd \cap \frAnd = \emptyset$, observe that a word $w \in \frA^*$ belongs to $UV$ if and only if there exists $b \in \frAnd$ such that it satisfies the following three properties:
	\begin{enumerate}
		\item $w$ admits a decomposition $w = ubv$ such that the prefix $u$ belongs to $(\frA \setminus \{b\})^*$ (\emph{i.e.}, this decomposition is unique as it showcases the leftmost occurrence of a letter in $\frAnd$ inside $w$).
		\item The prefix $u$ of $w$ belongs to $U$.
		\item The suffix $v$ of $w$ belongs to $b\inv V$.
	\end{enumerate}
	Thus, it suffices to prove that these three properties can be expressed with \tla{\cbeta} formulas. We start by presenting a preliminary formulas that we shall need. It pinpoints the leftmost occurrence of the letter $b \in \frAnd$ inside a word:
	\[
	\psi_b = b \wedge \neg \finallym{b}.
	\]
	The first property above is now expressed by the formula $\finally{\psi_b}$. It remains to prove that the other two properties can be defined as well. Let us start with the second one. By hypothesis, $U \in \tlc{\cbeta}$ is defined by a formula $\varphi_U \in \tla{\cbeta}$. We modify it to construct a formula $\varphi'_U$. It is obtained from $\varphi_U$ by restricting its evaluation to positions that are smaller than the leftmost one labeled by a ``$b$''. More precisely, we build $\varphi'_U$ by applying the following modifications to $\varphi_U$:
	\begin{enumerate}
		\item We recursively replace every subformula of the form ``$\finallyl{\zeta}$'' by ``$\finallyl{(\zeta \wedge (\psi_b \vee \finally{\psi_b}))}$''.
		\item We replace every atomic subformula of the form ``$max$'' by $\psi_b$.
		\item We replace every atomic subformula of the form ``$a$'' for $a \in \frA$ by $a \wedge \neg \psi_b$.
	\end{enumerate}
	We are ready to construct a formula for the third property. Since $V \in \tlc{\cbeta}$, it follows from Proposition~\ref{prop:tlvar} that $b\inv V \in \tlc{\cbeta}$. Hence, we get a formula $\varphi_{b,V} \in \tla{\cbeta}$ defining $b\inv V$. We now construct a formula $\varphi'_{b,V}$ from $\varphi_{b,V}$ by restricting its evaluation to positions that are larger than the leftmost one labeled by a ``$b$''. More precisely, we build $\varphi'_{b,V}$ by applying the following modifications to $\varphi_{b,V}$:
	\begin{enumerate}
		\item We recursively replace every subformula ``$\finallyml{\zeta}$'' by ``$\finallyml{(\zeta \wedge (\psi_b \vee \finallym{\psi_b}))}$''.
		\item We replace every atomic subformula of the form ``$min$'' by $\psi_b$.
		\item We replace every atomic subformula of the form ``$a$'' for $a \in \frA$ by $a \wedge \neg \psi_b$.
	\end{enumerate}
	The language $UV$ is now defined by the following formula:
	\[
	\bigvee_{b \in \frAnd} \left(\finally{\left(\psi_b \wedge \varphi'_{b,V}\right)} \wedge \varphi'_U\right).
	\]
	This concludes the proof of Lemma~\ref{lem:theutlconcat}.
\end{proof}

 \end{document}